\documentclass[11pt,a4paper]{article}

\usepackage{fontspec}
\usepackage[titletoc,title,header]{appendix}

\usepackage{amsmath}
\usepackage{amssymb}
\usepackage{amsthm}
\usepackage{mathtools}
\usepackage{xparse}    

\usepackage[margin=1in]{geometry}

\usepackage{graphicx}
\usepackage{xcolor}

\usepackage{hyperref}
\hypersetup{
  colorlinks=true,
  linkcolor=blue,
  citecolor=blue,
  urlcolor=blue,
  pdfauthor={Yoshitsugu Sekine},
  pdftitle={A Note on the Resolvent Algebra and Functional Integral Approach to the Free Bose Einstein Condensation}
}

\usepackage[]{natbib}
\theoremstyle{plain}
\newtheorem{thm}{Theorem}[section]
\AtEndEnvironment{thm}{\qed}
\newtheorem{lem}[thm]{Lemma}
\AtEndEnvironment{lem}{\qed}
\newtheorem{prop}[thm]{Proposition}
\AtEndEnvironment{prop}{\qed}
\newtheorem{cor}[thm]{Corollary}
\AtEndEnvironment{cor}{\qed}

\AtEndEnvironment{conj}{\qed}

\AtEndEnvironment{fact}{\qed}

\theoremstyle{definition}
\newtheorem{defn}[thm]{Definition}
\AtEndEnvironment{defn}{\qed}

\AtEndEnvironment{ex}{\qed}

\theoremstyle{remark}

\AtEndEnvironment{rem}{\qed}

\makeatletter
\providecommand*{\dashv}{\mathrel{\mathpalette\@Dashv\vDash}}
\newcommand*{\@dashv}[2]{\reflectbox{$\m@th#1#2$}}
\makeatother
\newcommand{\abs}[1]{\left| #1 \right|} 
\NewDocumentCommand{\weaknorm}{O{\dbk} m}{#1{#2}} 
\newcommand{\norm}[1]{\left\Vert #1 \right\Vert} 

\newcommand{\bkt}[2]{\left\langle #1,\,#2 \right\rangle} 
\newcommand{\rbkt}[2]{\left( #1,\,#2 \right)} 
\newcommand{\isomto}{\mathrel{\rightarrowtail\kern-1.9ex\twoheadrightarrow}} 
\newcommand{\cbk}[1]{\left\{ #1 \right\}} 
\newcommand{\dbk}[1]{\left\langle #1 \right\rangle} 
\newcommand{\pairbk}[1]{\rbk{#1}} 
\newcommand{\rbk}[1]{\left( #1 \right)} 
\newcommand{\sqbk}[1]{\left[ #1 \right]} 
\newcommand{\funcond}[3]{\fun{#1}{#2 \middle| #3}} 
\newcommand{\funrbk}[2]{\fun{\rbk{#1}}{#2}} 
\newcommand{\fun}[2]{#1 \rbk{#2}} 
\newcommand{\sqfuncond}[3]{\sqfun{#1}{#2 \middle| #3}} 
\newcommand{\sqfun}[2]{#1 \sqbk{#2}} 
\newcommand{\closedinterval}[2]{\sqbk{#1,\,#2}} 
\newcommand{\leftopeninterval}[2]{\left(#1, \, #2 \right]} 
\newcommand{\nonneginterval}{\rightopeninterval{0}{\infty}} 
\newcommand{\openinterval}[2]{\rbk{#1,\,#2}} 
\newcommand{\rightopeninterval}[2]{\left[#1, \, #2 \right)} 
\newcommand{\commutator}[2]{\sqbk{#1,\,#2}} 

\DeclareMathOperator{\sgn}{sgn} 
\DeclareMathOperator{\symdiff}{\raisebox{.2ex}{\scalebox{0.8}{$\triangle$}}} 
\NewDocumentCommand{\imunit}{O{\mathsf{i}}}{#1} 
\NewDocumentCommand{\placeholder}{O{\bullet}}{#1} 
\NewDocumentCommand{\trace}{O{\operatorname{Tr}}}{#1} 
\newcommand{\cmpconj}[1]{\overline{#1}} 
\newcommand{\dom}{\operatorname{dom}} 
\newcommand{\eqcsq}[1]{\sqbk{#1}} 
\newcommand{\fml}[2]{\cbk{#1}_{#2}} 
\newcommand{\id}{\mathrm{id}} 
\newcommand{\inv}[1]{#1^{-1}} 
\newcommand{\laplacian}{\Delta} 
\newcommand{\napiernum}{\mathsf{e}} 
\newcommand{\opchern}[1]{\operatorname{ch}} 
\newcommand{\opimag}{\operatorname{Im}} 
\newcommand{\opreal}{\operatorname{Re}} 
\newcommand{\pd}[2]{\frac{\partial #1}{\partial #2}} 
\newcommand{\setSymbolDownLeft}[2]{{\vphantom{#2}}_{#1}{#2}} 
\newcommand{\setSymbolUpLeft}[2]{{\vphantom{#2}}^{#1}{#2}} 
\newcommand{\supp}{\operatorname{supp}} 
\NewDocumentCommand{\agvariety}{O{\mathcal}}{#1} 
\NewDocumentCommand{\cmdrel}{O{\omega}}{#1} 
\NewDocumentCommand{\dfsp}{O{A}}{#1} 
\NewDocumentCommand{\eqcpointed}{O{\eqcsq} m}{#1{#2}_{\ast}} 
\NewDocumentCommand{\fnheaviside}{O{H}}{#1} 
\NewDocumentCommand{\fthol}{O{\mathcal{O}}}{#1} 
\NewDocumentCommand{\ftmero}{O{\mathcal{M}}}{#1} 
\NewDocumentCommand{\grcentralizer}{O{Z}}{#1} 
\NewDocumentCommand{\grmetform}{O{2} m m}{\grmet[#1] \! \rbkt{#2}{#3}} 
\NewDocumentCommand{\grmet}{O{2}}{\setSymbolDownLeft{#1}{g}} 
\NewDocumentCommand{\grnormalizer}{O{N}}{#1} 
\NewDocumentCommand{\gropasym}{O{A}}{#1} 
\NewDocumentCommand{\gropsym}{O{S}}{#1} 
\NewDocumentCommand{\grpermorderedpair}{O{\mathcal{P}}}{#1} 
\NewDocumentCommand{\grsym}{O{\mathfrak{S}} m}{#1_{#2}} 
\NewDocumentCommand{\gtbase}{O{\mathcal}}{#1} 
\NewDocumentCommand{\gtfilter}{O{\mathcal}}{#1} 
\NewDocumentCommand{\gtfmlclosed}{O{\mathcal}}{#1} 
\NewDocumentCommand{\gtfmlopen}{O{\mathcal}}{#1} 
\NewDocumentCommand{\gtopenball}{O{U}}{#1} 
\NewDocumentCommand{\gtopencover}{O{\mathcal}}{#1} 
\NewDocumentCommand{\gtopennbh}{O{\mathcal}}{#1} 
\NewDocumentCommand{\gtpreopencover}{O{\mathcal}}{#1} 
\NewDocumentCommand{\gtsubbase}{O{\mathcal}}{#1} 
\NewDocumentCommand{\gtvicinity}{O{\mathcal}}{#1} 
\NewDocumentCommand{\lasp}{O{\mathcal}}{#1} 
\NewDocumentCommand{\latprightrbk}{O{\top} m}{\rbk{#2}^{#1}} 
\NewDocumentCommand{\latpright}{O{\top} m}{#2^{#1}} 
\NewDocumentCommand{\latp}{O{t} m}{\setSymbolUpLeft{#1}{#2}} 
\NewDocumentCommand{\lpdistribution}{O{\mu} m}{#2_{\ast,#1}} 
\NewDocumentCommand{\lpmollifier}{O{\rho}}{#1} 
\NewDocumentCommand{\lpofpositive}{O{\chi}}{#1} 
\NewDocumentCommand{\manliederiv}{O{L}}{#1} 
\NewDocumentCommand{\mansmoothnbh}{O{\mathcal}}{#1} 
\NewDocumentCommand{\mblfmldsysgenerated}{O{d} m}{\fun{#1}{#2}} 
\NewDocumentCommand{\mblfmlgenerated}{O{\sigma} m}{\fun{#1}{#2}} 
\NewDocumentCommand{\oacorrfn}{O{\Gamma}}{#1} 
\NewDocumentCommand{\oagnsvector}{O{\Omega}}{#1} 
\NewDocumentCommand{\oaideal}{O{\mathcal}}{#1} 
\NewDocumentCommand{\oanumberoperator}{O{A}}{#1} 
\NewDocumentCommand{\oaposcone}{O{\mathcal{P}}}{#1} 
\NewDocumentCommand{\oapressure}{O{P}}{#1} 
\NewDocumentCommand{\oarepn}{O{\pi}}{#1} 
\NewDocumentCommand{\oaspkmsstate}{O{\txtkms}}{#1} 
\NewDocumentCommand{\oaspnormalstate}{O{N}}{#1} 
\NewDocumentCommand{\oasppurestate}{O{P}}{#1} 
\NewDocumentCommand{\oaspstate}{O{E}}{#1} 
\NewDocumentCommand{\oastatevector}{O{\Omega}}{#1} 
\NewDocumentCommand{\oastate}{O{\omega}}{#1} 
\NewDocumentCommand{\opdilation}{O{\delta}}{#1} 
\NewDocumentCommand{\opdmat}{O{\rho}}{#1} 
\NewDocumentCommand{\opfockan}{O{a}}{#1} 
\NewDocumentCommand{\opfockcran}{O{a}}{#1^{\#}} 
\NewDocumentCommand{\opfockcrdagger}{O{a}}{#1^{\dagger}} 
\NewDocumentCommand{\opfockcr}{O{a}}{#1^{\ast}} 
\NewDocumentCommand{\opfocknumber}{O{N}}{#1} 
\NewDocumentCommand{\opfocksegalconj}{O{\pi}}{#1} 
\NewDocumentCommand{\opfocksegal}{O{\phi}}{#1} 
\NewDocumentCommand{\opfockvac}{O{\Omega}}{#1} 
\NewDocumentCommand{\opspecmeas}{O{E}}{#1} 
\NewDocumentCommand{\opspec}{O{} m}{\fun{\sigma_{#1}}{#2}} 
\NewDocumentCommand{\optransl}{O{\tau}}{#1} 
\NewDocumentCommand{\opvarspec}{O{} m}{\operatorname{spec}_{#1} #2} 
\NewDocumentCommand{\physaction}{O{\mathcal{A}}}{#1} 
\NewDocumentCommand{\physcharge}{O{e}}{\mathrm{#1}} 
\NewDocumentCommand{\physcplconst}{O{\mathsf{g}}}{#1} 
\NewDocumentCommand{\physelectrostaticcapasity}{O{\mathrm{Cap}}}{#1} 
\NewDocumentCommand{\physenergy}{O{E}}{#1} 
\NewDocumentCommand{\physgse}{O{E}}{#1_{0}} 
\NewDocumentCommand{\physham}{O{H}}{#1} 
\NewDocumentCommand{\physlagdensity}{O{\mathcal{L}}}{#1} 
\NewDocumentCommand{\physlag}{O{L}}{#1} 
\NewDocumentCommand{\physliouvilean}{O{L}}{#1} 
\NewDocumentCommand{\physmass}{O{m}}{#1} 
\NewDocumentCommand{\prbbrownmv}{O{B}}{#1} 
\NewDocumentCommand{\prbcharfun}{O{\chi}}{#1} 
\NewDocumentCommand{\prbdist}{O{\mathcal{P}}}{#1} 
\NewDocumentCommand{\prbgaussianmeasure}{O{\msrcal{N}}}{#1} 
\NewDocumentCommand{\prbnormaldist}{O{N}}{#1} 
\NewDocumentCommand{\prbpoissonprocess}{O{N}}{#1} 
\NewDocumentCommand{\prbprocess}{O{X}}{#1} 
\NewDocumentCommand{\prbqspace}{O{\mathcal{Q}}}{#1} 
\NewDocumentCommand{\prbspsample}{O{\Omega}}{#1} 
\NewDocumentCommand{\prbspstate}{O{\mathsf{S}}}{#1} 
\NewDocumentCommand{\psh}{O{\mathfrak}}{#1} 
\NewDocumentCommand{\qtquantumchannel}{O{\mathcal{L}}}{#1} 
\NewDocumentCommand{\repn}{O{\pi}}{#1} 
\NewDocumentCommand{\schattencls}{O{\mathbb{K}}}{#1} 
\NewDocumentCommand{\setfmlcylinder}{O{\mathcal{C}}}{#1} 
\NewDocumentCommand{\setfml}{O{\mathcal}}{#1} 
\NewDocumentCommand{\setindex}{O{\mathcal} m}{#1{#2}} 
\NewDocumentCommand{\setlattice}{O{\Gamma}}{#1} 
\NewDocumentCommand{\setspecial}{O{\mathcal} m}{#1{#2}} 
\NewDocumentCommand{\shdiffform}{O{\sheaf{A}}}{#1} 
\NewDocumentCommand{\sheaf}{O{\mathfrak}}{#1} 
\NewDocumentCommand{\smchemicalpotential}{O{\mu}}{#1} 
\NewDocumentCommand{\smenergydensity}{O{\varrho}}{#1} 
\NewDocumentCommand{\smfluctuationwithdmat}{O{\beta} m}{\smuncertaintywithdmat[#1]{#2}^2} 
\NewDocumentCommand{\sminvtemperature}{O{\beta}}{#1} 
\NewDocumentCommand{\smlocaldensityoperator}{O{\rho}}{#1} 
\NewDocumentCommand{\smmicrocanonicalstate}{O{\beta} m}{\physmean{#2}_{#1}} 
\NewDocumentCommand{\smnumberdensity}{O{\rho}}{#1} 
\NewDocumentCommand{\smooth}{O{\mathcal{E}}}{#1} 
\NewDocumentCommand{\smparticlenumber}{O{N}}{#1} 
\NewDocumentCommand{\smpressure}{O{p}}{#1} 
\NewDocumentCommand{\smspecificfreeenergy}{O{\bar{f}}}{#1} 
\NewDocumentCommand{\smthermalvac}{O{\beta}}{\Omega_{#1}} 
\NewDocumentCommand{\smuncertaintywithdmat}{O{\beta} m}{\rbk{\triangle #2}_{#1}} 
\NewDocumentCommand{\sphilbfrak}{O{\mathfrak}}{#1} 
\NewDocumentCommand{\sphilb}{O{\mathcal}}{#1} 
\NewDocumentCommand{\splowerhalf}{O{\mathbb{H}}}{#1_{\txtneg}} 
\NewDocumentCommand{\spupperhalf}{O{\mathbb{H}}}{#1_{\txtnonneg}} 
\NewDocumentCommand{\topmetric}{O{d}}{#1} 
\NewDocumentCommand{\vaoutnormal}{O{\widehat}}{#1} 
\newcommand{\category}[1]{\mathop{\mathsf{#1}}} 

\newcommand{\catpresheaf}[1]{\category{PSh}} 
\newcommand{\conti}{C} 
\newcommand{\faadj}[1]{#1^{\ast}} 
\newcommand{\fafouriertr}{\mathcal{F}} 
\newcommand{\faftr}[1]{\widehat{#1}} 
\newcommand{\fldcmp}{\fld{C}} 
\newcommand{\fldrat}{\fld{Q}} 
\newcommand{\fldreal}{\fld{R}} 
\newcommand{\fld}[1]{\mathbb{#1}} 
\newcommand{\fndef}[1]{\boldsymbol{1}_{#1}} 
\newcommand{\fnexp}[1]{\fun{\exp}{#1}} 
\newcommand{\gtclos}[1]{\overline{#1}} 
\newcommand{\lpseq}{\ell} 
\newcommand{\lp}{L} 
\newcommand{\mblfmlborel}{\mblfml{B}} 
\newcommand{\mblfmlfrak}[1]{\mathfrak{#1}} 
\newcommand{\mblfml}[1]{\mathcal{#1}} 
\newcommand{\monnat}{\mathbb{N}} 
\newcommand{\msrcal}[1]{\mathcal{#1}} 
\newcommand{\msrdirac}{\delta} 
\newcommand{\msrprb}{\mathrm{Pr}} 

\newcommand{\oacstar}{C^{\ast}} 
\newcommand{\oadoublecommutant}[1]{#1^{\prime \prime}} 
\newcommand{\oaresolventalgebra}{\oa{R}} 
\newcommand{\oaresolvent}{R} 
\newcommand{\oa}[1]{\mathcal{#1}} 
\newcommand{\opdmsr}[1]{\mathop{d #1}} 
\newcommand{\opfocksndqntdiff}{d \Gamma} 
\newcommand{\opform}[1]{\mathsf{#1}} 
\newcommand{\opspecint}[1]{\mathcal{E}} 
\newcommand{\physmean}[1]{\dbk{#1}} 
\newcommand{\prbexp}{\mathbb{E}} 
\newcommand{\pushoutrbk}[1]{\rbk{#1}_{\ast}} 
\newcommand{\pushout}[1]{#1_{\ast}} 
\newcommand{\riemzeta}{\zeta} 
\newcommand{\ringratint}{\mathbb{Z}} 
\newcommand{\semigrposint}{\monnat_1} 

\newcommand{\seq}[2]{\if\relax\detokenize{#1}\relax \rbk{#1} \else \rbk{#1}_{#2} \fi} 
\newcommand{\setisomorphism}[1]{\operatorname{Iso}} 
\newcommand{\setone}[1]{\cbk{#1}}
\newcommand{\set}[2]{\left\{#1 \, \middle| \, #2\right\}}
\newcommand{\spfockgamma}{\Gamma} 
\newcommand{\txtarakiwoods}{\mathrm{AW},\mathrm{b}} 
\newcommand{\txtbec}{\mathrm{BEC}} 
\newcommand{\txtborel}{\mathrm{b}} 
\newcommand{\txtbsn}{\mathrm{b}} 
\newcommand{\txtcritical}{\mathrm{c}} 
\newcommand{\txtexternal}{\mathrm{ext}} 
\newcommand{\txtfr}{\mathrm{fr}} 
\newcommand{\txtgrandcanonical}{\mathrm{GC}} 
\newcommand{\txtkms}{\mathrm{KMS}} 
\newcommand{\txtleft}{\mathrm{l}} 
\newcommand{\txtloc}{\mathrm{loc}} 
\newcommand{\txtmix}{\mathrm{mix}} 
\newcommand{\txtneg}{\mathrm{-}} 
\newcommand{\txtnonneg}{\mathrm{+}} 
\newcommand{\txtnonzero}{\mathrm{nz}} 
\newcommand{\txtstd}{\mathrm{std}} 
\newcommand{\txttot}{\mathrm{tot}} 

\title{A Note on the Resolvent Algebra and Functional Integral Approach to the Free Bose Einstein Condensation}

\author{%
Yoshitsugu Sekine\\{\small\texttt{4429sekine@gmail.com}}%
}

\date{\today}

\begin{document}

\maketitle

\begin{abstract}
Quasiparticles whose number is not conserved are generally considered not to condense in equilibrium, although a Bose--Einstein condensation can formally appear in particle--Bose field interacting systems. As a basis for examining this mathematically, we formulate the condensation of the free Bose gas with resolvent algebras and probability theory: above the critical density, the condensate state is a standard complex Gaussian mixture of disjoint factor states, which is its central decomposition and corresponds to the ergodic decomposition of a Gaussian measure for the real-time evolution.

\end{abstract}

\setcounter{tocdepth}{3}
\tableofcontents

\section{Introduction}\label{introduction}

When the Bose field of a particle--Bose field interacting system describes quasiparticles, such as phonons, whose number is not conserved, these quasiparticles are generally considered not to undergo Bose--Einstein condensation (BEC) in equilibrium \cite[Abstract]{VIYukalov001}. However, as discussed in \cite{YoshitsuguSekine002} for the Hubbard--phonon interacting system with an infrared singular condition, BEC of the Bose field can formally appear at positive temperature. A mechanism that excludes the appearance of BEC therefore has to be examined mathematically. For this purpose, technical difficulties are expected if operator algebras alone are used. Dereziński and Gérard specialize the work of Klein and Landau \cite{KleinLandau001} in their textbook and develop in detail a correspondence between operator algebras and probability theory at positive temperature in a form suited to applications to nonrelativistic particle--Bose field interacting systems \cite[Sections 21.4 and 21.6]{DerezinskiGerard001}. For an inverse temperature \(\sminvtemperature>0\), this correspondence was extended in \cite{YoshitsuguSekine018} to singular Gaussian \(\sminvtemperature\)-Markov path spaces, that is, to nonnegative injective one-particle Hamiltonians without a spectral gap, including that of the free Bose gas, together with the Euclidean realization of the Araki--Woods representations. The purpose of the present paper is to formulate BEC of the free Bose gas clearly, using probability theory as well as operator algebras, so that the formulation can be used in future studies of particle--Bose field interacting systems. At BEC, a gauge-invariant equilibrium state is expected to be a mixture of states that break the gauge symmetry, labeled by an order parameter, whose components, but not the mixture, cluster. In the operator-algebraic description this concerns the central decomposition of the equilibrium state, and in the probabilistic description the ergodic decomposition of a probability measure; for the free Bose gas, we describe both completely and show that they correspond. Part of the argument should also apply to interacting Bose gases.

The representations of the free Bose gas with a condensate go back to Araki and Woods \cite{ArakiWoods001}, and the condensation of the ideal Bose gas on the Weyl algebra is treated in \cite[Section 5.2.5]{BratteliRobinson004} and in the textbook \cite{AsaoArai028}. The resolvent algebra of Buchholz and Grundling \cite{BuchholzGrundling002} is a \(\oacstar\)-algebra of canonical commutation relations with a nontrivial ideal structure \cite{DetlevBuchholz001}, on which, in contrast to the Weyl algebra, many limits of equilibrium states of Bose gases remain KMS states \cite{DetlevBuchholz002}. Bahns and Buchholz \cite{BahnsBuchholz001} studied on the resolvent algebra the thermodynamic limits of trapped and untrapped equilibrium states of noninteracting bosons at fixed temperature and chemical potential and at the maximal chemical potential, and they obtained condensates in the untrapped case by composing the limit states with the adjoint action of Weyl operators, which breaks the gauge invariance \cite[Sections 3 and 5]{BahnsBuchholz001}. Buchholz introduced proper condensates, characterized by the regular wave functions that stay finitely occupied in the limit \cite{DetlevBuchholz004}, and related them to off-diagonal long-range order \cite{DetlevBuchholz005}.

The main results are the following. On the resolvent algebra, resolvents of the normalized zero-mode test functions give an order parameter whose two scalings detect the divergence of the zero-mode occupation and its macroscopic part, and the scaled zero-mode occupation converges to an exponential law (Theorem \ref{thm:main-order-parameter}). The gauge-invariant BEC state is the mixture, over the standard complex Gaussian law of an order coordinate, of quasi-free states with a constant mean field, which are mutually disjoint factor states realized in the left Araki--Woods representation; resolvents of the scaled zero-mode test functions are asymptotically central, their strong limits generate the center of the BEC representation, and the mixture is the central decomposition (Theorem \ref{thm:main-central-decomposition}). The components are permuted by the gauge group and are clustering, whereas the BEC state is gauge invariant and has long-range order (Theorem \ref{thm:main-gauge-clustering}). On the singular Gaussian \(\sminvtemperature\)-Markov path space of \cite{YoshitsuguSekine018}, the component states and the BEC state are realized by Gaussian measures on one carrier. Under convergence of the finite-volume zero-mode and nonzero-mode forms, the one-dimensional characteristic functions of the density-selected finite-volume sharp-time fields converge to those of the BEC measure. The mixing law is the only law that gives the BEC characteristic functional, and the decomposition is a regular conditional probability (Theorem \ref{thm:main-euclidean-realization}). For the real-time evolution, the component measures are mixing, the invariant \(\sigma\)-algebra is generated by the order coordinate, and the decomposition is the ergodic decomposition, while the center of the BEC representation is parametrized by the same coordinate (Theorem \ref{thm:main-ergodic-decomposition}). For the time-averaged field, the conditional laws given the exterior of bounded regions are the same under all components and under the mixture, the components are tail trivial, and the tail \(\sigma\)-algebra of the mixture is generated by the part of the order coordinate that this field detects. On the \(\sigma\)-algebra generated by this field, the probability measures under which the field is almost surely linear and has these exterior conditional laws form a convex set whose extreme points are exactly its tail-trivial measures; every component is an extreme point, and, above the critical density, the mixture is not (Theorem \ref{thm:main-tail}).

The results concern the free Bose gas only. The Euclidean realization identifies correlations with real-valued test functions; the full correlations with complex test functions require the conjugate momentum observables, which are not reconstructed here. No specification of the Euclidean field on spatial regions is constructed. For a specification on a countable parameter set such as a lattice, the extreme points of the convex set of its Gibbs measures are exactly the tail-trivial ones \cite[Theorem 7.7(a)]{HansOttoGeorgii004}; this framework does not apply directly to the time-averaged field, which is indexed by test functions on \(\fldreal^d\). Our extremality refers instead to the convex set defined by the exterior conditional characteristic functions, for which the corresponding equivalence is proved directly, and whether this set is the set of Gibbs measures of a specification is not determined.

The paper is organized as follows. Section \ref{sec:main-results} fixes the notation and the model and states the main results. Section \ref{sec:resolvent-algebraic-bec} proves the results on the resolvent algebra. Section \ref{sec:gaussian-functional-bec} constructs the Euclidean Gaussian field and proves the results on its decomposition. Section \ref{sec:conclusion} compares the decomposition with the Bose-loop representation and states open problems.

\section{Setting and main results}\label{sec:main-results}

In this section, we fix the notation, the one-particle data, the resolvent algebras, the finite-volume states and the order parameter, the component states and the BEC state, and their Euclidean counterparts, and we summarize the main results with references to the precise statements in Sections \ref{sec:resolvent-algebraic-bec} and \ref{sec:gaussian-functional-bec}.

\subsection{Notation}\label{sec:notation}

In this subsection, we fix the general notation.

Throughout, \(\monnat\) is the set of nonnegative integers, \(\semigrposint=\monnat\setminus\setone{0}\), \(\ringratint\) is the set of integers. For a set \(J\) and \(n\in\semigrposint\), \(\fun{\Sigma_{n,<}}{J}\) is the set of strictly increasing \(n\)-tuples in a totally ordered set \(J\); in particular \[\begin{aligned}
\fun{\Sigma_{2,<}}{\setone{1,\ldots,n}}
&=
\set{\pairbk{j,k}\in\setone{1,\ldots,n}^{2}}{j<k}.
\end{aligned}\] For a set \(E\), \(\fndef{E}\) is its indicator function, and \(\msrdirac_{x}\) is the Dirac measure at a point \(x\). The Fourier transform on \(\fldreal^d\) carries \(2\pi\) in the exponent: \(\fun{\faftr{f}}{k}=\int_{\fldreal^d}\napiernum^{-2\pi\imunit k\cdot x}\fun{f}{x}\opdmsr{x}\) for \(f\in\fun{\lp^1}{\fldreal^d}\). It extends to a unitary operator \(\fafouriertr\) on \(\fun{\lp^2}{\fldreal^d}\), and we write \(\faftr{f}=\fafouriertr f\) for \(f\in\fun{\lp^2}{\fldreal^d}\) \cite[Theorems 8.22 and 8.29]{GeraldFolland001}. If an operator \(T\) on \(\fun{\lp^2}{\fldreal^d}\) is given by \(\fafouriertr Tf=a\fafouriertr f\) on the maximal domain for a measurable function \(a\), the function \(a\) is the Fourier transform of \(T\) and is denoted by \(\faftr{T}\), with the value \(\fun{\faftr{T}}{k}\) at \(k\); the same symbol is used for the restriction of \(\faftr{T}\) to a momentum lattice. For a self-adjoint operator \(A\) on a real or complex Hilbert space, \(\msrprb_A\) denotes its spectral measure, \(\fun{\varphi}{A}=\int\varphi\,\opdmsr{\msrprb_A}\) is the operator of the functional calculus for a Borel function \(\varphi\), and \(\fun{\msrprb_A}{E}\) is the spectral projection of \(A\) for a Borel set \(E\); we use this notation without further comment. For a probability measure \(\msrprb\), \(\prbexp_{\msrprb}\) is the expectation, and \(\prbexp_{\msrprb(x)}\) indicates the variable of integration \(x\) when it is not clear from the integrand; for a measurable map \(\tau\), \(\pushoutrbk{\tau}\msrprb=\msrprb\circ\tau^{-1}\) is the image measure.

The standard complex Gaussian measure is the probability measure \begin{equation}\label{eq:standard-complex-gaussian-measure}
\begin{aligned}
\fun{\msrprb_{\txtstd,\fldcmp}}{B}
&=
\frac{1}{\pi}\int_{B}\napiernum^{-\abs{w}^{2}}\opdmsr{w}
\end{aligned}
\end{equation} on the Borel sets \(B\) of \(\fldcmp=\fldreal^{2}\); it has mean \(0\) and \(\sqfun{\prbexp_{\msrprb_{\txtstd,\fldcmp}(w)}}{\abs{w}^{2}}=1\). A point of \(\fldcmp\) is called an order coordinate and is written in polar coordinates \(\pairbk{r,\theta}\), where \(w=\sqrt{r}\napiernum^{\imunit\theta}\) with the amplitude \(r=\abs{w}^{2}\in\nonneginterval\) and the phase \(\theta\in\fldreal\); an object indexed by \(\pairbk{r,\theta}\) depends only on \(w\), and the phase rotation \(\pairbk{r,\theta}\mapsto\pairbk{r,\theta+\theta_0}\) is the rotation \(w\mapsto\napiernum^{\imunit\theta_0}w\). For every bounded Borel function \(F\) on \(\fldreal^{2}\), with \(\fun{F}{r,\theta}\) the value of \(F\) at \(\sqrt{r}\napiernum^{\imunit\theta}\), \begin{equation}\label{eq:free-bose-radial-phase-law}
\begin{aligned}
\sqfun{\prbexp_{\msrprb_{\txtstd,\fldcmp}}}{F}
&=
\int_{0}^{\infty}\int_{0}^{2\pi}
\fun{F}{r,\theta}\napiernum^{-r}
\frac{\opdmsr{\theta}}{2\pi}\opdmsr{r}.
\end{aligned}
\end{equation} In particular, \(\msrprb_{\txtstd,\fldcmp}\) is invariant under the phase rotations, and the amplitude has the exponential law of mean \(1\).

\subsection{One-particle data and resolvent algebras}\label{sec:one-particle-resolvent-algebras}

In this subsection, we fix the one-particle Hamiltonian, the thermal form, and the resolvent algebras of the free Bose gas.

\begin{defn}[Free-Bose one-particle data]\label{def:free-bose-one-particle-data}
Let $d\in\semigrposint$, $s>0$, and $\sminvtemperature>0$.
The one-particle Hamiltonian $\physham[h]=\rbk{-\laplacian}^{\frac{s}{2}}$ is
the operator on $\fun{\lp^{2}}{\fldreal^{d}}$
with the Fourier transform $\fun{\faftr{\physham[h]}}{k}=\abs{2\pi k}^{s}$.
For a chemical potential $\smchemicalpotential\leq0$, let
$\kappa_{\sminvtemperature,\smchemicalpotential}$ be the Borel
function on $\fldreal$ with
$\fun{\kappa_{\sminvtemperature,\smchemicalpotential}}{\lambda}=\fun{\coth}{\frac{\sminvtemperature\rbk{\lambda-\smchemicalpotential}}{2}}$
for $\lambda>\smchemicalpotential$ and
$\fun{\kappa_{\sminvtemperature,\smchemicalpotential}}{\lambda}=1$ otherwise,
let
$K_{\sminvtemperature,\smchemicalpotential}=\fun{\kappa_{\sminvtemperature,\smchemicalpotential}}{\physham[h]}$,
and let
\begin{equation}\label{eq:free-bose-nonzero-mode-form}
\begin{aligned}
\fun{\opform{q}_{\txtnonzero,\sminvtemperature,\smchemicalpotential}}{f,g}
&=
\bkt{K_{\sminvtemperature,\smchemicalpotential}^{\frac{1}{2}}f}{K_{\sminvtemperature,\smchemicalpotential}^{\frac{1}{2}}g},
\quad
f,g\in\dom K_{\sminvtemperature,\smchemicalpotential}^{\frac{1}{2}}.
\end{aligned}
\end{equation}
We write
$\fun{\opform{q}_{\txtnonzero,\sminvtemperature,\smchemicalpotential}}{f}$ for
$\fun{\opform{q}_{\txtnonzero,\sminvtemperature,\smchemicalpotential}}{f,f}$,
and similarly for the other forms below.
When $\smchemicalpotential=0$, the subscript $\smchemicalpotential$ is omitted.
\end{defn}

Since \(\kappa_{\sminvtemperature,\smchemicalpotential}\geq1\), the form \(\opform{q}_{\txtnonzero,\sminvtemperature,\smchemicalpotential}\) is an inner product on \(\dom K_{\sminvtemperature,\smchemicalpotential}^{\frac{1}{2}}\) for which this domain is complete. The resolvent algebra \(\fun{\oaresolventalgebra}{X,\opform{\sigma}}\) of a real symplectic space \(\pairbk{X,\opform{\sigma}}\) is the \(\oacstar\)-algebra generated by the resolvents \(\fun{\oaresolvent}{\lambda,f}\), \(\lambda\in\fldreal\setminus\setone{0}\), \(f\in X\), subject to the resolvent relations \cite[Definitions 3.1 and 3.4]{BuchholzGrundling002}. A representation \(\oarepn\) is regular if \(\fun{\oarepn}{\fun{\oaresolvent}{1,f}}\) is injective for every \(f\in X\); then \(\fun{\oarepn}{\fun{\oaresolvent}{\lambda,f}}=\rbk{\imunit\lambda-\fun{\phi_{\oarepn}}{f}}^{-1}\) for self-adjoint field operators \(\fun{\phi_{\oarepn}}{f}\) \cite[Corollary 4.4]{BuchholzGrundling002}, and every regular representation is faithful \cite[Theorem 4.10]{BuchholzGrundling002}. A state is regular if its GNS representation is regular. For a real subspace \(S\subset X\), \(\oaresolventalgebra_{S}\) denotes the \(\oacstar\)-subalgebra of \(\fun{\oaresolventalgebra}{X,\opform{\sigma}}\) generated by the resolvents \(\fun{\oaresolvent}{\lambda,f}\), \(f\in S\). If \(S\) is a complex linear subspace on which \(\opform{\sigma}\) is nondegenerate, the canonical map from \(\fun{\oaresolventalgebra}{S,\opform{\sigma}}\) onto \(\oaresolventalgebra_{S}\) is an isometric \(\ast\)-isomorphism, because its composite with a faithful regular representation of \(\fun{\oaresolventalgebra}{X,\opform{\sigma}}\) is a regular representation of \(\fun{\oaresolventalgebra}{S,\opform{\sigma}}\), which is faithful by \cite[Theorem 4.10]{BuchholzGrundling002}.

\begin{defn}[Free-Bose resolvent algebras]\label{def:free-bose-resolvent-system}
In the setting of Definition \ref{def:free-bose-one-particle-data}, let
$\sphilbfrak{H}_{\sminvtemperature}$ be the
complex Hilbert space $\dom K_{\sminvtemperature}^{\frac{1}{2}}$ with the inner
product
$\opform{q}_{\txtnonzero,\sminvtemperature}$, and let
\begin{equation}\label{eq:free-bose-zero-mode-domain}
\begin{aligned}
\sphilbfrak{D}_{0,\sminvtemperature}
&=
\fun{\lp^{1}}{\fldreal^{d}}\cap\sphilbfrak{H}_{\sminvtemperature}.
\end{aligned}
\end{equation}
The symplectic form is
$\fun{\opform{\sigma}}{f,g}=\opimag\bkt{f}{g}_{\fun{\lp^{2}}{\fldreal^{d}}}$ for
$f,g\in\sphilbfrak{H}_{\sminvtemperature}$.
For a bounded open set $O\subset\fldreal^{d}$, let
$$\begin{aligned}
\sphilbfrak{H}_{\sminvtemperature,O}
&=
\set{f\in\sphilbfrak{H}_{\sminvtemperature}}{f=\fndef{O}f\ \text{almost everywhere}}.
\end{aligned}$$
In the ambient algebra
$\oaresolventalgebra=\fun{\oaresolventalgebra}{\sphilbfrak{H}_{\sminvtemperature},\opform{\sigma}}$,
set
\begin{equation}\label{eq:free-bose-resolvent-subalgebras}
\begin{aligned}
\fun{\oaresolventalgebra}{\sphilbfrak{D}_{0,\sminvtemperature}}
&=
\oaresolventalgebra_{\sphilbfrak{D}_{0,\sminvtemperature}},
&
\fun{\oaresolventalgebra}{O}
&=
\oaresolventalgebra_{\sphilbfrak{H}_{\sminvtemperature,O}},
&
\oaresolventalgebra_{\txtloc}
&=
\bigcup_{O}\fun{\oaresolventalgebra}{O}.
\end{aligned}
\end{equation}
The free dynamics is
$\fun{\alpha_{\txtfr,t}}{\fun{\oaresolvent}{\lambda,f}}=\fun{\oaresolvent}{\lambda,\napiernum^{\imunit t\physham[h]}f}$,
$t\in\fldreal$.
For $L>0$, $I_L^d=\closedinterval{-\frac{L}{2}}{\frac{L}{2}}^{d}$.
\end{defn}

The space \(\sphilbfrak{D}_{0,\sminvtemperature}\) is the domain on which the zero-momentum value \(\fun{\faftr{f}}{0}\) is defined; it is auxiliary in the sense that \(\gtclos{\oaresolventalgebra_{\txtloc}}^{\norm{\cdot}}\), \(\fun{\oaresolventalgebra}{\sphilbfrak{D}_{0,\sminvtemperature}}\), and \(\oaresolventalgebra\) are strictly increasing (Lemma \ref{lem:free-bose-resolvent-inclusions}).

\subsection{Finite-volume states and the order parameter}\label{sec:finite-volume-order-parameter}

In this subsection, we fix the periodic grand-canonical states, the density-selected chemical potentials, and the order parameter.

For \(L>0\), let \(\setlattice_L=L^{-1}\ringratint\), let \(e_k=L^{-\frac{d}{2}}\napiernum^{2\pi\imunit k\cdot x}\), \(k\in\setlattice_L^d\), be the periodic Fourier basis of \(\fun{\lp^2}{I_L^d}\), and let \(f_k=\int_{I_L^d}\napiernum^{-2\pi\imunit k\cdot x}\fun{f}{x}\opdmsr{x}\). The periodic one-particle Hamiltonian \(\physham[h]_{I_L^d}\) is the self-adjoint operator with \(\physham[h]_{I_L^d}e_k=\fun{\faftr{\physham[h]}}{k}e_k\). For \(\smchemicalpotential<0\), the finite-volume forms are \begin{equation}\label{eq:free-bose-finite-volume-forms}
\begin{aligned}
\fun{\opform{q}_{0,\sminvtemperature,\smchemicalpotential,L}}{f,g}
&=
\frac{1}{L^d}\fun{\coth}{-\frac{\sminvtemperature\smchemicalpotential}{2}}\cmpconj{f_{0}}g_{0},
\\ 
\fun{\opform{q}_{\txtnonzero,\sminvtemperature,\smchemicalpotential,L}}{f,g}
&=
\frac{1}{L^d}\sum_{k\in\setlattice_L^{d}\setminus\setone{0}}
\fun{\coth}{\frac{\sminvtemperature\rbk{\fun{\faftr{\physham[h]}}{k}-\smchemicalpotential}}{2}}\cmpconj{f_{k}}g_{k},
\\ 
\fun{\opform{q}_{\txtloc,\sminvtemperature,\smchemicalpotential,I_L^d}}{f,g}
&=
\fun{\opform{q}_{0,\sminvtemperature,\smchemicalpotential,L}}{f,g}
+\fun{\opform{q}_{\txtnonzero,\sminvtemperature,\smchemicalpotential,L}}{f,g},
\end{aligned}
\end{equation} for \(f,g\in\fun{\lp^2}{I_L^d}\). Let \(\fun{\opform{\sigma}_L}{f,g}=\opimag\bkt{f}{g}_{\fun{\lp^2}{I_L^d}}\). The periodic grand-canonical state \(\oastate[\psi_{\txtgrandcanonical,I_L^d,\sminvtemperature,\smchemicalpotential}]\) is the state of \(\fun{\oaresolventalgebra}{\fun{\lp^2}{I_L^d},\opform{\sigma}_L}\) whose correlations are given by \eqref{eq:free-bose-shifted-quasifree-resolvent} with \(\opform{\sigma}_L\) in place of \(\opform{\sigma}\), \(\opform{q}_{\txtloc,\sminvtemperature,\smchemicalpotential,I_L^d}\) in place of \(\opform{q}_{\txtnonzero,\sminvtemperature}\), and \(0\) in place of \(\ell_{\sminvtemperature,r,\theta}\); it exists and is unique by Proposition \ref{prop:free-bose-finite-volume-resolvent-expectations}. Its two-point function is that of the grand-canonical Gibbs state of the free Bose gas in \(I_L^d\) with periodic boundary conditions \cite[Propositions 5.2.27 and 5.2.28]{BratteliRobinson004}.

For \(y=\napiernum^{-\sminvtemperature\smchemicalpotential}>1\), the mean particle density in \(I_L^d\) and the mean zero-mode occupation are \begin{equation}\label{eq:free-bose-zero-mode-occupation-parameters}
\begin{aligned}
\fun{\smnumberdensity_{L^d}}{\sminvtemperature,y}
&=
\frac{1}{L^d}\sum_{k\in\setlattice_L^{d}}\frac{1}{y\napiernum^{\sminvtemperature\fun{\faftr{\physham[h]}}{k}}-1},
&
\fun{N_0}{y}
&=
\frac{1}{y-1}.
\end{aligned}
\end{equation} For \(y\geq1\), set \begin{equation}\label{eq:free-bose-resolvent-critical-density}
\begin{aligned}
\fun{\smnumberdensity}{\sminvtemperature,y}
&=
\int_{\fldreal^{d}}\rbk{y\napiernum^{\sminvtemperature\fun{\faftr{\physham[h]}}{k}}-1}^{-1}\opdmsr{k}
\in\leftopeninterval{0}{\infty},
\end{aligned}
\end{equation} let \(\fun{\smnumberdensity_{\txtcritical}}{\sminvtemperature}=\fun{\smnumberdensity}{\sminvtemperature,1}\) be the critical density, which is finite if and only if \(s<d\), fix a particle density \(\bar{\smnumberdensity}>0\), and let the condensate density be \begin{equation}\label{eq:free-bose-resolvent-condensate-density}
\begin{aligned}
\fun{\smnumberdensity_0}{\sminvtemperature}
&=
\max\setone{\bar{\smnumberdensity}-\fun{\smnumberdensity_{\txtcritical}}{\sminvtemperature},0}.
\end{aligned}
\end{equation} By Proposition \ref{prop:free-bose-density-selected-fugacity}, there is exactly one \(y_{L^d}>1\) with \(\fun{\smnumberdensity_{L^d}}{\sminvtemperature,y_{L^d}}=\bar{\smnumberdensity}\); the density-selected chemical potential is \begin{equation}\label{eq:free-bose-density-selected-chemical-potential}
\begin{aligned}
\smchemicalpotential_{L^d}
&=
-\sminvtemperature^{-1}\log y_{L^d}
<
0.
\end{aligned}
\end{equation}

\begin{defn}[Order parameter]\label{def:free-bose-scaled-zero-mode-probes}
For $L>0$, set
\begin{equation}\label{eq:free-bose-scaled-zero-mode-tests}
\begin{aligned}
\mathsf{b}_{0,L}
&=
L^{-\frac{d}{2}}\fndef{I_L^d},
&
\mathsf{b}_{1,L}
&=
L^{-d}\fndef{I_L^d}.
\end{aligned}
\end{equation}
The order parameter is
\begin{equation}\label{eq:free-bose-zero-mode-resolvent-probes}
\begin{aligned}
\mathsf{o}_{\txtbec,\#}
&=
\lim_{L\to\infty}
\imunit
\fun{\oastate[\psi_{\txtgrandcanonical,I_L^d,\sminvtemperature,\smchemicalpotential_{L^d}}]}{\fun{\oaresolvent}{1,\mathsf{b}_{\#,L}}}
\in\closedinterval{0}{1},
\quad
\#=0,1,
\end{aligned}
\end{equation}
whenever the limit exists.
\end{defn}

The test function \(\mathsf{b}_{0,L}\) is normalized in \(\fun{\lp^2}{I_L^d}\) and \(\mathsf{b}_{1,L}\) has integral \(1\); their only nonzero Fourier coefficients are \(\rbk{\mathsf{b}_{0,L}}_0=L^{\frac{d}{2}}\) and \(\rbk{\mathsf{b}_{1,L}}_0=1\).

\subsection{Component states and the BEC state}\label{sec:component-states}

In this subsection, we define the zero-mode form, the mean functionals, the component states, and the BEC state.

For \(\fun{\smnumberdensity_0}{\sminvtemperature}\) of \eqref{eq:free-bose-resolvent-condensate-density}, \(f,g\in\sphilbfrak{D}_{0,\sminvtemperature}\), and an order coordinate \(\pairbk{r,\theta}\), set \begin{equation}\label{eq:free-bose-zero-mode-form}
\begin{aligned}
\fun{\opform{q}_0}{f,g}
&=
2\fun{\smnumberdensity_0}{\sminvtemperature}\opreal\rbk{\cmpconj{\fun{\faftr{f}}{0}}\fun{\faftr{g}}{0}},
&
\fun{\ell_{\sminvtemperature,r,\theta}}{f}
&=
\sqrt{2\fun{\smnumberdensity_0}{\sminvtemperature}r}\,\opreal\rbk{\napiernum^{\imunit\theta}\fun{\faftr{f}}{0}}.
\end{aligned}
\end{equation} The functional \(\ell_{\sminvtemperature,r,\theta}\) is real linear and depends only on the point \(\sqrt{r}\napiernum^{\imunit\theta}\), and \[\begin{aligned}
\fun{\opform{q}_0}{f}
&=
2\fun{\smnumberdensity_0}{\sminvtemperature}\abs{\fun{\faftr{f}}{0}}^{2}.
\end{aligned}\]

\begin{defn}[Component states, BEC state, and nonzero-mode state]\label{def:free-bose-resolvent-component-states}
For an order coordinate $\pairbk{r,\theta}$, the component state
$\oastate[\psi_{\sminvtemperature,r,\theta}]$ is the
state of $\fun{\oaresolventalgebra}{\sphilbfrak{D}_{0,\sminvtemperature}}$ such
that, for $n\in\semigrposint$,
$\lambda_1,\ldots,\lambda_n\in\fldreal\setminus\setone{0}$, and
$f_1,\ldots,f_n\in\sphilbfrak{D}_{0,\sminvtemperature}$,
\begin{equation}\label{eq:free-bose-shifted-quasifree-resolvent}
\begin{aligned}
&\fun{\oastate[\psi_{\sminvtemperature,r,\theta}]}{\fun{\oaresolvent}{\lambda_1,f_1}\cdots\fun{\oaresolvent}{\lambda_n,f_n}}
\\ 
&=\rbk{-\imunit}^{n}\int_0^{\rbk{\sgn\lambda_1}\infty}\cdots\int_0^{\rbk{\sgn\lambda_n}\infty}
\fnexp{-\frac{\imunit}{2}\sum_{\rbk{j,k}\in\fun{\Sigma_{2,<}}{\setone{1,\ldots,n}}}t_jt_k\fun{\opform{\sigma}}{f_j,f_k}}
\\ 
&\quad\times
\fnexp{-\sum_{j=1}^{n}\rbk{\lambda_j+\imunit\fun{\ell_{\sminvtemperature,r,\theta}}{f_j}}t_j}
\fnexp{-\frac{1}{4}\fun{\opform{q}_{\txtnonzero,\sminvtemperature}}{\sum_{j=1}^{n}t_jf_j}}
\opdmsr{t_n}\cdots\opdmsr{t_1};
\end{aligned}
\end{equation}
it exists and is unique by Lemma \ref{lem:free-bose-shifted-quasifree-states}.
The BEC state is
\begin{equation}\label{eq:free-bose-resolvent-state-mixture}
\begin{aligned}
\fun{\oastate[\psi_{\txtbec,\sminvtemperature}]}{A}
&=
\sqfun{\prbexp_{\msrprb_{\txtstd,\fldcmp}(r,\theta)}}{\fun{\oastate[\psi_{\sminvtemperature,r,\theta}]}{A}},
\quad
A\in\fun{\oaresolventalgebra}{\sphilbfrak{D}_{0,\sminvtemperature}}.
\end{aligned}
\end{equation}
The nonzero-mode state $\oastate[\psi_{\txtnonzero,\sminvtemperature}]$ is the
state of $\oaresolventalgebra$ given by
\eqref{eq:free-bose-shifted-quasifree-resolvent} with
$\ell_{\sminvtemperature,r,\theta}$ replaced by $0$ and
$f_j\in\sphilbfrak{H}_{\sminvtemperature}$.
\end{defn}

The integrand of \eqref{eq:free-bose-resolvent-state-mixture} is continuous in \(\pairbk{r,\theta}\) by Lemma \ref{lem:free-bose-shifted-quasifree-states}, and the BEC state has the correlations \eqref{eq:free-bose-shifted-quasifree-resolvent} with the mean term deleted and with \(\fun{\opform{q}_{\txtnonzero,\sminvtemperature}}{\cdot}\) replaced by \(\fun{\opform{q}_0}{\cdot}+\fun{\opform{q}_{\txtnonzero,\sminvtemperature}}{\cdot}\) (Proposition \ref{prop:free-bose-resolvent-component-correlations}). When \(\fun{\smnumberdensity_0}{\sminvtemperature}=0\), every component state equals the restriction of the nonzero-mode state.

\subsection{The Euclidean field}\label{sec:euclidean-field-setting}

In this subsection, we fix the Euclidean objects used in the main results.

Let \(\sphilbfrak{h}_{\fldreal}\) be \(\fun{\lp^2}{\fldreal^d;\fldcmp}\) regarded as a real Hilbert space with the inner product \(\opreal\bkt{\cdot}{\cdot}\). The one-particle Hamiltonian \(\physham[h]\) of Definition \ref{def:free-bose-one-particle-data} is the basic model of \cite[Example 2.2]{YoshitsuguSekine018}, and we use the test space of \(f\) with \(\fafouriertr f\in\fun{C_{\mathrm{c}}^{\infty}}{\fldreal^d\setminus\setone{0}}\) \cite[Example 3.6]{YoshitsuguSekine018} and an invariant regularizer of \cite[Lemma 4.2]{YoshitsuguSekine018}. The singular Gaussian \(\sminvtemperature\)-Markov path space of \cite[Definition 4.4 and Theorem 4.6]{YoshitsuguSekine018} for these data is denoted by \begin{equation}\label{eq:free-bose-centered-path-space}
\begin{aligned}
\rbk{\prbqspace_{\sminvtemperature,-},\mblfmlfrak{S}_{\sminvtemperature},\mblfmlfrak{S}_{\sminvtemperature,0},
\fml{U_{\sminvtemperature,t}}{t\in\fldreal},R_{\sminvtemperature},\msrprb_{\txtnonzero,\sminvtemperature}};
\end{aligned}
\end{equation} that is, \(\prbqspace_{\sminvtemperature,-}\), \(\msrprb_{\txtnonzero,\sminvtemperature}\), \(\tilde{u}_{\sminvtemperature,t}\), \(\tilde{r}_{\sminvtemperature}\), \(j_{\sminvtemperature,t}\), and \(\sphilbfrak{H}_{\sminvtemperature,C_{\sminvtemperature}}\) denote the objects \(\prbqspace_-\), \(\msrprb_{\sminvtemperature,0}\), \(\tilde{u}_t\), \(\tilde{r}\), \(j_t\), and \(\sphilbfrak{H}_{\sminvtemperature,C}\) of \cite[Sections 3 and 4]{YoshitsuguSekine018} for the covariance \(C_{\sminvtemperature}=\rbk{D_t^2+\physham[h]^2}^{-1}\) on \(\fun{\lp^2}{S_{\sminvtemperature};\sphilbfrak{h}_{\fldreal}}\), where \(S_{\sminvtemperature}=\fldreal/\sminvtemperature\ringratint\) is the imaginary-time circle. Here \(D_t^2\) is the closure of \(-\pd{^2}{t^2}\) on the smooth \(\sminvtemperature\)-periodic \(\sphilbfrak{h}_{\fldreal}\)-valued functions \cite[Section 2.2, (4)]{YoshitsuguSekine018}. For the space \(\sphilbfrak{h}_{\fldreal,\smnumberdensity}\) and the isometry \(\iota_{\physham[h]}\) of \cite[(26) and Lemma 5.2]{YoshitsuguSekine018}, \(f\in\sphilbfrak{h}_{\fldreal,\smnumberdensity}\), and \(t\in S_{\sminvtemperature}\), we write \(\fun{\opfocksegal_t}{\rbk{2\physham[h]}^{\frac{1}{2}}f}=\fun{\opfocksegal}{j_{\sminvtemperature,t}\fun{\iota_{\physham[h]}}{f}}\) for the Segal sharp-time field \cite[Corollary 5.7]{YoshitsuguSekine018}.

The total carrier is \begin{equation}\label{eq:free-bose-total-carrier}
\begin{aligned}
\pairbk{\prbqspace_{\txttot,\sminvtemperature},\mblfmlfrak{S}_{\txttot,\sminvtemperature},\mblfmlfrak{S}_{\txttot,\sminvtemperature,0}}
&=
\pairbk{\prbqspace_{\sminvtemperature,-}\times\fldreal^{2},\,
\mblfmlfrak{S}_{\sminvtemperature}\times\fun{\mblfmlborel}{\fldreal^{2}},\,
\mblfmlfrak{S}_{\sminvtemperature,0}\times\fun{\mblfmlborel}{\fldreal^{2}}},
\end{aligned}
\end{equation} with the time translations \(U_{\sminvtemperature,t}F=F\circ\rbk{\tilde{u}_{\sminvtemperature,-t}\times\id}\) and the reflection \(R_{\sminvtemperature}F=F\circ\rbk{\tilde{r}_{\sminvtemperature}\times\id}\), where the second factor \(\fldreal^{2}=\fldcmp\) is the space of order coordinates. For \(f\in\sphilbfrak{D}_{0,\sminvtemperature}\), regarded in \(\sphilbfrak{h}_{\fldreal,\smnumberdensity}\), and \(t\in S_{\sminvtemperature}\), the sharp-time field of the total system is \begin{equation}\label{eq:free-bose-bec-enlarged-sharp-time-field}
\begin{aligned}
\pairbk{q,r,\theta}
&\mapsto
\fun{\fun{\opfocksegal_t}{\rbk{2\physham[h]}^{\frac{1}{2}}f}}{q}
+\fun{\ell_{\sminvtemperature,r,\theta}}{f},
\end{aligned}
\end{equation} and the component measures and the BEC measure are \begin{equation}\label{eq:free-bose-bec-component-carrier-law}
\begin{aligned}
\msrprb_{\sminvtemperature,r,\theta}
&=
\msrprb_{\txtnonzero,\sminvtemperature}\otimes\msrdirac_{\pairbk{r,\theta}},
&
\msrprb_{\txtbec,\sminvtemperature}
&=
\msrprb_{\txtnonzero,\sminvtemperature}\otimes\msrprb_{\txtstd,\fldcmp}.
\end{aligned}
\end{equation}

\subsection{Main results}\label{sec:main-theorems}

In this subsection, we summarize the main results; each theorem is a summary of the results cited after it, whose hypotheses and conclusions are stated precisely there.

\begin{thm}[Order parameter]\label{thm:main-order-parameter}
Assume that either the limit $y_{\infty}$ of $y_{L^d}$ is larger than $1$ or that
$\lim_{L\to\infty}L^{-d}\fun{N_0}{y_{L^d}}$ exists
and is positive.
Then the condensate density is positive if and only if
$\mathsf{o}_{\txtbec,0}=0$, if and only if
$\mathsf{o}_{\txtbec,1}<1$.
If
$\bar{\smnumberdensity}>\fun{\smnumberdensity_{\txtcritical}}{\sminvtemperature}$,
the law of the scaled zero-mode
occupation $L^{-d}n_0$ in the grand-canonical occupation law at $y_{L^d}$ converges
weakly to the exponential law of mean
$\fun{\smnumberdensity_0}{\sminvtemperature}$, which is the law of
$\fun{\smnumberdensity_0}{\sminvtemperature}r$ under
$\msrprb_{\txtstd,\fldcmp}$.
\end{thm}

These are Propositions \ref{prop:free-bose-zero-mode-resolvent-probes} and \ref{prop:free-zero-mode-geometric-limit}.

\begin{thm}[Central decomposition]\label{thm:main-central-decomposition}
Assume $0<s<d$ and $\fun{\smnumberdensity_0}{\sminvtemperature}>0$.
Then the component states with $r>0$ are mutually disjoint factor states,
realized in the left Araki--Woods
representation with constant mean fields $\ell_{\sminvtemperature,r,\theta}$,
and the GNS representation of the BEC state
is the direct integral of the component representations over
$\msrprb_{\txtstd,\fldcmp}$, whose von Neumann algebra has
the center $\fun{\lp^\infty}{\fldreal^{2},\msrprb_{\txtstd,\fldcmp}}$.
The resolvents $\fun{\oaresolvent}{1,\mathsf{b}_{1,L}}$ and
$\fun{\oaresolvent}{1,\imunit\mathsf{b}_{1,L}}$ are
asymptotically central in norm, they have norm one in the BEC representation,
and their strong limits are the
multiplication operators by
$\rbk{\imunit-\sqrt{2\fun{\smnumberdensity_0}{\sminvtemperature}r}\fun{\cos}{\theta}}^{-1}$
and
$\rbk{\imunit+\sqrt{2\fun{\smnumberdensity_0}{\sminvtemperature}r}\fun{\sin}{\theta}}^{-1}$,
which generate the center.
The norm closure of the represented algebra has trivial center.
\end{thm}

These are Propositions \ref{prop:free-bose-order-parameter-norm-one} and \ref{prop:free-bose-resolvent-central-probes}, Lemma \ref{lem:free-bose-araki-woods-components}, and Theorem \ref{thm:free-bose-resolvent-central-decomposition}.

\begin{thm}[Gauge symmetry breaking and clustering]\label{thm:main-gauge-clustering}
The gauge automorphisms rotate the phases of the component states, the BEC state
is gauge invariant, every component state
is clustering for the spatial translations, and, if
$\fun{\smnumberdensity_0}{\sminvtemperature}>0$, the BEC state has
long-range order in the two-point functions of the resolvents.
Every component state satisfies the KMS boundary relation for $\alpha_{\txtfr}$
on the resolvent algebra of every
$\napiernum^{\imunit t\physham[h]}$-invariant subspace of
$\sphilbfrak{D}_{0,\sminvtemperature}$.
\end{thm}

These are Propositions \ref{prop:free-bose-resolvent-gauge-clustering} and \ref{prop:free-bose-resolvent-component-correlations}(4).

\begin{thm}[Euclidean realization]\label{thm:main-euclidean-realization}
On the total carrier \eqref{eq:free-bose-total-carrier}, the component measures
are mutually singular, the BEC measure is
their mixture over $\msrprb_{\txtstd,\fldcmp}$, and the characteristic functions
of the sharp-time fields
\eqref{eq:free-bose-bec-enlarged-sharp-time-field} under these measures are
those of the component states and of the BEC
state.
Assume
$\bar{\smnumberdensity}>\fun{\smnumberdensity_{\txtcritical}}{\sminvtemperature}$.
For real-valued test functions whose finite-volume zero-mode and nonzero-mode
forms converge to those of the limit,
the one-dimensional characteristic functions of the sharp-time fields of
the finite-volume Gaussian measures at the
density-selected chemical potentials converge to those of the BEC measure.
If $\fun{\smnumberdensity_0}{\sminvtemperature}>0$, then
$\msrprb_{\txtstd,\fldcmp}$ is the only mixing law whose mixture
has the BEC characteristic functional, and the component measures form a regular
conditional probability of the BEC
measure given the order coordinate.
\end{thm}

These are Propositions \ref{prop:free-bose-bec-component-carrier} and \ref{prop:free-bose-component-disintegration}, Theorem \ref{thm:free-bose-density-selected-functional-integral}, Lemma \ref{lem:free-bose-mixing-law-uniqueness}, and Corollary \ref{cor:free-bose-gaussian-components-carrier}.

\begin{thm}[Ergodic decomposition]\label{thm:main-ergodic-decomposition}
On a carrier of continuous real-time paths, on which the real-time evolution
$\napiernum^{\imunit t\physham[h]}$ is
realized by a Borel action of $\fldreal$, the component measures are invariant
and mixing, the BEC measure is invariant and
not mixing, the invariant $\sigma$-algebra is generated by the order coordinate
up to null sets, and the decomposition into
the component measures is the ergodic decomposition.
If $0<s<d$ and $\fun{\smnumberdensity_0}{\sminvtemperature}>0$, the center of
the BEC representation is the multiplication
algebra of the same order coordinate.
The gauge-invariant events see the order coordinate only through the amplitude
$r$.
\end{thm}

These are Proposition \ref{prop:free-bose-real-time-realization}, Corollaries \ref{cor:free-bose-real-time-mixing} and \ref{cor:free-bose-invariant-decomposition-center}, and Proposition \ref{prop:free-bose-amplitude-disintegration}.

\begin{thm}[Exterior conditional laws and tail]\label{thm:main-tail}
Assume $0<s<d$.
For the time-averaged field, the conditional characteristic functions given the
exterior of a bounded open set are the same
Gaussian functions under every component measure and under the BEC measure,
every component measure is trivial on the tail
$\sigma$-algebra, and, if
$\bar{\smnumberdensity}>\fun{\smnumberdensity_{\txtcritical}}{\sminvtemperature}$,
the tail
$\sigma$-algebra of the BEC measure is generated up to null sets by
$\sqrt{2\fun{\smnumberdensity_0}{\sminvtemperature}r}\fun{\cos}{\theta}$ and is
not trivial.
On the $\sigma$-algebra generated by this field, the probability measures under
which it is almost surely linear and has these conditional characteristic
functions form a convex set, whose extreme points are exactly its tail-trivial
measures.
The restrictions of the component measures are extreme points, and, if
$\bar{\smnumberdensity}>\fun{\smnumberdensity_{\txtcritical}}{\sminvtemperature}$,
the restriction of the BEC measure is not.
\end{thm}

This is Theorem \ref{thm:free-bose-time-averaged-exterior-tail} and Proposition \ref{prop:free-bose-time-averaged-extreme-points}.

\section{The resolvent-algebraic description}\label{sec:resolvent-algebraic-bec}

In this section, we prove the results on the resolvent algebra: the inclusions of the resolvent algebras and the finite-volume expectations, the order parameter and the zero-mode occupation law, the properties of the component states and of the BEC state, the asymptotically central resolvents and the central decomposition, and gauge symmetry breaking and clustering.

\subsection{Resolvent algebras and finite-volume expectations}\label{sec:free-bose-resolvent-finite-volume}

In this subsection, we compare the resolvent algebras of Definition \ref{def:free-bose-resolvent-system} and compute the expectations of the periodic grand-canonical states.

\begin{lem}[Inclusions of the free-Bose resolvent algebras]\label{lem:free-bose-resolvent-inclusions}
Let $S_{\txtloc}=\bigcup_O\sphilbfrak{H}_{\sminvtemperature,O}$ over the bounded
open sets $O\subset\fldreal^d$.
Then the canonical maps of
$\fun{\oaresolventalgebra}{\sphilbfrak{H}_{\sminvtemperature,O},\opform{\sigma}}$
and of
$\fun{\oaresolventalgebra}{\sphilbfrak{D}_{0,\sminvtemperature},\opform{\sigma}}$
into $\oaresolventalgebra$ are isometric
with ranges $\fun{\oaresolventalgebra}{O}$ and
$\fun{\oaresolventalgebra}{\sphilbfrak{D}_{0,\sminvtemperature}}$, and
\begin{equation}\label{eq:free-bose-resolvent-inclusions}
\begin{aligned}
\gtclos{\oaresolventalgebra_{\txtloc}}^{\norm{\cdot}}
&=
\oaresolventalgebra_{S_{\txtloc}}
\subsetneq
\fun{\oaresolventalgebra}{\sphilbfrak{D}_{0,\sminvtemperature}}
\subsetneq
\oaresolventalgebra.
\end{aligned}
\end{equation}
\end{lem}

\begin{proof}
For a complex linear subspace $G$ and $0\neq g\in G$,
$\fun{\opform{\sigma}}{g,\imunit g}=\norm{g}_{\lp^2}^{2}\neq0$, so
the isometry statements follow from the remark before Definition
\ref{def:free-bose-resolvent-system}.
The equality in \eqref{eq:free-bose-resolvent-inclusions} holds because the
union is directed, and the first inclusion
holds because $\norm{f}_{\lp^1}\leq\abs{O}^{\frac{1}{2}}\norm{f}_{\lp^2}$ for
$f\in\sphilbfrak{H}_{\sminvtemperature,O}$.
For a real subspace $S$ and $f\notin S$, the numbers equal to
$\rbk{\imunit\mu}^{-1}$ at $\fun{\oaresolvent}{\mu,g}$ for $g\in\fldreal f$ and
to $0$
for the other $g$ satisfy the resolvent relations, because $\fldreal f$ is
isotropic, and so do the numbers equal to
$\rbk{\imunit\mu}^{-1}$ for $g=0$ and to $0$ otherwise; by the universal
property \cite[Definition 3.4]{BuchholzGrundling002}
they define two characters of $\oaresolventalgebra$, which agree on
$\oaresolventalgebra_S$ and differ by
$\rbk{\imunit\lambda}^{-1}$ at $\fun{\oaresolvent}{\lambda,f}$.
Hence $\norm{\fun{\oaresolvent}{\lambda,f}-A}\geq\rbk{2\abs{\lambda}}^{-1}$ for
$A\in\oaresolventalgebra_S$.
For the first strict inclusion, let $N\in\monnat$ with $4N+d>s$ and
$\fun{f}{x}=\rbk{-\laplacian}^{N}\napiernum^{-\abs{x}^2}$.
Then $\fun{\faftr{f}}{k}$ is a constant times
$\abs{k}^{2N}\napiernum^{-\pi^2\abs{k}^2}$, and
$\fun{\kappa_{\sminvtemperature}}{\lambda}\leq1+2\rbk{\sminvtemperature\lambda}^{-1}$
gives
$f\in\sphilbfrak{D}_{0,\sminvtemperature}\setminus S_{\txtloc}$.
For the second, $u=\fafouriertr^{-1}\fndef{\closedinterval{1}{2}^d}$ lies in
$\sphilbfrak{H}_{\sminvtemperature}$, and
$\fun{u}{x}=\prod_{j=1}^d\napiernum^{3\pi\imunit x_j}\rbk{\pi x_j}^{-1}\fun{\sin}{\pi x_j}$
is not integrable.
\end{proof}

\begin{prop}[Resolvent expectations of the periodic grand-canonical state]\label{prop:free-bose-finite-volume-resolvent-expectations}
Let $L>0$ and $\smchemicalpotential<0$.
Then there is exactly one state of
$\fun{\oaresolventalgebra}{\fun{\lp^2}{I_L^d},\opform{\sigma}_L}$ with the
correlations of the periodic grand-canonical state
$\oastate[\psi_{\txtgrandcanonical,I_L^d,\sminvtemperature,\smchemicalpotential}]$,
it is regular, and, for $\lambda\in\fldreal\setminus\setone{0}$ and
$f\in\fun{\lp^2}{I_L^d}$,
\begin{equation}\label{eq:free-bose-finite-volume-resolvent-expectation}
\begin{aligned}
\fun{\oastate[\psi_{\txtgrandcanonical,I_L^d,\sminvtemperature,\smchemicalpotential}]}{\fun{\oaresolvent}{\lambda,f}}
&=
-\imunit
\int_0^{\rbk{\sgn\lambda}\infty}
\napiernum^{-\lambda t}
\fnexp{-\frac{t^{2}}{4}\fun{\opform{q}_{\txtloc,\sminvtemperature,\smchemicalpotential,I_L^d}}{f}}
\opdmsr{t}.
\end{aligned}
\end{equation}
\end{prop}

\begin{proof}
For $k\in\setlattice_L^d$,
$\fun{\faftr{\physham[h]}}{k}-\smchemicalpotential\geq-\smchemicalpotential>0$,
and $\coth$ decreases on $\rbk{0,\infty}$ with values greater than $1$; hence
the coefficients of $\cmpconj{f_k}g_k$ in
\eqref{eq:free-bose-finite-volume-forms} lie between $L^{-d}$ and
$L^{-d}\fun{\coth}{-\frac{\sminvtemperature\smchemicalpotential}{2}}$.
Parseval's identity $\norm{f}^2=L^{-d}\sum_k\abs{f_k}^2$ shows that
$\opform{q}_{\txtloc,\sminvtemperature,\smchemicalpotential,I_L^d}$ is finite on
$\fun{\lp^2}{I_L^d}$ and that
$\norm{f}^2\leq\fun{\opform{q}_{\txtloc,\sminvtemperature,\smchemicalpotential,I_L^d}}{f}$,
and hence
$\abs{\fun{\opform{\sigma}_L}{u,v}}\leq\fun{\opform{q}_{\txtloc,\sminvtemperature,\smchemicalpotential,I_L^d}}{u}^{\frac{1}{2}}\fun{\opform{q}_{\txtloc,\sminvtemperature,\smchemicalpotential,I_L^d}}{v}^{\frac{1}{2}}$.
Therefore the sesquilinear extension $\opform{s}_L$ of
$\frac{1}{2}\opreal\fun{\opform{q}_{\txtloc,\sminvtemperature,\smchemicalpotential,I_L^d}}{u,v}+\frac{\imunit}{2}\fun{\opform{\sigma}_L}{u,v}$
to the complexification of the real vector space $\fun{\lp^2}{I_L^d}$ is
positive semidefinite and satisfies
$\fun{\opform{s}_L}{f,g}-\fun{\opform{s}_L}{g,f}=\imunit\fun{\opform{\sigma}_L}{f,g}$,
as in the proof of Lemma \ref{lem:free-bose-shifted-quasifree-states}.
The regular quasi-free state for this form
\cite[p.~2736 and Remark 4.12]{BuchholzGrundling002}, its values on products
of resolvents \cite[Eq.~(17), p.~2736]{BuchholzGrundling002}, analytic
continuation in each parameter with
$\fun{\oaresolvent}{-z,f}=-\fun{\oaresolvent}{z,-f}$
\cite[Eq.~(10), p.~2731]{BuchholzGrundling002}, and the density of the span
of products of resolvents give, as in that proof, existence, regularity, and
uniqueness.
The case of one resolvent is
\eqref{eq:free-bose-finite-volume-resolvent-expectation}.
\end{proof}

The finite-volume states are not compatible with embeddings of the algebras \(\fun{\oaresolventalgebra}{\fun{\lp^2}{I_L^d},\opform{\sigma}_L}\) at the density-selected chemical potentials \eqref{eq:free-bose-density-selected-chemical-potential}, which depend on \(L\), and the infinite-volume BEC state is therefore introduced through its correlations in Definition \ref{def:free-bose-resolvent-component-states} and compared with the finite-volume states through the order parameter and through the Euclidean fields of Section \ref{sec:free-bose-density-selected-limit}.

\subsection{Order parameter and the occurrence of condensation}\label{sec:free-bose-order-parameter}

In this subsection, we prove that the order parameter detects condensation and that the scaled zero-mode occupation converges to an exponential law.

\begin{prop}[Density-selected fugacities and the zero-mode density]\label{prop:free-bose-density-selected-fugacity}
Let $\bar{\smnumberdensity}>0$.
\begin{enumerate}
\item For every $L>0$, there is exactly one $y_{L^d}>1$ with
$\fun{\smnumberdensity_{L^d}}{\sminvtemperature,y_{L^d}}=\bar{\smnumberdensity}$.
\item If
$\bar{\smnumberdensity}<\fun{\smnumberdensity_{\txtcritical}}{\sminvtemperature}$,
there is exactly one $y_{\infty}>1$
with
$\fun{\smnumberdensity}{\sminvtemperature,y_{\infty}}=\bar{\smnumberdensity}$;
otherwise put $y_{\infty}=1$.
Then $\lim_{L\to\infty}y_{L^d}=y_{\infty}$.
\item The zero-mode density converges to the condensate density:
\begin{equation}\label{eq:free-bose-zero-mode-density-limit}
\begin{aligned}
\lim_{L\to\infty}\frac{1}{L^d}\fun{N_0}{y_{L^d}}
&=
\fun{\smnumberdensity_0}{\sminvtemperature}.
\end{aligned}
\end{equation}
\end{enumerate}
\end{prop}

\begin{proof}
Write
$\fun{\smnumberdensity_{L^d}}{\sminvtemperature,y}=L^{-d}\fun{N_0}{y}+\fun{R_L}{y}$
with the nonzero-momentum part
$R_L$.
The bound
$\rbk{y\napiernum^{\sminvtemperature\fun{\faftr{\physham[h]}}{k}}-1}^{-1}\leq\rbk{y-1}^{-1}\napiernum^{-\sminvtemperature\fun{\faftr{\physham[h]}}{k}}$
for $y>1$ gives continuity and strict monotonicity of
$\fun{\smnumberdensity_{L^d}}{\sminvtemperature,\cdot}$ on
$\openinterval{1}{\infty}$ with limits $\infty$ and $0$ at the ends, which gives
(1).
The summands are radial nonincreasing functions of $k$, so the Riemann sums
$\fun{R_L}{y}$ converge to
$\fun{\smnumberdensity}{\sminvtemperature,y}$ for $y>1$, and also for $y=1$ if
the critical density is finite.
Along a subsequence, $y_{L^d}$ converges to some $y_{\ast}\geq1$, since
$\fun{R_L}{Y}<\bar{\smnumberdensity}/2$ and
$L^{-d}\rbk{Y-1}^{-1}<\bar{\smnumberdensity}/2$ for a large $Y$ and all large
$L$.
If $y_{\ast}>1$, then $L^{-d}\fun{N_0}{y_{L^d}}\to0$ and
$\bar{\smnumberdensity}=\fun{\smnumberdensity}{\sminvtemperature,y_{\ast}}<\fun{\smnumberdensity_{\txtcritical}}{\sminvtemperature}$
by monotonicity.
If $y_{\ast}=1$, then, for every $\epsilon>0$, eventually
$\fun{R_L}{y_{L^d}}\geq\fun{R_L}{1+\epsilon}$, and letting $L\to\infty$ and then
$\epsilon\downarrow0$ gives
$\fun{\smnumberdensity_{\txtcritical}}{\sminvtemperature}\leq\liminf_{L\to\infty}\fun{R_L}{y_{L^d}}\leq\bar{\smnumberdensity}<\infty$.
The Riemann-sum limit at $y=1$ therefore applies, and the bounds
$\fun{R_L}{1+\epsilon}\leq\fun{R_L}{y_{L^d}}\leq\fun{R_L}{1}$ give
$\fun{R_L}{y_{L^d}}\to\fun{\smnumberdensity_{\txtcritical}}{\sminvtemperature}$
and
$L^{-d}\fun{N_0}{y_{L^d}}\to\bar{\smnumberdensity}-\fun{\smnumberdensity_{\txtcritical}}{\sminvtemperature}$.
The two cases are determined by $\bar{\smnumberdensity}$, and hence the whole
sequences converge to the limits in (2) and
(3).
\end{proof}

\begin{prop}[Zero-mode occupation detected by the order parameter]\label{prop:free-bose-zero-mode-resolvent-probes}
Let $\bar{\smnumberdensity}>0$, and assume that either $y_{\infty}>1$ or the
limit
$\lim_{L\to\infty}L^{-d}\fun{N_0}{y_{L^d}}=\fun{\smnumberdensity_0}{\sminvtemperature}>0$
exists.
Then the following are equivalent.
\begin{enumerate}
\item Bose--Einstein condensation occurs, that is,
$\fun{\smnumberdensity_0}{\sminvtemperature}>0$.
\item The order parameter satisfies $\mathsf{o}_{\txtbec,0}=0$.
\item The order parameter satisfies $\mathsf{o}_{\txtbec,1}<1$.
\end{enumerate}
If $\fun{N_0}{y_{L^d}}\to\infty$ and $L^{-d}\fun{N_0}{y_{L^d}}\to0$, then
$\mathsf{o}_{\txtbec,0}=0$ and $\mathsf{o}_{\txtbec,1}=1$.
\end{prop}

\begin{proof}
Only the zero coefficient of $\mathsf{b}_{\#,L}$ is nonzero, and
$\fun{\coth}{-\frac{\sminvtemperature\smchemicalpotential_{L^d}}{2}}=\frac{y_{L^d}+1}{y_{L^d}-1}=2\fun{N_0}{y_{L^d}}+1$,
and hence
Proposition \ref{prop:free-bose-finite-volume-resolvent-expectations} gives
$$\begin{aligned}
\imunit\fun{\oastate[\psi_{\txtgrandcanonical,I_L^d,\sminvtemperature,\smchemicalpotential_{L^d}}]}{\fun{\oaresolvent}{1,\mathsf{b}_{\#,L}}}
&=
\int_0^\infty
\fnexp{-t-\frac{t^{2}}{4}L^{-d\#}\rbk{2\fun{N_0}{y_{L^d}}+1}}
\opdmsr{t}.
\end{aligned}$$
If $y_{\infty}>1$, then $\fun{N_0}{y_{L^d}}$ stays bounded, and dominated
convergence gives $\mathsf{o}_{\txtbec,0}>0$ and
$\mathsf{o}_{\txtbec,1}=1$.
If $L^{-d}\fun{N_0}{y_{L^d}}\to\fun{\smnumberdensity_0}{\sminvtemperature}>0$, it
gives
$\mathsf{o}_{\txtbec,0}=0$ and
$\mathsf{o}_{\txtbec,1}=\int_0^\infty\napiernum^{-t-\fun{\smnumberdensity_0}{\sminvtemperature}t^{2}/2}\opdmsr{t}<1$.
The last assertion is the same computation.
\end{proof}

The criterion \(\mathsf{o}_{\txtbec,1}<1\) detects macroscopic condensation, whereas \(\mathsf{o}_{\txtbec,0}=0\) detects only the divergence of the zero-mode occupation.

For \(L>0\) and \(y>1\), the grand-canonical occupation law \(\msrprb_{L,y}\) is the product over \(k\in\setlattice_L^d\) of the geometric laws of the occupation numbers \(n_k\) with the ratios \(y^{-1}\napiernum^{-\sminvtemperature\fun{\faftr{\physham[h]}}{k}}\) \cite[Proposition 5.2.27]{BratteliRobinson004}.

\begin{prop}[Zero-mode occupation law and its scaling limit]\label{prop:free-zero-mode-geometric-limit}
For $L>0$ and $y>1$, the zero-mode occupation has the law
\begin{equation}\label{eq:free-zero-mode-geometric-law}
\begin{aligned}
\fun{\msrprb_{L,y}}{\set{\mathbf{n}}{n_0=m}}
&=
\rbk{1-y^{-1}}y^{-m},
\quad
m\in\monnat,
\end{aligned}
\end{equation}
with mean $\fun{N_0}{y}$.
If
$\bar{\smnumberdensity}>\fun{\smnumberdensity_{\txtcritical}}{\sminvtemperature}$
and $L_n\to\infty$, then the image of
$\msrprb_{L_n,y_{L_n^d}}$ under $\mathbf{n}\mapsto L_n^{-d}n_0$
converges weakly to the exponential law of mean
$\fun{\smnumberdensity_0}{\sminvtemperature}$, which is the law of
$\fun{\smnumberdensity_0}{\sminvtemperature}r$ under
$\msrprb_{\txtstd,\fldcmp}$ by \eqref{eq:free-bose-radial-phase-law}.
\end{prop}

\begin{proof}
The first statement is the geometric law of one mode.
The characteristic function of the image law is
$\frac{L_n^d\rbk{y_{L_n^d}-1}}{L_n^d\rbk{y_{L_n^d}-1}+L_n^d\rbk{1-\napiernum^{\imunit u/L_n^d}}}$,
which converges to
$\rbk{1-\imunit u\fun{\smnumberdensity_0}{\sminvtemperature}}^{-1}$ by
\eqref{eq:free-bose-zero-mode-density-limit}, and Lévy's
continuity theorem gives the weak convergence.
\end{proof}

\subsection{Component states and the BEC state}\label{sec:free-bose-component-states}

In this subsection, we construct the component states, compute the correlations of the component states and of the BEC state, and prove their regularity and KMS property.

\begin{lem}[Quasi-free states with constant mean]\label{lem:free-bose-shifted-quasifree-states}
Let $X$ be a complex linear subspace of $\sphilbfrak{H}_{\sminvtemperature}$ and
$\ell\colon X\to\fldreal$ real linear.
Then there is exactly one state
$\oastate[\psi_{\txtnonzero,\sminvtemperature,\ell}]$ of
$\fun{\oaresolventalgebra}{X,\opform{\sigma}}$ whose correlations are given by
\eqref{eq:free-bose-shifted-quasifree-resolvent} with $\ell$ in place of
$\ell_{\sminvtemperature,r,\theta}$, and it is
regular.
For a GNS triple
$\pairbk{\sphilbfrak{H}_{\txtnonzero,\sminvtemperature,0},\oarepn_{\txtnonzero,\sminvtemperature,0},\oagnsvector[\Psi_{\txtnonzero,\sminvtemperature,0}]}$
of $\oastate[\psi_{\txtnonzero,\sminvtemperature,0}]$, the assignment
\begin{equation}\label{eq:free-bose-shifted-gns-representation}
\begin{aligned}
\fun{\oarepn_{\txtnonzero,\sminvtemperature,\ell}}{\fun{\oaresolvent}{\lambda,f}}
&=
\rbk{\imunit\lambda-\fun{\phi_{\oarepn_{\txtnonzero,\sminvtemperature,0}}}{f}-\fun{\ell}{f}}^{-1}
\end{aligned}
\end{equation}
defines a regular representation with the cyclic vector
$\oagnsvector[\Psi_{\txtnonzero,\sminvtemperature,0}]$ and the
vector state $\oastate[\psi_{\txtnonzero,\sminvtemperature,\ell}]$, and
$\fun{\oarepn_{\txtnonzero,\sminvtemperature,\ell}}{A}$
depends norm continuously on $\ell$ for the topology of pointwise convergence on
$X$, for every
$A\in\fun{\oaresolventalgebra}{X,\opform{\sigma}}$.
\end{lem}

\begin{proof}
The sesquilinear extension $\opform{s}$ of
$\frac{1}{2}\opreal\fun{\opform{q}_{\txtnonzero,\sminvtemperature}}{u,v}+\frac{\imunit}{2}\fun{\opform{\sigma}}{u,v}$
to
the complexification of the real vector space $X$ is positive semidefinite,
since
$\abs{\fun{\opform{\sigma}}{u,v}}\leq\fun{\opform{q}_{\txtnonzero,\sminvtemperature}}{u}^{\frac{1}{2}}\fun{\opform{q}_{\txtnonzero,\sminvtemperature}}{v}^{\frac{1}{2}}$
by $K_{\sminvtemperature}\geq1$, and it satisfies
$\fun{\opform{s}}{f,g}-\fun{\opform{s}}{g,f}=\imunit\fun{\opform{\sigma}}{f,g}$.
Hence there is a regular quasi-free state for $\opform{s}$
\cite[p.~2736 and Remark 4.12]{BuchholzGrundling002}; its
values on products of resolvents with positive parameters are given by
\cite[Eq.~(17), p.~2736]{BuchholzGrundling002},
which is \eqref{eq:free-bose-shifted-quasifree-resolvent} with $\ell=0$.
The same formula for parameters of either sign follows by analytic continuation
in each parameter and the relation
$\fun{\oaresolvent}{-z,f}=-\fun{\oaresolvent}{z,-f}$
\cite[Eq.~(10), p.~2731]{BuchholzGrundling002}; on the domain of
integration, the modulus of the integrand is at most
$\fnexp{-\sum_j\abs{\lambda_jt_j}}$.
Since products of resolvents span a dense subspace, the state is unique.
For general $\ell$, the shift automorphism $\Theta_{\ell}$ with
$\fun{\Theta_{\ell}}{\fun{\oaresolvent}{\lambda,f}}=\fun{\oaresolvent}{\lambda+\imunit\fun{\ell}{f},f}$
\cite[Proposition 3.7]{BuchholzGrundling002} gives
$\oarepn_{\txtnonzero,\sminvtemperature,\ell}=\oarepn_{\txtnonzero,\sminvtemperature,0}\circ\Theta_{\ell}$
and
$\oastate[\psi_{\txtnonzero,\sminvtemperature,\ell}]=\oastate[\psi_{\txtnonzero,\sminvtemperature,0}]\circ\Theta_{\ell}$,
whose correlations are those at the parameters
$\lambda_j+\imunit\fun{\ell}{f_j}$.
The values of \eqref{eq:free-bose-shifted-gns-representation} are resolvents of
self-adjoint operators, which gives the
regularity, and they depend norm continuously on $\fun{\ell}{f}$, which gives
the continuity on the dense
$\ast$-subalgebra and hence on the whole algebra.
\end{proof}

By Lemma \ref{lem:free-bose-resolvent-inclusions}, the isometric \(\ast\)-isomorphism \(\Phi_{\sphilbfrak{D}_{0,\sminvtemperature}}\) of \(\fun{\oaresolventalgebra}{\sphilbfrak{D}_{0,\sminvtemperature},\opform{\sigma}}\) onto \(\fun{\oaresolventalgebra}{\sphilbfrak{D}_{0,\sminvtemperature}}\) identifies the states of Lemma \ref{lem:free-bose-shifted-quasifree-states} for \(X=\sphilbfrak{D}_{0,\sminvtemperature}\) with states of \(\fun{\oaresolventalgebra}{\sphilbfrak{D}_{0,\sminvtemperature}}\); the component state is \(\oastate[\psi_{\sminvtemperature,r,\theta}]=\oastate[\psi_{\txtnonzero,\sminvtemperature,\ell_{\sminvtemperature,r,\theta}}]\circ\Phi_{\sphilbfrak{D}_{0,\sminvtemperature}}^{-1}\). We fix a GNS triple \(\pairbk{\sphilbfrak{H}_{\txtnonzero,\sminvtemperature,0},\oarepn_{\txtnonzero,\sminvtemperature,0},\oagnsvector[\Psi_{\txtnonzero,\sminvtemperature,0}]}\) of the state with \(\ell=0\) on \(\fun{\oaresolventalgebra}{\sphilbfrak{D}_{0,\sminvtemperature}}\), and the component representation is \(\oarepn_{\txtnonzero,\sminvtemperature,r,\theta}=\oarepn_{\txtnonzero,\sminvtemperature,\ell_{\sminvtemperature,r,\theta}}\circ\Phi_{\sphilbfrak{D}_{0,\sminvtemperature}}^{-1}\), whose vector state at \(\oagnsvector[\Psi_{\txtnonzero,\sminvtemperature,0}]\) is \(\oastate[\psi_{\sminvtemperature,r,\theta}]\).

\begin{lem}[Polar representation of the zero-mode form]\label{lem:free-bose-zero-mode-polar-representation}
For every $f\in\sphilbfrak{D}_{0,\sminvtemperature}$,
\begin{equation}\label{eq:free-bose-zero-mode-polar-representation}
\begin{aligned}
\fnexp{-\frac{1}{4}\fun{\opform{q}_0}{f}}
&=
\sqfun{\prbexp_{\msrprb_{\txtstd,\fldcmp}(r,\theta)}}{\napiernum^{\imunit\fun{\ell_{\sminvtemperature,r,\theta}}{f}}}.
\end{aligned}
\end{equation}
\end{lem}

\begin{proof}
Under $\msrprb_{\txtstd,\fldcmp}$, the real part of $w\fun{\faftr{f}}{0}$ is
centered Gaussian with variance
$\frac{1}{2}\abs{\fun{\faftr{f}}{0}}^{2}$, and
$\fun{\ell_{\sminvtemperature,r,\theta}}{f}$ is
$\sqrt{2\fun{\smnumberdensity_0}{\sminvtemperature}}$ times this variable for
$w=\sqrt{r}\napiernum^{\imunit\theta}$.
\end{proof}

\begin{prop}[Correlations, regularity, and KMS property of the BEC components and the BEC state]\label{prop:free-bose-resolvent-component-correlations}
The following statements hold.
\begin{enumerate}
\item The component representations
$\oarepn_{\txtnonzero,\sminvtemperature,r,\theta}$ are regular, and the states
$\oastate[\psi_{\sminvtemperature,r,\theta}]$ and
$\oastate[\psi_{\txtnonzero,\sminvtemperature}]$ are regular.
\item For $n\in\semigrposint$,
$\lambda_1,\ldots,\lambda_n\in\fldreal\setminus\setone{0}$, and
$f_1,\ldots,f_n\in\sphilbfrak{D}_{0,\sminvtemperature}$,
\begin{equation}\label{eq:free-bose-bec-resolvent-correlation}
\begin{aligned}
&\fun{\oastate[\psi_{\txtbec,\sminvtemperature}]}{\fun{\oaresolvent}{\lambda_1,f_1}\cdots\fun{\oaresolvent}{\lambda_n,f_n}}
\\ 
&=\rbk{-\imunit}^{n}\int_0^{\rbk{\sgn\lambda_1}\infty}\cdots\int_0^{\rbk{\sgn\lambda_n}\infty}
\fnexp{-\frac{\imunit}{2}\sum_{\rbk{j,k}\in\fun{\Sigma_{2,<}}{\setone{1,\ldots,n}}}t_jt_k\fun{\opform{\sigma}}{f_j,f_k}-\sum_{j=1}^{n}\lambda_jt_j}
\\ 
&\quad\times
\fnexp{-\frac{1}{4}\fun{\opform{q}_{0}}{\sum_{j=1}^{n}t_jf_j}-\frac{1}{4}\fun{\opform{q}_{\txtnonzero,\sminvtemperature}}{\sum_{j=1}^{n}t_jf_j}}
\opdmsr{t_n}\cdots\opdmsr{t_1}.
\end{aligned}
\end{equation}
\item The BEC state is the quasi-free state for the sesquilinear extension of
$\frac{1}{2}\opreal\fun{\opform{q}_{\txtnonzero,\sminvtemperature}}{u,v}+\frac{1}{2}\fun{\opform{q}_0}{u,v}+\frac{\imunit}{2}\fun{\opform{\sigma}}{u,v}$;
it is regular, and its GNS representation $\oarepn_{\txtbec,\sminvtemperature}$
is faithful.
\item Let $\sphilbfrak{D}_{U}\subset\sphilbfrak{D}_{0,\sminvtemperature}$ be a
complex linear subspace preserved by
$\napiernum^{\imunit t\physham[h]}$, $t\in\fldreal$.
Every component state restricted to
$\fun{\oaresolventalgebra}{\sphilbfrak{D}_{U}}$ satisfies the KMS boundary
relation at
the inverse temperature $\sminvtemperature$ for $\alpha_{\txtfr}$, that is, for
$A,B$ in this algebra there is a bounded
continuous function $F_{A,B}$ on
$\set{z\in\fldcmp}{-\sminvtemperature\leq\opimag z\leq0}$, analytic in the
interior, with
$\fun{F_{A,B}}{t}=\fun{\oastate[\psi_{\sminvtemperature,r,\theta}]}{\fun{\alpha_{\txtfr,t}}{A}B}$
and
$\fun{F_{A,B}}{t-\imunit\sminvtemperature}=\fun{\oastate[\psi_{\sminvtemperature,r,\theta}]}{B\fun{\alpha_{\txtfr,t}}{A}}$.
Component states whose functionals $\ell_{\sminvtemperature,r,\theta}$ differ on
$\sphilbfrak{D}_{U}$ are different, and the
restriction of the BEC state is their mixture over $\msrprb_{\txtstd,\fldcmp}$
and satisfies the same boundary relation.
\end{enumerate}
\end{prop}

\begin{proof}
Item (1) is Lemma \ref{lem:free-bose-shifted-quasifree-states}.
For (2), the mean term of \eqref{eq:free-bose-shifted-quasifree-resolvent} is
$\fnexp{-\imunit\fun{\ell_{\sminvtemperature,r,\theta}}{\sum_jt_jf_j}}$, and
Fubini's theorem with the bound of the integrand
and Lemma \ref{lem:free-bose-zero-mode-polar-representation} give
\eqref{eq:free-bose-bec-resolvent-correlation}.
Item (3) follows from (2) and the uniqueness of quasi-free states in the proof
of
Lemma \ref{lem:free-bose-shifted-quasifree-states}, and the faithfulness follows
from the regularity
\cite[Theorem 4.10]{BuchholzGrundling002}.

For (4), the proof uses only the correlation formula
\eqref{eq:free-bose-shifted-quasifree-resolvent}.
Since $\fun{\faftr{\napiernum^{\imunit t\physham[h]}f}}{0}=\fun{\faftr{f}}{0}$,
the functional $\ell_{\sminvtemperature,r,\theta}$ is invariant under
$\napiernum^{\imunit t\physham[h]}$, and so are $\opform{\sigma}$ and
$\opform{q}_{\txtnonzero,\sminvtemperature}$.
For $a,b\in\sphilbfrak{D}_U$, put $u_t=\napiernum^{\imunit t\physham[h]}a$ and
$$\begin{aligned}
\fun{E_{a,b}}{t}
&=
\opreal\fun{\opform{q}_{\txtnonzero,\sminvtemperature}}{u_t,b}+\imunit\fun{\opform{\sigma}}{u_t,b}
\\ 
&=
\frac{1}{2}\int_{\openinterval{0}{\infty}}\rbk{\fun{\kappa_{\sminvtemperature}}{\lambda}+1}\napiernum^{-\imunit t\lambda}\opdmsr{\bkt{a}{\msrprb_{\physham[h]}b}(\lambda)}
+\frac{1}{2}\int_{\openinterval{0}{\infty}}\rbk{\fun{\kappa_{\sminvtemperature}}{\lambda}-1}\napiernum^{\imunit t\lambda}\opdmsr{\bkt{b}{\msrprb_{\physham[h]}a}(\lambda)}.
\end{aligned}$$
For $\lambda>0$ and $0\leq s\leq\sminvtemperature$,
$\rbk{\fun{\kappa_{\sminvtemperature}}{\lambda}-1}\napiernum^{s\lambda}\leq\fun{\kappa_{\sminvtemperature}}{\lambda}+1$
and
$\rbk{\fun{\kappa_{\sminvtemperature}}{\lambda}+1}\napiernum^{-\sminvtemperature\lambda}=\fun{\kappa_{\sminvtemperature}}{\lambda}-1$.
Hence $E_{a,b}$ extends to a bounded continuous function on the strip
$-\sminvtemperature\leq\opimag z\leq0$, analytic in its interior, and
$\fun{E_{a,b}}{t-\imunit\sminvtemperature}$ is $\fun{E_{a,b}}{t}$ with $u_t$ and
$b$
exchanged.
The function
$$\begin{aligned}
\fun{G_{a,b}}{z}
&=
\fnexp{-\frac{1}{4}\fun{\opform{q}_{\txtnonzero,\sminvtemperature}}{a}-\frac{1}{4}\fun{\opform{q}_{\txtnonzero,\sminvtemperature}}{b}+\imunit\fun{\ell_{\sminvtemperature,r,\theta}}{a+b}-\frac{1}{2}\fun{E_{a,b}}{z}}
\end{aligned}$$
has modulus
$\fnexp{-\frac{1}{4}\fun{\opform{q}_{\txtnonzero,\sminvtemperature}}{u_t+b}}\leq1$
on both boundary lines and hence on the strip by the three-lines theorem.
For products $A$ and $B$ of resolvents of vectors of $\sphilbfrak{D}_U$,
\eqref{eq:free-bose-shifted-quasifree-resolvent} writes
$\fun{\oastate[\psi_{\sminvtemperature,r,\theta}]}{\fun{\alpha_{\txtfr,t}}{A}B}$
and
$\fun{\oastate[\psi_{\sminvtemperature,r,\theta}]}{B\fun{\alpha_{\txtfr,t}}{A}}$
as
integrals, against the Laplace weights of the parameters, of $G_{a,b}$ at $t$
and at
$t-\imunit\sminvtemperature$ times phases independent of $t$, with
$a=-\sum_js_jf_j$ and $b=-\sum_ks'_kg_k$; integrating $G_{a,b}$ itself gives the
boundary function on a dense $\ast$-subalgebra, and the three-lines theorem
extends it to all $A$ and $B$.
If the functionals differ at $f$, the one-point functions
$\fun{\oastate[\psi_{\sminvtemperature,r,\theta}]}{\fun{\oaresolvent}{\lambda,f}}$
differ as Laplace transforms.
The boundary functions of the components are bounded by $\norm{A}\norm{B}$, and
their integrals against
$\msrprb_{\txtstd,\fldcmp}$ give the boundary function of the mixture.
\end{proof}

\subsection{Order parameter and the center}\label{sec:free-bose-order-parameter-center}

In this subsection, we show that the resolvents of \(\mathsf{b}_{1,L}\) are asymptotically central, identify the strong limits in the BEC representation, realize the components in the left Araki--Woods representation, and prove the central decomposition. For Lemma \ref{lem:free-bose-zero-mode-symplectic-decay} and Propositions \ref{prop:free-bose-order-parameter-norm-one} and \ref{prop:free-bose-resolvent-central-probes}, assume \(0<s<d\); then \(\mathsf{b}_{\#,L}\in\sphilbfrak{D}_{0,\sminvtemperature}\), since \(\fun{\kappa_{\sminvtemperature}}{\fun{\faftr{\physham[h]}}{k}}\) is integrable near \(k=0\).

\begin{lem}[Decay of the zero-mode pairings]\label{lem:free-bose-zero-mode-symplectic-decay}
For $f\in\fun{\lp^1}{\fldreal^d}\cap\fun{\lp^2}{\fldreal^d}$,
$\lambda\in\fldreal\setminus\setone{0}$, and $\#=0,1$,
\begin{equation}\label{eq:free-bose-zero-mode-symplectic-decay}
\begin{aligned}
\abs{\fun{\opform{\sigma}}{\mathsf{b}_{\#,L},f}}
&\leq
L^{-\frac{\rbk{1+\#}d}{2}}\norm{f}_{\fun{\lp^1}{\fldreal^d}},
&
\norm{\commutator{\fun{\oaresolvent}{1,\mathsf{b}_{\#,L}}}{\fun{\oaresolvent}{\lambda,f}}}
&\leq
\frac{\abs{\fun{\opform{\sigma}}{\mathsf{b}_{\#,L},f}}}{\lambda^2}.
\end{aligned}
\end{equation}
Moreover, with the finite constant
$C_{s,d}=\int_{\fldreal^{d}}\fun{\faftr{\physham[h]}}{k}^{-1}\prod_{j=1}^{d}\rbk{\frac{\fun{\sin}{\pi k_j}}{\pi k_j}}^{2}\opdmsr{k}$,
\begin{equation}\label{eq:free-bose-zero-mode-thermal-form-decay}
\begin{aligned}
\fun{\opform{q}_{\txtnonzero,\sminvtemperature}}{\mathsf{b}_{1,L}}
&\leq
L^{-d}+\frac{2C_{s,d}}{\sminvtemperature}L^{s-d}.
\end{aligned}
\end{equation}
\end{lem}

\begin{proof}
The first bound is
$\abs{\bkt{\mathsf{b}_{\#,L}}{f}}\leq\norm{\mathsf{b}_{\#,L}}_{\infty}\norm{f}_{\lp^1}$.
The commutator relation
$$\begin{aligned}
\commutator{\fun{\oaresolvent}{1,g}}{\fun{\oaresolvent}{\lambda,f}}
&=
\imunit\fun{\opform{\sigma}}{g,f}\fun{\oaresolvent}{1,g}\fun{\oaresolvent}{\lambda,f}^{2}\fun{\oaresolvent}{1,g}
\end{aligned}$$
\cite[Definition 3.1]{BuchholzGrundling002} and
$\norm{\fun{\oaresolvent}{\lambda,f}}=\abs{\lambda}^{-1}$ give the second.
For the third,
$$\begin{aligned}
\fun{\faftr{\mathsf{b}_{1,L}}}{k}
&=
\prod_j\rbk{\pi Lk_j}^{-1}\fun{\sin}{\pi Lk_j},
\end{aligned}$$
$\fun{\coth}{x/2}\leq1+2x^{-1}$, and the substitution $k\mapsto L^{-1}k$ give
the bound.
The constant is finite, since $s<d$ near $k=0$ and the product decays like
$\prod_j\rbk{1+k_j^2}^{-1}$.
\end{proof}

\begin{prop}[Norm of the order-parameter resolvents]\label{prop:free-bose-order-parameter-norm-one}
For every $L>0$, $\norm{\fun{\oaresolvent}{1,\mathsf{b}_{0,L}}}=1$ in
$\fun{\oaresolventalgebra}{\sphilbfrak{D}_{0,\sminvtemperature}}$,
and, since $\oarepn_{\txtbec,\sminvtemperature}$ is faithful,
$\norm{\fun{\oarepn_{\txtbec,\sminvtemperature}}{\fun{\oaresolvent}{1,\mathsf{b}_{0,L}}}}=1$;
the family does not converge to
$0$ in norm in the BEC representation.
\end{prop}

\begin{proof}
This is $\norm{\fun{\oaresolvent}{\lambda,f}}=\abs{\lambda}^{-1}$ for $f\neq0$
\cite[Theorem 3.6(iii)]{BuchholzGrundling002} and
Proposition \ref{prop:free-bose-resolvent-component-correlations}(3).
\end{proof}

Thus, although \(\mathsf{o}_{\txtbec,0}=0\) under condensation, the resolvents \(\fun{\oaresolvent}{1,\mathsf{b}_{0,L}}\) do not vanish; the order parameter is the expectation, not the norm.

\begin{prop}[Asymptotically central order-parameter resolvents]\label{prop:free-bose-resolvent-central-probes}
The following statements hold.
\begin{enumerate}
\item For $\#=0,1$ and every
$A\in\fun{\oaresolventalgebra}{\sphilbfrak{D}_{0,\sminvtemperature}}$,
$\lim_{L\to\infty}\norm{\commutator{\fun{\oaresolvent}{1,\mathsf{b}_{\#,L}}}{A}}=0$.
\item The value
$\fun{\ell_{\sminvtemperature,r,\theta}}{\mathsf{b}_{1,L}}=\sqrt{2\fun{\smnumberdensity_0}{\sminvtemperature}r}\fun{\cos}{\theta}$
does not depend on $L$; let
\begin{equation}\label{eq:free-bose-central-probe-functions}
\begin{aligned}
\fun{F}{r,\theta}
&=
\rbk{\imunit-\sqrt{2\fun{\smnumberdensity_0}{\sminvtemperature}r}\fun{\cos}{\theta}}^{-1},
&
\fun{G}{r,\theta}
&=
\rbk{\imunit+\sqrt{2\fun{\smnumberdensity_0}{\sminvtemperature}r}\fun{\sin}{\theta}}^{-1}.
\end{aligned}
\end{equation}
Then, for every order coordinate $\pairbk{r,\theta}$ and as $L,L'\to\infty$,
\begin{equation}\label{eq:free-bose-central-probe-state-limits}
\begin{aligned}
\fun{\oastate[\psi_{\sminvtemperature,r,\theta}]}{\faadj{\rbk{\fun{\oaresolvent}{1,\mathsf{b}_{1,L}}-\fun{F}{r,\theta}}}\rbk{\fun{\oaresolvent}{1,\mathsf{b}_{1,L}}-\fun{F}{r,\theta}}}
&\to
0,
\\ 
\fun{\oastate[\psi_{\txtbec,\sminvtemperature}]}{\faadj{\rbk{\fun{\oaresolvent}{1,\mathsf{b}_{1,L}}-\fun{\oaresolvent}{1,\mathsf{b}_{1,L'}}}}\rbk{\fun{\oaresolvent}{1,\mathsf{b}_{1,L}}-\fun{\oaresolvent}{1,\mathsf{b}_{1,L'}}}}
&\to
0,
\end{aligned}
\end{equation}
and the same holds for $\imunit\mathsf{b}_{1,L}$ and $G$.
\item Let $\oagnsvector[\Psi_{\txtbec,\sminvtemperature}]$ be the cyclic vector
of $\oarepn_{\txtbec,\sminvtemperature}$ and
$\oa{M}_{\txtbec,\sminvtemperature}=\oadoublecommutant{\fun{\oarepn_{\txtbec,\sminvtemperature}}{\fun{\oaresolventalgebra}{\sphilbfrak{D}_{0,\sminvtemperature}}}}$.
Then
$\fun{\oarepn_{\txtbec,\sminvtemperature}}{\fun{\oaresolvent}{1,\mathsf{b}_{1,L}}}$
converges strongly to an element $Z$
of the center of $\oa{M}_{\txtbec,\sminvtemperature}$, and
\begin{equation}\label{eq:free-bose-central-probe-limit}
\begin{aligned}
\bkt{\fun{\oarepn_{\txtbec,\sminvtemperature}}{B}\oagnsvector[\Psi_{\txtbec,\sminvtemperature}]}
{Z\fun{\oarepn_{\txtbec,\sminvtemperature}}{C}\oagnsvector[\Psi_{\txtbec,\sminvtemperature}]}
&=
\sqfun{\prbexp_{\msrprb_{\txtstd,\fldcmp}(r,\theta)}}{\fun{F}{r,\theta}\fun{\oastate[\psi_{\sminvtemperature,r,\theta}]}{\faadj{B}C}}
\end{aligned}
\end{equation}
for $B,C\in\fun{\oaresolventalgebra}{\sphilbfrak{D}_{0,\sminvtemperature}}$.
The convergence is not in norm:
$\norm{\fun{\oaresolvent}{1,\mathsf{b}_{1,L}}-\fun{\oaresolvent}{1,\mathsf{b}_{1,L'}}}\geq1$
for $L\neq L'$.
\end{enumerate}
\end{prop}

\begin{proof}
Item (1) follows from Lemma \ref{lem:free-bose-zero-mode-symplectic-decay} on
the generators and from the bound
$\norm{\fun{\oaresolvent}{1,\mathsf{b}_{\#,L}}}\leq1$.

For (2), since $\mathsf{b}_{1,L}$ and $\mathsf{b}_{1,L'}$ are real,
$\fun{\opform{\sigma}}{\mathsf{b}_{1,L},\mathsf{b}_{1,L'}}=0$,
and the restriction of the component state to the commutative algebra generated
by their resolvents is the expectation of
functions of a Gaussian vector of $\fldreal^{2}$ with mean
$\pairbk{\fun{\ell_{\sminvtemperature,r,\theta}}{\mathsf{b}_{1,L}},\fun{\ell_{\sminvtemperature,r,\theta}}{\mathsf{b}_{1,L'}}}$
and covariance $\frac{1}{2}\opreal\opform{q}_{\txtnonzero,\sminvtemperature}$:
since $\opform{\sigma}$ vanishes on these vectors,
\eqref{eq:free-bose-shifted-quasifree-resolvent} and
$\rbk{\imunit\lambda-x}^{-1}=-\imunit\int_0^{\rbk{\sgn\lambda}\infty}\napiernum^{-\lambda t-\imunit xt}\opdmsr{t}$
identify the correlations of products of these resolvents with Gaussian
expectations of products of such functions.
The inequality $\abs{\rbk{\imunit-a}^{-1}-\rbk{\imunit-b}^{-1}}\leq\abs{a-b}$
for $a,b\in\fldreal$ bounds the left-hand
sides of \eqref{eq:free-bose-central-probe-state-limits} by
$\frac{1}{2}\fun{\opform{q}_{\txtnonzero,\sminvtemperature}}{\mathsf{b}_{1,L}}$
and by
$\fun{\opform{q}_{\txtnonzero,\sminvtemperature}}{\mathsf{b}_{1,L}}+\fun{\opform{q}_{\txtnonzero,\sminvtemperature}}{\mathsf{b}_{1,L'}}$,
uniformly in $\pairbk{r,\theta}$, and
\eqref{eq:free-bose-zero-mode-thermal-form-decay} gives the limits.
The case of $\imunit\mathsf{b}_{1,L}$ is the same, with
$\fun{\ell_{\sminvtemperature,r,\theta}}{\imunit\mathsf{b}_{1,L}}=-\sqrt{2\fun{\smnumberdensity_0}{\sminvtemperature}r}\fun{\sin}{\theta}$.

For (3), the second limit of \eqref{eq:free-bose-central-probe-state-limits} and
(1) give
$$\begin{aligned}
&\norm{\fun{\oarepn_{\txtbec,\sminvtemperature}}{\fun{\oaresolvent}{1,\mathsf{b}_{1,L}}-\fun{\oaresolvent}{1,\mathsf{b}_{1,L'}}}\fun{\oarepn_{\txtbec,\sminvtemperature}}{A}\oagnsvector[\Psi_{\txtbec,\sminvtemperature}]}
\\ 
&\leq
\norm{\commutator{\fun{\oaresolvent}{1,\mathsf{b}_{1,L}}-\fun{\oaresolvent}{1,\mathsf{b}_{1,L'}}}{A}}
+\norm{A}\norm{\fun{\oarepn_{\txtbec,\sminvtemperature}}{\fun{\oaresolvent}{1,\mathsf{b}_{1,L}}-\fun{\oaresolvent}{1,\mathsf{b}_{1,L'}}}\oagnsvector[\Psi_{\txtbec,\sminvtemperature}]}
\to
0,
\end{aligned}$$
and the uniform bound gives the strong convergence to some
$Z\in\oa{M}_{\txtbec,\sminvtemperature}$, which commutes with
$\oa{M}_{\txtbec,\sminvtemperature}$ by (1).
The first limit of \eqref{eq:free-bose-central-probe-state-limits} gives the
strong convergence to $\fun{F}{r,\theta}$ in
the GNS representation of each component, and dominated convergence in
\eqref{eq:free-bose-resolvent-state-mixture} gives
\eqref{eq:free-bose-central-probe-limit}.
The last assertion is \cite[Theorem 5.3(ii)]{BuchholzGrundling002}, since
$\mathsf{b}_{1,L'}\notin\fldreal\mathsf{b}_{1,L}$,
together with the isometry of the faithful representation.
\end{proof}

The following Araki--Woods construction and Lemma \ref{lem:free-bose-araki-woods-components} apply for every \(s>0\). The components are realized in the ordinary left Araki--Woods representation over the complex Hilbert space \(\sphilbfrak{h}=\fun{\lp^2}{\fldreal^d;\fldcmp}\) \cite[Section 17.1.5]{DerezinskiGerard001}. With the Planck density \(\smnumberdensity_{\sminvtemperature}=\rbk{\napiernum^{\sminvtemperature\physham[h]}-1}^{-1}\), set \begin{equation}\label{eq:free-bose-araki-woods-data}
\begin{aligned}
\sphilbfrak{K}_{\txtarakiwoods}
&=
\sphilbfrak{h}\oplus\cmpconj{\sphilbfrak{h}},
\\ 
\fun{\ell_{\sminvtemperature}}{f}
&=
\rbk{1+\smnumberdensity_{\sminvtemperature}}^{\frac{1}{2}}f
\oplus\cmpconj{\smnumberdensity_{\sminvtemperature}^{\frac{1}{2}}f},
\quad
f\in\dom\smnumberdensity_{\sminvtemperature}^{\frac{1}{2}},
\end{aligned}
\end{equation} where \(\cmpconj{x}\) denotes a vector \(x\) regarded as a vector of the conjugate space. The left field is \(\fun{\phi_{\txtleft,\sminvtemperature}}{f}=\fun{\phi}{\fun{\ell_{\sminvtemperature}}{f}}\) for the Segal field \(\phi\) on the bosonic Fock space \(\fun{\spfockgamma_{\txtbsn}}{\sphilbfrak{K}_{\txtarakiwoods}}\) with the vacuum \(\oagnsvector[\Psi]_{\txtarakiwoods}\), the left Weyl operators are \(\napiernum^{\imunit\fun{\phi_{\txtleft,\sminvtemperature}}{f}}\), and the Liouvillean is \(\physliouvilean_{\txtarakiwoods}=\fun{\opfocksndqntdiff}{\physham[h]\oplus-\cmpconj{\physham[h]}}\) \cite[Sections 17.1.5 and 17.1.7]{DerezinskiGerard001}. Since \(1+2\smnumberdensity_{\sminvtemperature}=K_{\sminvtemperature}\), \(\dom\smnumberdensity_{\sminvtemperature}^{\frac{1}{2}}=\sphilbfrak{H}_{\sminvtemperature}\). This representation is not the one of \cite[Definition 5.1]{YoshitsuguSekine018}, which is built on the complexification of the realification of \(\sphilbfrak{h}\); the two agree on the real subspace \(\fun{\lp^2}{\fldreal^d;\fldreal}\), whose complexification is \(\sphilbfrak{h}\).

\begin{lem}[Araki--Woods realization of the BEC components]\label{lem:free-bose-araki-woods-components}
For an order coordinate $\pairbk{r,\theta}$,
$f\in\sphilbfrak{D}_{0,\sminvtemperature}$, and
$\lambda\in\fldreal\setminus\setone{0}$, set
\begin{equation}\label{eq:free-bose-shifted-resolvent-representation}
\begin{aligned}
\fun{\oarepn_{\sminvtemperature,r,\theta}}{\fun{\oaresolvent}{\lambda,f}}
&=
\rbk{\imunit\lambda-\fun{\phi_{\txtleft,\sminvtemperature}}{f}-\fun{\ell_{\sminvtemperature,r,\theta}}{f}}^{-1}.
\end{aligned}
\end{equation}
\begin{enumerate}
\item These assignments define regular representations of
$\fun{\oaresolventalgebra}{\sphilbfrak{D}_{0,\sminvtemperature}}$,
the component state $\oastate[\psi_{\sminvtemperature,r,\theta}]$ is the vector
state of $\oagnsvector[\Psi]_{\txtarakiwoods}$
in $\oarepn_{\sminvtemperature,r,\theta}$, and
$\oastate[\psi_{\txtnonzero,\sminvtemperature}]$ is the vector state of
$\oagnsvector[\Psi]_{\txtarakiwoods}$ in the left Araki--Woods representation of
$\oaresolventalgebra$.
\item For every order coordinate,
$\oadoublecommutant{\fun{\oarepn_{\sminvtemperature,r,\theta}}{\fun{\oaresolventalgebra}{\sphilbfrak{D}_{0,\sminvtemperature}}}}$
is the von Neumann algebra generated by the left Araki--Woods Weyl operators
$\napiernum^{\imunit\fun{\phi_{\txtleft,\sminvtemperature}}{f}}$,
$f\in\dom\smnumberdensity_{\sminvtemperature}^{\frac{1}{2}}$, which is a factor
with the cyclic vector
$\oagnsvector[\Psi]_{\txtarakiwoods}$.
There is a unitary operator $U$ from
$\sphilbfrak{H}_{\txtnonzero,\sminvtemperature,0}$ onto
$\fun{\spfockgamma_{\txtbsn}}{\sphilbfrak{K}_{\txtarakiwoods}}$ with
$U\oagnsvector[\Psi_{\txtnonzero,\sminvtemperature,0}]=\oagnsvector[\Psi]_{\txtarakiwoods}$
and
$U\fun{\oarepn_{\txtnonzero,\sminvtemperature,r,\theta}}{A}\faadj{U}=\fun{\oarepn_{\sminvtemperature,r,\theta}}{A}$
for every
order coordinate and every $A$.
\end{enumerate}
\end{lem}

\begin{proof}
For (1),
$$\begin{aligned}
\norm{\fun{\ell_{\sminvtemperature}}{f}}^{2}
&=
\norm{\rbk{1+\smnumberdensity_{\sminvtemperature}}^{\frac{1}{2}}f}^{2}+\norm{\smnumberdensity_{\sminvtemperature}^{\frac{1}{2}}f}^{2}
\\ 
&=
\fun{\opform{q}_{\txtnonzero,\sminvtemperature}}{f}
\end{aligned}$$
for the left one-particle map $\ell_{\sminvtemperature}$ of
\eqref{eq:free-bose-araki-woods-data}, and
$\opimag\bkt{\fun{\ell_{\sminvtemperature}}{f}}{\fun{\ell_{\sminvtemperature}}{g}}=\fun{\opform{\sigma}}{f,g}$.
Since $\ell_{\sminvtemperature,r,\theta}$ is real linear, the unitary groups
generated by
$\fun{\phi_{\txtleft,\sminvtemperature}}{f}+\fun{\ell_{\sminvtemperature,r,\theta}}{f}$
satisfy the Weyl relations, and their
vacuum expectations are
$\fnexp{-\frac{t^2}{4}\fun{\opform{q}_{\txtnonzero,\sminvtemperature}}{f}+\imunit t\fun{\ell_{\sminvtemperature,r,\theta}}{f}}$.
Writing each resolvent as a Laplace transform of these groups gives the
correlations
\eqref{eq:free-bose-shifted-quasifree-resolvent}, and uniqueness in Lemma
\ref{lem:free-bose-shifted-quasifree-states}
identifies the vector state.

For (2), the von Neumann algebra generated by $\rbk{\imunit-T}^{-1}$ for a
self-adjoint $T$ contains all bounded Borel
functions of $T$, and the resolvents of
$\fun{\phi_{\txtleft,\sminvtemperature}}{f}$ are norm limits of strong integrals
of
$\napiernum^{\imunit t\fun{\phi_{\txtleft,\sminvtemperature}}{f}}$; hence the
von Neumann algebra of $\oarepn_{\sminvtemperature,r,\theta}$
is generated by the left Weyl operators of
$f\in\sphilbfrak{D}_{0,\sminvtemperature}$, for every $\pairbk{r,\theta}$.
The space $\sphilbfrak{D}_{0,\sminvtemperature}$ is dense in
$\sphilbfrak{H}_{\sminvtemperature}$, since it contains the
Schwartz functions with Fourier transforms in
$\fun{C_{\mathrm{c}}^{\infty}}{\fldreal^{d}\setminus\setone{0}}$, and the Weyl
operators depend strongly continuously on $\fun{\ell_{\sminvtemperature}}{f}$;
hence this algebra is generated by the left
Weyl operators of $f\in\dom\smnumberdensity_{\sminvtemperature}^{\frac{1}{2}}$,
which is a factor with the cyclic vector
$\oagnsvector[\Psi]_{\txtarakiwoods}$
\cite[Theorem 17.24(5) and (6), p.~438]{DerezinskiGerard001}.
For $\pairbk{r,\theta}=\pairbk{0,0}$, the uniqueness of the GNS representation
gives $U$, and
$U\oarepn_{\txtnonzero,\sminvtemperature,r,\theta}\faadj{U}=\oarepn_{\sminvtemperature,r,\theta}$
follows from
\eqref{eq:free-bose-shifted-gns-representation} and
\eqref{eq:free-bose-shifted-resolvent-representation}, since $U$ maps
$\fun{\phi_{\oarepn_{\txtnonzero,\sminvtemperature,0}}}{f}$ to
$\fun{\phi_{\txtleft,\sminvtemperature}}{f}$.
\end{proof}

\begin{thm}[Direct integral and central decomposition of the BEC state]\label{thm:free-bose-resolvent-central-decomposition}
Assume $0<s<d$ and $\fun{\smnumberdensity_0}{\sminvtemperature}>0$, and let
$\oarepn_{\txtbec,\sminvtemperature}$,
$\oagnsvector[\Psi_{\txtbec,\sminvtemperature}]$,
$\oa{M}_{\txtbec,\sminvtemperature}$, and $Z$ be as in
Proposition \ref{prop:free-bose-resolvent-central-probes}(3), with $F$ of
\eqref{eq:free-bose-central-probe-functions}.
\begin{enumerate}
\item For every order coordinate,
$\oadoublecommutant{\fun{\oarepn_{\txtnonzero,\sminvtemperature,r,\theta}}{\fun{\oaresolventalgebra}{\sphilbfrak{D}_{0,\sminvtemperature}}}}=\oadoublecommutant{\fun{\oarepn_{\txtnonzero,\sminvtemperature,0}}{\fun{\oaresolventalgebra}{\sphilbfrak{D}_{0,\sminvtemperature}}}}$
is a factor, and $\oagnsvector[\Psi_{\txtnonzero,\sminvtemperature,0}]$ is
cyclic for $\oarepn_{\txtnonzero,\sminvtemperature,r,\theta}$.
For distinct order coordinates with $r,r'>0$, the representations
$\oarepn_{\txtnonzero,\sminvtemperature,r,\theta}$ and
$\oarepn_{\txtnonzero,\sminvtemperature,r',\theta'}$ are disjoint.
\item There is a unitary operator from the GNS space of
$\oastate[\psi_{\txtbec,\sminvtemperature}]$ onto
$$\begin{aligned}
\int_{\fldreal^{2}}^{\oplus}\sphilbfrak{H}_{\txtnonzero,\sminvtemperature,0}\opdmsr{\msrprb_{\txtstd,\fldcmp}(r,\theta)}
&=
\fun{\lp^{2}}{\fldreal^{2},\msrprb_{\txtstd,\fldcmp}}\otimes\sphilbfrak{H}_{\txtnonzero,\sminvtemperature,0}
\end{aligned}$$
that maps $\oagnsvector[\Psi_{\txtbec,\sminvtemperature}]$ to the constant
section $1\otimes\oagnsvector[\Psi_{\txtnonzero,\sminvtemperature,0}]$
and implements
\begin{equation}\label{eq:free-bose-resolvent-direct-integral}
\begin{aligned}
\oarepn_{\txtbec,\sminvtemperature}
&\simeq
\int_{\fldreal^{2}}^{\oplus}\oarepn_{\txtnonzero,\sminvtemperature,r,\theta}
\opdmsr{\msrprb_{\txtstd,\fldcmp}(r,\theta)},
&
\fun{\oa{Z}}{\oa{M}_{\txtbec,\sminvtemperature}}
&\cong
\fun{\lp^\infty}{\fldreal^{2},\msrprb_{\txtstd,\fldcmp}}.
\end{aligned}
\end{equation}
Under this unitary operator, $\oa{M}_{\txtbec,\sminvtemperature}$ corresponds to
the von Neumann tensor product
$\fun{\lp^\infty}{\fldreal^{2},\msrprb_{\txtstd,\fldcmp}}\otimes\oadoublecommutant{\fun{\oarepn_{\txtnonzero,\sminvtemperature,0}}{\fun{\oaresolventalgebra}{\sphilbfrak{D}_{0,\sminvtemperature}}}}$,
its center $\fun{\oa{Z}}{\oa{M}_{\txtbec,\sminvtemperature}}$ to
$\fun{\lp^\infty}{\fldreal^{2},\msrprb_{\txtstd,\fldcmp}}\otimes\fldcmp$, and
$Z$ to the multiplication operator by $F$.
\end{enumerate}
\end{thm}

\begin{proof}
For (1), Lemma \ref{lem:free-bose-araki-woods-components}(2) gives the equality
of the von Neumann algebras, the factor
property, and the cyclicity.
By the proof of Proposition \ref{prop:free-bose-resolvent-central-probes}, the
operators
$\fun{\oarepn_{\txtnonzero,\sminvtemperature,r,\theta}}{\fun{\oaresolvent}{1,\mathsf{b}_{1,L}}}$
and
$\fun{\oarepn_{\txtnonzero,\sminvtemperature,r,\theta}}{\fun{\oaresolvent}{1,\imunit\mathsf{b}_{1,L}}}$
converge strongly to
$\fun{F}{r,\theta}$ and $\fun{G}{r,\theta}$.
An intertwiner $Y$ from $\oarepn_{\txtnonzero,\sminvtemperature,r,\theta}$ to
$\oarepn_{\txtnonzero,\sminvtemperature,r',\theta'}$
therefore satisfies $\fun{F}{r,\theta}Y=\fun{F}{r',\theta'}Y$ and
$\fun{G}{r,\theta}Y=\fun{G}{r',\theta'}Y$, and $Y\neq0$
forces $\sqrt{r}\napiernum^{\imunit\theta}=\sqrt{r'}\napiernum^{\imunit\theta'}$
because
$\fun{\smnumberdensity_0}{\sminvtemperature}>0$.
Hence the representations have no nonzero intertwiners and are disjoint
\cite[Section 4.2.2, p.~370]{BratteliRobinson003}.

For (2), the space $\sphilbfrak{H}_{\txtnonzero,\sminvtemperature,0}$ is
separable, because the correlations depend
continuously on the parameters and on the test functions for
$\opform{q}_{\txtnonzero,\sminvtemperature}^{\frac{1}{2}}$.
By Lemma \ref{lem:free-bose-shifted-quasifree-states},
$\pairbk{r,\theta}\mapsto\fun{\oarepn_{\txtnonzero,\sminvtemperature,r,\theta}}{A}$
is norm continuous, and the constant field defines the direct-integral
representation $\oarepn$ on
$\fun{\lp^{2}}{\fldreal^{2},\msrprb_{\txtstd,\fldcmp}}\otimes\sphilbfrak{H}_{\txtnonzero,\sminvtemperature,0}$,
whose diagonal
algebra $\fun{\lp^\infty}{\fldreal^{2},\msrprb_{\txtstd,\fldcmp}}\otimes\fldcmp$
is the multiplication algebra
\cite[Section IV.8]{MasamichiTakesaki003}.
By (1),
$\oa{M}=\oadoublecommutant{\fun{\oarepn}{\fun{\oaresolventalgebra}{\sphilbfrak{D}_{0,\sminvtemperature}}}}$
is contained in
$\fun{\lp^\infty}{\fldreal^{2},\msrprb_{\txtstd,\fldcmp}}\otimes\oadoublecommutant{\fun{\oarepn_{\txtnonzero,\sminvtemperature,0}}{\fun{\oaresolventalgebra}{\sphilbfrak{D}_{0,\sminvtemperature}}}}$
\cite[Theorem IV.5.9, p.~226]{MasamichiTakesaki003}.
Conversely, dominated convergence shows that
$\fun{\oarepn}{\fun{\oaresolvent}{1,\mathsf{b}_{1,L}}}$ and
$\fun{\oarepn}{\fun{\oaresolvent}{1,\imunit\mathsf{b}_{1,L}}}$ converge strongly
to the multiplication operators by $F$ and $G$,
so that $\oa{M}$ contains the multiplication operators by the bounded Borel
functions of
$\imunit-F^{-1}=\sqrt{2\fun{\smnumberdensity_0}{\sminvtemperature}r}\fun{\cos}{\theta}$
and
$G^{-1}-\imunit=\sqrt{2\fun{\smnumberdensity_0}{\sminvtemperature}r}\fun{\sin}{\theta}$,
and hence the diagonal algebra, since the map
$$\begin{aligned}
\pairbk{r,\theta}
&\mapsto
\sqrt{2\fun{\smnumberdensity_0}{\sminvtemperature}r}\pairbk{\fun{\cos}{\theta},\fun{\sin}{\theta}}
\end{aligned}$$
is a
Borel isomorphism of $\fldreal^{2}$.
Multiplying the decomposable unitary operator with the values
$\napiernum^{\imunit\rbk{\fun{\phi_{\oarepn_{\txtnonzero,\sminvtemperature,0}}}{f}+\fun{\ell_{\sminvtemperature,r,\theta}}{f}}}$,
a Borel function of $\fun{\oarepn}{\fun{\oaresolvent}{1,f}}$, by the diagonal
operator
$\napiernum^{-\imunit\fun{\ell_{\sminvtemperature,r,\theta}}{f}}$ shows that
$\oa{M}$ contains
$1\otimes\fun{\oarepn_{\txtnonzero,\sminvtemperature,0}}{\fun{\oaresolvent}{\lambda,f}}$.
Hence $\oa{M}$ is the tensor product, and its center is
$\fun{\lp^\infty}{\fldreal^{2},\msrprb_{\txtstd,\fldcmp}}\otimes\fldcmp$ by the
factor property
\cite[Corollary IV.5.11, p.~227]{MasamichiTakesaki003}.
The constant section is cyclic, since $\oa{M}$ maps it onto a total set, and its
vector state is the mixture
\eqref{eq:free-bose-resolvent-state-mixture}; the uniqueness of the GNS
representation gives the unitary operator, which
maps $Z$ to the multiplication operator by $F$.
\end{proof}

\begin{cor}[Araki--Woods form of the central decomposition]\label{cor:free-bose-araki-woods-components}
Theorem \ref{thm:free-bose-resolvent-central-decomposition} holds with
$\oarepn_{\sminvtemperature,r,\theta}$,
$\oagnsvector[\Psi]_{\txtarakiwoods}$, and
$\fun{\spfockgamma_{\txtbsn}}{\sphilbfrak{K}_{\txtarakiwoods}}$ of
Lemma \ref{lem:free-bose-araki-woods-components} in place of
$\oarepn_{\txtnonzero,\sminvtemperature,r,\theta}$,
$\oagnsvector[\Psi_{\txtnonzero,\sminvtemperature,0}]$, and
$\sphilbfrak{H}_{\txtnonzero,\sminvtemperature,0}$.
\end{cor}

\begin{proof}
Conjugate by $1\otimes U$ with $U$ of Lemma
\ref{lem:free-bose-araki-woods-components}(2).
\end{proof}

In Proposition \ref{prop:free-bose-resolvent-central-probes}, the smeared Segal field disappears in the limit and only the scalar \(\fun{\ell_{\sminvtemperature,r,\theta}}{\mathsf{b}_{1,L}}\) survives. Under condensation, the bounded measurable functions of the order coordinate form the center of the represented von Neumann algebra, whereas the represented \(\oacstar\)-algebra has trivial center \cite[remark after Theorem 4.10, p.~2735]{BuchholzGrundling002}, so that its central elements cannot realize nonconstant functions of the order coordinate. The phase therefore labels the components and is not an observable of the \(\oacstar\)-algebra. The measure \(\msrprb_{\txtstd,\fldcmp}\) represents the BEC state by mutually disjoint factor states, and the diagonal algebra of the direct integral is the center of the BEC representation; in this sense it is the central decomposition of the BEC state.

\subsection{Gauge symmetry breaking and clustering}\label{sec:free-bose-gauge-clustering}

In this subsection, we describe the gauge action on the components and the clustering of the components and of their mixture.

For \(\theta_0\in\fldreal\) and \(x\in\fldreal^d\), the maps \(f\mapsto\napiernum^{\imunit\theta_0}f\) and \(\funrbk{\tau_xf}{y}=\fun{f}{y-x}\) are symplectic linear bijections of \(\sphilbfrak{D}_{0,\sminvtemperature}\), and they induce the gauge automorphisms \(\gamma_{\theta_0}\) and the translations \(\alpha_x\) of \(\fun{\oaresolventalgebra}{\sphilbfrak{D}_{0,\sminvtemperature}}\) with \(\fun{\gamma_{\theta_0}}{\fun{\oaresolvent}{\lambda,f}}=\fun{\oaresolvent}{\lambda,\napiernum^{\imunit\theta_0}f}\) and \(\fun{\alpha_x}{\fun{\oaresolvent}{\lambda,f}}=\fun{\oaresolvent}{\lambda,\tau_xf}\) \cite[Definition 3.1]{BuchholzGrundling002}. A state \(\oastate[\psi]\) is clustering for \(\alpha\) if \(\lim_{\abs{x}\to\infty}\fun{\oastate[\psi]}{A\fun{\alpha_x}{B}}=\fun{\oastate[\psi]}{A}\fun{\oastate[\psi]}{B}\) for all \(A\) and \(B\) \cite[Example 4.3.24]{BratteliRobinson003}, and it has long-range order in the two-point functions of a set \(S\) if this fails for some \(A,B\in S\).

\begin{prop}[Resolvent gauge orbit and clustering]\label{prop:free-bose-resolvent-gauge-clustering}
The following statements hold.
\begin{enumerate}
\item For every order coordinate,
$\oastate[\psi_{\sminvtemperature,r,\theta}]\circ\gamma_{\theta_0}=\oastate[\psi_{\sminvtemperature,r,\theta+\theta_0}]$.
If $\fun{\smnumberdensity_0}{\sminvtemperature}>0$, the component states with
$r>0$ are not gauge invariant and the one
with $r=0$ is; if $\fun{\smnumberdensity_0}{\sminvtemperature}=0$, all component
states are gauge invariant.
\item The BEC state is gauge invariant.
\item Every component state is clustering for $\alpha$.
If $\fun{\smnumberdensity_0}{\sminvtemperature}>0$, the BEC state has long-range
order in the two-point functions of the
resolvents and is not clustering, while its restriction to the
$\oacstar$-subalgebra generated by the resolvents
$\fun{\oaresolvent}{\lambda,g}$ with $\fun{\faftr{g}}{0}=0$ is clustering.
\end{enumerate}
\end{prop}

\begin{proof}
For (1), $\opform{\sigma}$ and $\opform{q}_{\txtnonzero,\sminvtemperature}$ are
gauge invariant and
$\fun{\ell_{\sminvtemperature,r,\theta}}{\napiernum^{\imunit\theta_0}f}=\fun{\ell_{\sminvtemperature,r,\theta+\theta_0}}{f}$,
so
\eqref{eq:free-bose-shifted-quasifree-resolvent} gives the identity.
If $\fun{\smnumberdensity_0}{\sminvtemperature}>0$, then $s<d$,
$\mathsf{b}_{1,1}\in\sphilbfrak{D}_{0,\sminvtemperature}$, and
$$\begin{aligned}
\fun{\ell_{\sminvtemperature,r,\theta}}{\mathsf{b}_{1,1}}+\imunit\fun{\ell_{\sminvtemperature,r,\theta}}{-\imunit\mathsf{b}_{1,1}}
&=
\sqrt{2\fun{\smnumberdensity_0}{\sminvtemperature}r}\napiernum^{\imunit\theta}
\end{aligned}$$
determines $\theta$ for $r>0$.
Item (2) follows from (1) and the rotation invariance of
$\msrprb_{\txtstd,\fldcmp}$.

For (3), $\tau_x$ preserves $\opform{\sigma}$,
$\opform{q}_{\txtnonzero,\sminvtemperature}$, and $\fun{\faftr{g}}{0}$, and the
Riemann--Lebesgue lemma gives $\fun{\opform{\sigma}}{f,\tau_xg}\to0$ and
$\fun{\opform{q}_{\txtnonzero,\sminvtemperature}}{f,\tau_xg}\to0$.
In \eqref{eq:free-bose-shifted-quasifree-resolvent} for $A\fun{\alpha_x}{B}$,
with $A$ and $B$ products of resolvents, the
cross terms are
$\fnexp{-\frac{\imunit}{2}\fun{\opform{\sigma}}{u,\tau_xv}-\frac{1}{2}\opreal\fun{\opform{q}_{\txtnonzero,\sminvtemperature}}{u,\tau_xv}}$
for $u=\sum_jt_jf_j$ and $v=\sum_kt'_kg_k$, and dominated convergence gives the
clustering of the components on products and
hence on the whole algebra.
For the BEC state, the clustering of the components and dominated convergence in
\eqref{eq:free-bose-resolvent-state-mixture} give
$$\begin{aligned}
\lim_{\abs{x}\to\infty}\fun{\oastate[\psi_{\txtbec,\sminvtemperature}]}{A\fun{\alpha_x}{B}}
&=
\sqfun{\prbexp_{\msrprb_{\txtstd,\fldcmp}(r,\theta)}}{\fun{\oastate[\psi_{\sminvtemperature,r,\theta}]}{A}\fun{\oastate[\psi_{\sminvtemperature,r,\theta}]}{B}}.
\end{aligned}$$
For $A=\fun{\oaresolvent}{1,f}$ and $B=\faadj{A}$ with
$\fun{\faftr{f}}{0}\neq0$, this is the second moment of the function
$$\begin{aligned}
\pairbk{r,\theta}\mapsto\fun{\oastate[\psi_{\sminvtemperature,r,\theta}]}{A}
&=
-\imunit\int_0^\infty\fnexp{-t-\imunit t\fun{\ell_{\sminvtemperature,r,\theta}}{f}-\frac{t^2}{4}\fun{\opform{q}_{\txtnonzero,\sminvtemperature}}{f}}\opdmsr{t},
\end{aligned}$$
which is not constant, since $\fun{\ell_{\sminvtemperature,r,\theta}}{f}$ is not
constant and the Laplace transform is
injective; hence it is larger than
$\abs{\fun{\oastate[\psi_{\txtbec,\sminvtemperature}]}{A}}^{2}=\fun{\oastate[\psi_{\txtbec,\sminvtemperature}]}{A}\fun{\oastate[\psi_{\txtbec,\sminvtemperature}]}{B}$.
On the resolvents of $g$ with $\fun{\faftr{g}}{0}=0$, all the components
coincide, and the argument for the components
applies.
\end{proof}

In the component \(\oastate[\psi_{\sminvtemperature,r,\theta}]\), the smeared Segal field with the test function \(\mathsf{b}_{1,L}\) is replaced in the limit by the scalar \(\sqrt{2\fun{\smnumberdensity_0}{\sminvtemperature}r}\fun{\cos}{\theta}\) (Proposition \ref{prop:free-bose-resolvent-central-probes}(2)). This is the operator-algebraic form of the Bogoliubov substitution of the zero-mode annihilation and creation operators by \(L^{\frac{d}{2}}\alpha\) and \(L^{\frac{d}{2}}\cmpconj{\alpha}\) for a complex number \(\alpha\), and the choice of \(\alpha\) corresponds to the choice of a component of the central decomposition.

The auxiliary space \(\sphilbfrak{D}_{0,\sminvtemperature}\) excludes vectors of \(\sphilbfrak{H}_{\sminvtemperature}\) without zero-momentum value. For example, for \(d=3\) and \(0<s<3\), the function \(\fun{f}{x}=\fndef{\set{y}{\abs{y}\geq1}}(x)\abs{x}^{-3}\) lies in \(\fun{\lp^2}{\fldreal^3}\setminus\fun{\lp^1}{\fldreal^3}\), its Fourier transform is \(\fun{\faftr{f}}{k}=4\pi\int_{2\pi\abs{k}}^{\infty}u^{-2}\fun{\sin}{u}\opdmsr{u}\), which tends to \(\infty\) as \(k\to0\), and \(\int_{\fldreal^3}\fun{\coth}{\frac{\sminvtemperature\fun{\faftr{\physham[h]}}{k}}{2}}\abs{\fun{\faftr{f}}{k}}^{2}\opdmsr{k}<\infty\). Thus \(f\in\sphilbfrak{H}_{\sminvtemperature}\), while \(\fun{\opform{q}_0}{f}\) is undefined. A representation of a resolvent algebra containing such a vector in which its resolvents vanish is not regular, and the corresponding ideals are proper \cite[Theorem 3.8, p.~2731]{BuchholzGrundling002}.

\section{The Euclidean Gaussian field}\label{sec:gaussian-functional-bec}

In this section, the free Bose gas is described by a Euclidean random field on the singular Gaussian \(\sminvtemperature\)-Markov path spaces of \cite{YoshitsuguSekine018}: the finite-volume fields at a fixed chemical potential form a projective system, the component states and the BEC state are realized by Gaussian measures on a common carrier, under convergence of the zero-mode and nonzero-mode forms the one-dimensional characteristic functions of the density-selected finite-volume sharp-time fields converge to those of the BEC measure, and the decomposition of the BEC measure into the component measures is the decomposition into invariant and ergodic conditional laws, which corresponds to the central decomposition of Section \ref{sec:resolvent-algebraic-bec}.

\subsection{Finite-volume Gaussian fields}\label{sec:free-bose-finite-volume-gaussian-fields}

In this subsection, we recall the finite-volume Euclidean fields at a fixed chemical potential and express the finite-volume expectations and the order parameter by them.

For \(L>0\) and \(\smchemicalpotential<0\), let \(\physham[h]_{\setlattice_L^d}\) be the multiplication by \(\fun{\faftr{\physham[h]}}{k}\) on \(\fun{\lpseq^{2}}{\setlattice_L^{d}}\), which corresponds to \(\physham[h]_{I_L^d}\) under the unitary map \(f\mapsto\rbk{L^{-\frac{d}{2}}f_k}_{k\in\setlattice_L^{d}}\). The operator \(\physham[h]_{\setlattice_L^d}-\smchemicalpotential\), the finitely supported sequences as the test space, and the diagonal regularizer of \cite[Example 6.5]{YoshitsuguSekine018} define a singular Gaussian \(\sminvtemperature\)-Markov path space \begin{equation}\label{eq:free-bose-finite-volume-path-space}
\begin{aligned}
\pairbk{\prbqspace_{\sminvtemperature,\smchemicalpotential,L,-},
\mblfmlfrak{S}_{\txttot,\sminvtemperature,\smchemicalpotential,L},
\mblfmlfrak{S}_{\txttot,\sminvtemperature,\smchemicalpotential,L,0},
\fml{U_{\sminvtemperature,\smchemicalpotential,L,t}}{t\in\fldreal},
R_{\sminvtemperature,\smchemicalpotential,L},
\msrprb_{\txttot,\sminvtemperature,\smchemicalpotential,L}},
\end{aligned}
\end{equation} with the sharp-time fields \(\opfocksegal_{\sminvtemperature,\smchemicalpotential,L,t}\) of \cite[Definition 4.4]{YoshitsuguSekine018}. For a fixed \(\smchemicalpotential<0\) and \(L_n=M_na\) with \(M_n\mid M_{n+1}\), the momentum lattices are nested, and these path spaces form a projective system whose limit carries the time translations and the reflection \cite[Theorem 6.4 and Example 6.5]{YoshitsuguSekine018}. The embeddings extend momentum sequences by zero, so the limit is not an infinite-volume field, and the density-selected chemical potentials \(\smchemicalpotential_{L^d}\) depend on \(L\); the thermodynamic limit at \(\smchemicalpotential_{L^d}\) is therefore taken for the characteristic functions in Subsection \ref{sec:free-bose-density-selected-limit}. A real-valued \(f\in\fun{\lp^{2}}{I_L^d}\) is regarded as a vector of the realification of \(\fun{\lpseq^{2}}{\setlattice_L^{d}}\) through the unitary map above.

\begin{lem}[Finite-volume Segal characteristic function]\label{lem:free-bose-finite-volume-segal-characteristic-function}
Let $f\in\fun{\lp^{2}}{I_L^d}$ be real valued with finitely many nonzero Fourier
coefficients, $\smchemicalpotential<0$,
$t\in S_{\sminvtemperature}$, and $u\in\fldreal$.
Then
\begin{equation}\label{eq:free-bose-finite-volume-segal-characteristic-function}
\begin{aligned}
\sqfun{\prbexp_{\msrprb_{\txttot,\sminvtemperature,\smchemicalpotential,L}}}{
\fnexp{\imunit u\fun{\opfocksegal_{\sminvtemperature,\smchemicalpotential,L,t}}{\rbk{2\rbk{\physham[h]_{\setlattice_L^d}-\smchemicalpotential}}^{\frac{1}{2}}f}}}
&=
\fnexp{-\frac{u^{2}}{4}\fun{\opform{q}_{\txtloc,\sminvtemperature,\smchemicalpotential,I_L^d}}{f}}.
\end{aligned}
\end{equation}
\end{lem}

\begin{proof}
The field is centered Gaussian with the variance
$\frac{1}{2}\bkt{f}{\fun{\coth}{\frac{\sminvtemperature\rbk{\physham[h]_{\setlattice_L^d}-\smchemicalpotential}}{2}}f}$
\cite[(27)]{YoshitsuguSekine018}, which is
$\frac{1}{2}\fun{\opform{q}_{\txtloc,\sminvtemperature,\smchemicalpotential,I_L^d}}{f}$
by Parseval's identity.
\end{proof}

\begin{prop}[Order parameter in the finite-volume Gaussian measures]\label{prop:free-bose-gaussian-order-parameter}
For $L>0$ and $\#=0,1$,
\begin{equation}\label{eq:free-bose-gaussian-order-parameter}
\begin{aligned}
\sqfun{\prbexp_{\msrprb_{\txttot,\sminvtemperature,\smchemicalpotential_{L^d},L}}}{
\rbk{\imunit-\fun{\opfocksegal_{\sminvtemperature,\smchemicalpotential_{L^d},L,0}}{\rbk{2\rbk{\physham[h]_{\setlattice_L^d}-\smchemicalpotential_{L^d}}}^{\frac{1}{2}}\mathsf{b}_{\#,L}}}^{-1}}
&=
\fun{\oastate[\psi_{\txtgrandcanonical,I_L^d,\sminvtemperature,\smchemicalpotential_{L^d}}]}{\fun{\oaresolvent}{1,\mathsf{b}_{\#,L}}},
\end{aligned}
\end{equation}
and hence the order parameter \eqref{eq:free-bose-zero-mode-resolvent-probes} is
$\imunit$ times the limit of the
left-hand side as $L\to\infty$.
\end{prop}

\begin{proof}
Both sides equal
$$\begin{aligned}
-\imunit\int_0^\infty\napiernum^{-t}\fnexp{-\frac{t^{2}}{4}\fun{\opform{q}_{\txtloc,\sminvtemperature,\smchemicalpotential_{L^d},I_L^d}}{\mathsf{b}_{\#,L}}}\opdmsr{t},
\end{aligned}$$
by
$\rbk{\imunit-x}^{-1}=-\imunit\int_0^\infty\napiernum^{-t-\imunit xt}\opdmsr{t}$,
Fubini's theorem, and
Lemma \ref{lem:free-bose-finite-volume-segal-characteristic-function} for the
left-hand side, and by
Proposition \ref{prop:free-bose-finite-volume-resolvent-expectations} for the
right-hand side.
\end{proof}

\subsection{Infinite-volume carrier and the component measures}\label{sec:free-bose-concrete-carrier}

In this subsection, we realize the nonzero-mode state on the path space \eqref{eq:free-bose-centered-path-space}, implement the unimodular Fourier multipliers, and construct the component measures and the BEC measure on the total carrier \eqref{eq:free-bose-total-carrier}.

The realification of \(\sphilbfrak{H}_{\sminvtemperature}\) is \(\sphilbfrak{h}_{\fldreal,\smnumberdensity}\) with \(\norm{K_{\sminvtemperature}^{\frac{1}{2}}f}=\norm{f}_{\smnumberdensity}\), since \(1+2\smnumberdensity_{\sminvtemperature}=K_{\sminvtemperature}\) \cite[Lemma 5.2(2)]{YoshitsuguSekine018}. Hence, by \cite[(27) and Corollary 5.7]{YoshitsuguSekine018}, for \(f\in\sphilbfrak{h}_{\fldreal,\smnumberdensity}\) and \(t\in S_{\sminvtemperature}\), \begin{equation}\label{eq:free-bose-nonzero-segal-characteristic-function}
\begin{aligned}
\sqfun{\prbexp_{\msrprb_{\txtnonzero,\sminvtemperature}}}{\napiernum^{\imunit\fun{\opfocksegal_t}{\rbk{2\physham[h]}^{\frac{1}{2}}f}}}
&=
\fnexp{-\frac{1}{4}\fun{\opform{q}_{\txtnonzero,\sminvtemperature}}{f}},
\end{aligned}
\end{equation} which is the vacuum expectation of the left Araki--Woods Weyl operator of \(f\) in the representation \eqref{eq:free-bose-araki-woods-data}. For real-valued test functions, the imaginary-time correlations of the fields \(\opfocksegal_t\rbk{\rbk{2\physham[h]}^{\frac{1}{2}}f}\) are those of this representation: \cite[Theorem 5.4 and Proposition 5.6]{YoshitsuguSekine018}, restricted to the real subspace \(\fun{\lp^2}{\fldreal^d;\fldreal}\), which is invariant under \(\physham[h]\) and whose complexification is \(\fun{\lp^2}{\fldreal^d;\fldcmp}\), give these correlations. The space \(\sphilbfrak{D}_{0,\sminvtemperature}\) is not contained in \(\dom\physham[h]^{\frac{1}{2}}\) in general, and \(\rbk{2\physham[h]}^{\frac{1}{2}}f\) in the argument of \(\opfocksegal_t\) denotes \(\fun{\iota_{\physham[h]}}{f}\).

The phase rotation \(f\mapsto\napiernum^{\imunit\theta_0}f\), the real-time evolution \(\napiernum^{\imunit t\physham[h]}\), and the spatial translation \(\funrbk{\tau_xf}{y}=\fun{f}{y-x}\) are the Fourier multipliers by the functions of modulus one \begin{equation}\label{eq:free-bose-unimodular-multiplier-examples}
\begin{aligned}
\napiernum^{\imunit\theta_0},
&&
\napiernum^{\imunit t\fun{\faftr{\physham[h]}}{k}},
&&
\napiernum^{-2\pi\imunit k\cdot x}.
\end{aligned}
\end{equation}

\begin{lem}[Unimodular Fourier multipliers on the path space]\label{lem:free-bose-sharp-time-phase-invariance}
Let $m$ be a Borel function on $\fldreal^d$ with $\abs{m}=1$, such as the
functions
\eqref{eq:free-bose-unimodular-multiplier-examples}, and let $U$ be the
realification of the operator with the Fourier
transform $m$.
Then $U$ is an orthogonal operator on
$\sphilbfrak{h}_{\fldreal,\smnumberdensity}$, and there is a unique orthogonal
operator $W_U$ on $\sphilbfrak{H}_{\sminvtemperature,C_{\sminvtemperature}}$
with
\begin{equation}\label{eq:free-bose-multiplier-covariance-operator}
\begin{aligned}
W_Uj_{\sminvtemperature,t}\fun{\iota_{\physham[h]}}{f}
&=
j_{\sminvtemperature,t}\fun{\iota_{\physham[h]}}{Uf},
\quad
t\in S_{\sminvtemperature},\ f\in\sphilbfrak{h}_{\fldreal,\smnumberdensity}.
\end{aligned}
\end{equation}
There is a Borel map $\tilde{W}_U$ of $\prbqspace_{\sminvtemperature,-}$ with
$\pushoutrbk{\tilde{W}_U}\msrprb_{\txtnonzero,\sminvtemperature}=\msrprb_{\txtnonzero,\sminvtemperature}$
and
$\fun{\opfocksegal}{v}\circ\tilde{W}_U=\fun{\opfocksegal}{W_Uv}$ almost surely
for every
$v\in\sphilbfrak{H}_{\sminvtemperature,C_{\sminvtemperature}}$.
\end{lem}

\begin{proof}
The operator $U$ commutes with the real functional calculus of
$\physham[h]$ and preserves $\norm{\cdot}_{\smnumberdensity}$.
For vectors whose Fourier transforms are supported in a compact annulus, the
Green-kernel formula \cite[(10) and Proposition 3.4]{YoshitsuguSekine018} shows
that simultaneous application of $U$ preserves all sharp-time pairings, and
approximation in $\norm{\cdot}_{\smnumberdensity}$ extends this identity to all
the vectors in the statement.
Hence
$$\begin{aligned}
\sum_ic_ij_{\sminvtemperature,t_i}\fun{\iota_{\physham[h]}}{f_i}
&\mapsto
\sum_ic_ij_{\sminvtemperature,t_i}\fun{\iota_{\physham[h]}}{Uf_i}
\end{aligned}$$
is a well-defined isometry, whose domain and range are dense by
\cite[Lemma 3.8]{YoshitsuguSekine018}; its extension is
$W_U$.
The Cameron--Martin space of $\msrprb_{\txtnonzero,\sminvtemperature}$ is
$\sphilbfrak{H}_{\sminvtemperature,C_{\sminvtemperature}}$
with the norm multiplied by $\sqrt{2}$ \cite[Section 4.1]{YoshitsuguSekine018},
and \cite[Theorem 3.7.6, p.~124]{VladimirBogachev002}, applied to the orthogonal
operator $W_U^{\ast}$, gives a measurable linear extension whose image measure
is
$\msrprb_{\txtnonzero,\sminvtemperature}$.
A Borel version $\tilde{W}_U$ of it satisfies
$\fun{\opfocksegal}{v}\circ\tilde{W}_U=\fun{\opfocksegal}{\rbk{W_U^{\ast}}^{\ast}v}=\fun{\opfocksegal}{W_Uv}$
almost surely, by the adjoint relation for the Gaussian linear functionals.
\end{proof}

For \(F\in\fun{\lp^2}{S_{\sminvtemperature};\sphilbfrak{h}_{\fldreal}}\) with the Matsubara coefficients \(F_n\) of \cite[(2)]{YoshitsuguSekine018} and \(\nu_n=2\pi n/\sminvtemperature\), \cite[Lemma 2.3]{YoshitsuguSekine018} shows that \(F\in\dom C_{\sminvtemperature}^{\frac{1}{2}}\) if and only if \(\int_{\fldreal^d}\fun{\faftr{\physham[h]}}{k}^{-2}\abs{\fun{\faftr{F_0}}{k}}^{2}\opdmsr{k}<\infty\), and then the vector \(\iota_{C_{\sminvtemperature}}F=V^{-1}C_{\sminvtemperature}^{\frac{1}{2}}F\) of \(\sphilbfrak{H}_{\sminvtemperature,C_{\sminvtemperature}}\), where \(V\) is the unitary operator of \cite[Definition 3.1]{YoshitsuguSekine018}, satisfies \begin{equation}\label{eq:free-bose-coordinate-covariance}
\begin{aligned}
\norm{\iota_{C_{\sminvtemperature}}F}^{2}
&=
\sum_{n\in\ringratint}
\int_{\fldreal^d}
\frac{\abs{\fun{\faftr{F_n}}{k}}^{2}}{\nu_n^{2}+\fun{\faftr{\physham[h]}}{k}^{2}}
\opdmsr{k}.
\end{aligned}
\end{equation} The field \(\fun{\opfocksegal}{\iota_{C_{\sminvtemperature}}F}\) is the coordinate field, and \(\opfocksegal_t\rbk{\rbk{2\physham[h]}^{\frac{1}{2}}f}\) is the rescaled sharp-time field. For the time-constant functions \(\fun{F_{g}}{t}=\sminvtemperature^{-1}g\), \eqref{eq:free-bose-coordinate-covariance} gives \begin{equation}\label{eq:free-bose-time-constant-covariances}
\begin{aligned}
\norm{\iota_{C_{\sminvtemperature}}F_{g}}^{2}
&=
\frac{1}{\sminvtemperature}\int_{\fldreal^d}\fun{\faftr{\physham[h]}}{k}^{-2}\abs{\fun{\faftr{g}}{k}}^{2}\opdmsr{k},
&
\norm{\iota_{C_{\sminvtemperature}}F_{\rbk{2\physham[h]}^{\frac{1}{2}}f}}^{2}
&=
\frac{2}{\sminvtemperature}\int_{\fldreal^d}\fun{\faftr{\physham[h]}}{k}^{-1}\abs{\fun{\faftr{f}}{k}}^{2}\opdmsr{k}.
\end{aligned}
\end{equation} In particular, for a real-valued \(f\in\fun{C_{\mathrm{c}}^{\infty}}{\fldreal^d}\) with \(\int f\opdmsr{x}\neq0\), the time-constant coordinate test function \(F_f\) belongs to \(\dom C_{\sminvtemperature}^{\frac{1}{2}}\) if and only if \(d>2s\), and the time-constant rescaled test function \(F_{\rbk{2\physham[h]}^{\frac{1}{2}}f}\) if and only if \(d>s\). The Fourier transform of \(\rbk{2\physham[h]}^{\frac{1}{2}}f\) vanishes at \(k=0\), whereas the mean \(\fun{\ell_{\sminvtemperature,r,\theta}}{f}\) depends on \(\fun{\faftr{f}}{0}\); the condensate contribution in \eqref{eq:free-bose-bec-enlarged-sharp-time-field} is indexed by the test function \(f\) of the rescaled field.

\begin{prop}[Component measures and the BEC measure on the common carrier]\label{prop:free-bose-bec-component-carrier}
Let $\msrprb_{\sminvtemperature,r,\theta}$ and
$\msrprb_{\txtbec,\sminvtemperature}$ be the measures
\eqref{eq:free-bose-bec-component-carrier-law}.
\begin{enumerate}
\item For every order coordinate, $\msrprb_{\sminvtemperature,r,\theta}$ is the
image of $\msrprb_{\txtnonzero,\sminvtemperature}$
under $q\mapsto\pairbk{q,r,\theta}$, which intertwines
$\tilde{u}_{\sminvtemperature,t}$ and $\tilde{r}_{\sminvtemperature}$
with $\tilde{u}_{\sminvtemperature,t}\times\id$ and
$\tilde{r}_{\sminvtemperature}\times\id$; under it, the field
\eqref{eq:free-bose-bec-enlarged-sharp-time-field} is the centered field plus
the constant $\fun{\ell_{\sminvtemperature,r,\theta}}{f}$.
The total carrier with its $\sigma$-algebras, time translations, and reflection,
completed for
$\msrprb_{\sminvtemperature,r,\theta}$, is a $\sminvtemperature$-Markov path
space.
\item For $f\in\sphilbfrak{D}_{0,\sminvtemperature}$,
$t\in S_{\sminvtemperature}$, and $u\in\fldreal$,
\begin{equation}\label{eq:free-bose-component-characteristic-function}
\begin{aligned}
\sqfun{\prbexp_{\msrprb_{\sminvtemperature,r,\theta}}}{
\napiernum^{\imunit u\rbk{\opfocksegal_t\rbk{\rbk{2\physham[h]}^{\frac{1}{2}}f}+\fun{\ell_{\sminvtemperature,r,\theta}}{f}}}}
&=
\fnexp{-\frac{u^{2}}{4}\fun{\opform{q}_{\txtnonzero,\sminvtemperature}}{f}+\imunit u\fun{\ell_{\sminvtemperature,r,\theta}}{f}},
\end{aligned}
\end{equation}
the map $\pairbk{r,\theta}\mapsto\fun{\msrprb_{\sminvtemperature,r,\theta}}{A}$
is Borel for every
$A\in\mblfmlfrak{S}_{\txttot,\sminvtemperature}$, and
\begin{equation}\label{eq:free-bose-bec-measure-mixture}
\begin{aligned}
\fun{\msrprb_{\txtbec,\sminvtemperature}}{A}
&=
\sqfun{\prbexp_{\msrprb_{\txtstd,\fldcmp}(r,\theta)}}{\fun{\msrprb_{\sminvtemperature,r,\theta}}{A}}.
\end{aligned}
\end{equation}
\item The characteristic function of the BEC measure is
\begin{equation}\label{eq:free-bose-bec-measure-characteristic-function}
\begin{aligned}
\sqfun{\prbexp_{\msrprb_{\txtbec,\sminvtemperature}}}{
\napiernum^{\imunit u\rbk{\opfocksegal_t\rbk{\rbk{2\physham[h]}^{\frac{1}{2}}f}+\fun{\ell_{\sminvtemperature,r,\theta}}{f}}}}
&=
\fnexp{-\frac{u^{2}}{4}\rbk{\fun{\opform{q}_{0}}{f}+\fun{\opform{q}_{\txtnonzero,\sminvtemperature}}{f}}}.
\end{aligned}
\end{equation}
\item Distinct order coordinates give mutually singular component measures.
\end{enumerate}
\end{prop}

\begin{proof}
Item (1) is immediate from the product structure, and the
$\sminvtemperature$-Markov property is that of
\cite[Theorem 4.6]{YoshitsuguSekine018} transported by the isomorphism
$q\mapsto\pairbk{q,r,\theta}$ of measure spaces.
Item (2) follows from \eqref{eq:free-bose-nonzero-segal-characteristic-function}
and Fubini's theorem, and (3) from (2) and
Lemma \ref{lem:free-bose-zero-mode-polar-representation} with $uf$ in place of
$f$.
For (4), the second coordinate is $\pairbk{r,\theta}$ almost surely under
$\msrprb_{\sminvtemperature,r,\theta}$.
\end{proof}

By (1), the component measures realize the \(\sminvtemperature\)-Markov path spaces whose sharp-time fields are the centered ones shifted by the constants \(\fun{\ell_{\sminvtemperature,r,\theta}}{f}\), and \cite[Theorem 5.4, Proposition 5.6, and Corollary 5.7]{YoshitsuguSekine018}, restricted to \(\fun{\lp^2}{\fldreal^d;\fldreal}\) as above, identify their correlation functions with the imaginary-time correlations of the left Araki--Woods representation with the shifted fields \eqref{eq:free-bose-shifted-resolvent-representation}, for real-valued test functions in \(\sphilbfrak{D}_{0,\sminvtemperature}\); \eqref{eq:free-bose-bec-measure-mixture} is the functional-integral counterpart of \eqref{eq:free-bose-resolvent-state-mixture}. This identification does not extend to complex test functions: the random variables \(\opfocksegal_0\rbk{\rbk{2\physham[h]}^{\frac{1}{2}}f}\) and \(\opfocksegal_0\rbk{\rbk{2\physham[h]}^{\frac{1}{2}}\imunit f}\) commute, whereas \(\fun{\opform{\sigma}}{f,\imunit f}=\norm{f}^{2}\neq0\).

\subsection{Thermodynamic limit at the density-selected chemical potentials}\label{sec:free-bose-density-selected-limit}

In this subsection, we prove that the one-dimensional characteristic functions of the density-selected finite-volume sharp-time fields converge to those of the BEC measure.

\begin{thm}[Density-selected limit of the sharp-time characteristic functions]\label{thm:free-bose-density-selected-functional-integral}
Assume
$\bar{\smnumberdensity}>\fun{\smnumberdensity_{\txtcritical}}{\sminvtemperature}$.
Let $f\in\sphilbfrak{D}_{0,\sminvtemperature}$ be real valued, and let
$f_n\in\fun{\lp^{2}}{I_{L_n}^d}$ be real valued with
finitely many nonzero Fourier coefficients, where $L_n\to\infty$,
such that
\begin{equation}\label{eq:free-bose-density-selected-form-hypotheses}
\begin{aligned}
\lim_{n\to\infty}\fun{\opform{q}_{0,\sminvtemperature,\smchemicalpotential_{L_n^d},L_n}}{f_n}
&=
\fun{\opform{q}_0}{f},
&
\lim_{n\to\infty}\fun{\opform{q}_{\txtnonzero,\sminvtemperature,\smchemicalpotential_{L_n^d},L_n}}{f_n}
&=
\fun{\opform{q}_{\txtnonzero,\sminvtemperature}}{f}.
\end{aligned}
\end{equation}
Then, for $t\in S_{\sminvtemperature}$ and $u\in\fldreal$,
\begin{equation}\label{eq:free-bose-density-selected-characteristic-functional-limit}
\begin{aligned}
&\lim_{n\to\infty}
\sqfun{\prbexp_{\msrprb_{\txttot,\sminvtemperature,\smchemicalpotential_{L_n^d},L_n}}}{
\fnexp{\imunit u\fun{\opfocksegal_{\sminvtemperature,\smchemicalpotential_{L_n^d},L_n,t}}{\rbk{2\rbk{\physham[h]_{\setlattice_{L_n}^d}-\smchemicalpotential_{L_n^d}}}^{\frac{1}{2}}f_n}}}
\\ 
&=
\sqfun{\prbexp_{\msrprb_{\txtbec,\sminvtemperature}}}{
\fnexp{\imunit u\rbk{\opfocksegal_t\rbk{\rbk{2\physham[h]}^{\frac{1}{2}}f}+\fun{\ell_{\sminvtemperature,r,\theta}}{f}}}}.
\end{aligned}
\end{equation}
\end{thm}

\begin{proof}
Combine Lemma \ref{lem:free-bose-finite-volume-segal-characteristic-function},
\eqref{eq:free-bose-density-selected-form-hypotheses}, and
\eqref{eq:free-bose-bec-measure-characteristic-function}.
\end{proof}

The zero-mode hypothesis holds when \(\rbk{f_n}_0\to\fun{\faftr{f}}{0}\), since \(\fun{\coth}{-\frac{\sminvtemperature\smchemicalpotential_{L_n^d}}{2}}=2\fun{N_0}{y_{L_n^d}}+1\) and \eqref{eq:free-bose-zero-mode-density-limit} give \[\begin{aligned}
\fun{\opform{q}_{0,\sminvtemperature,\smchemicalpotential_{L_n^d},L_n}}{f_n}
&=
L_n^{-d}\rbk{2\fun{N_0}{y_{L_n^d}}+1}\abs{\rbk{f_n}_0}^{2}
\\ 
&\to
2\fun{\smnumberdensity_0}{\sminvtemperature}\abs{\fun{\faftr{f}}{0}}^{2}.
\end{aligned}\] At a fixed \(\smchemicalpotential<0\), by contrast, the zero-mode term \(L_n^{-d}\fun{\coth}{-\frac{\sminvtemperature\smchemicalpotential}{2}}\abs{\rbk{f_n}_0}^{2}\) tends to \(0\). If, in addition, \[\begin{aligned}
\lim_{n\to\infty}\fun{\opform{q}_{\txtnonzero,\sminvtemperature,\smchemicalpotential,L_n}}{f_n}
&=
\fun{\opform{q}_{\txtnonzero,\sminvtemperature,\smchemicalpotential}}{f}
\end{aligned}\] for every fixed \(\smchemicalpotential<0\), then monotone convergence as \(\smchemicalpotential\uparrow0\) shows that the iterated limit \(L\to\infty\) followed by \(\smchemicalpotential\uparrow0\) gives \(\fnexp{-\frac{u^{2}}{4}\fun{\opform{q}_{\txtnonzero,\sminvtemperature}}{f}}\) and loses the zero-mode form \(\opform{q}_0\).

\subsection{Gaussian component kernels and mixtures}\label{sec:free-bose-gaussian-component-kernels}

In this subsection, we state the field-space data for a mixture of Gaussian component measures with a common covariance, prove the uniqueness of the mixing law and the component disintegration, and verify the data on the total carrier.

The covariance is independent of \(\pairbk{r,\theta}\), while the BEC information enters the mean. Even if every component measure is mixing, the mixture over \(\pairbk{r,\theta}\) can fail to be mixing; the order coordinate remaining in the invariant \(\sigma\)-algebra expresses this failure, and in the operator-algebraic description the center of the GNS representation records the same information.

\begin{defn}[Common Gaussian component setting]\label{def:bec-common-gaussian-component-setting}
Let $\pairbk{\prbqspace_{\txttot},\mblfmlfrak{S}_{\txttot}}$ be a measurable
space, let $\fun{\opfocksegal}{f}$,
$f\in\sphilbfrak{D}_{0,\sminvtemperature}$, be real-valued measurable functions
on $\prbqspace_{\txttot}$, let
$\kappa_0\colon\prbqspace_{\txttot}\to\fldreal^{2}$ be measurable, let
$\seq{f_j}{j\in\semigrposint}\subset\sphilbfrak{D}_{0,\sminvtemperature}$
be a countable family, and let
$\mathcal{N}_{\txttot}\subset\mblfmlfrak{S}_{\txttot}$ be a family of sets.
Assume that
\begin{equation}\label{eq:bec-common-gaussian-generating-sigma-algebras}
\begin{aligned}
\mblfmlfrak{S}_{\txttot}
&=
\mblfmlgenerated{\setone{\kappa_0}\cup\set{\fun{\opfocksegal}{f_j}}{j\in\semigrposint}\cup\mathcal{N}_{\txttot}}.
\end{aligned}
\end{equation}
For each order coordinate $\pairbk{r,\theta}$, let the component measure
$\msrprb_{\sminvtemperature,r,\theta}$ be a probability
measure on $\pairbk{\prbqspace_{\txttot},\mblfmlfrak{S}_{\txttot}}$ with
$\fun{\msrprb_{\sminvtemperature,r,\theta}}{\inv{\kappa_0}\rbk{\setone{\pairbk{r,\theta}}}}=1$,
$\fun{\msrprb_{\sminvtemperature,r,\theta}}{N}=0$ for
$N\in\mathcal{N}_{\txttot}$, and
\begin{equation}\label{eq:bec-common-gaussian-component-characteristic-functional}
\begin{aligned}
\sqfun{\prbexp_{\msrprb_{\sminvtemperature,r,\theta}}}{\fnexp{\imunit\sum_{a=1}^{n}u_a\fun{\opfocksegal}{g_a}}}
&=
\fnexp{-\frac{1}{4}\fun{\opform{q}_{\txtnonzero,\sminvtemperature}}{\sum_{a=1}^{n}u_ag_a}+\imunit\fun{\ell_{\sminvtemperature,r,\theta}}{\sum_{a=1}^{n}u_ag_a}}
\end{aligned}
\end{equation}
for every finite family $g_1,\ldots,g_n\in\sphilbfrak{D}_{0,\sminvtemperature}$
and $u_1,\ldots,u_n\in\fldreal$.
\end{defn}

\begin{lem}[Uniqueness and measurability of the Gaussian component measures]\label{lem:bec-common-gaussian-component-uniqueness}
In the setting of Definition \ref{def:bec-common-gaussian-component-setting},
every probability measure $\msrprb$ with
$\fun{\msrprb}{\inv{\kappa_0}\rbk{\setone{\pairbk{r,\theta}}}}=1$,
$\fun{\msrprb}{N}=0$ for $N\in\mathcal{N}_{\txttot}$, and
\eqref{eq:bec-common-gaussian-component-characteristic-functional} equals
$\msrprb_{\sminvtemperature,r,\theta}$.
The map $\pairbk{r,\theta}\mapsto\msrprb_{\sminvtemperature,r,\theta}$ is a
probability kernel from
$\pairbk{\fldreal^{2},\fun{\mblfmlborel}{\fldreal^{2}}}$ to
$\pairbk{\prbqspace_{\txttot},\mblfmlfrak{S}_{\txttot}}$, and
$\fun{\opfocksegal}{f+cg}=\fun{\opfocksegal}{f}+c\fun{\opfocksegal}{g}$ holds
$\msrprb_{\sminvtemperature,r,\theta}$-almost
surely for $f,g\in\sphilbfrak{D}_{0,\sminvtemperature}$ and $c\in\fldreal$.
\end{lem}

\begin{proof}
Every set of $\mblfmlfrak{S}_{\txttot}$ differs from a set of the
$\sigma$-algebra generated by $\kappa_0$ and the
$\fun{\opfocksegal}{f_j}$ by a subset of a countable union of sets of
$\mathcal{N}_{\txttot}$, by a monotone-class argument.
On this generated $\sigma$-algebra, two measures that give mass one to
$\inv{\kappa_0}\rbk{\setone{\pairbk{r,\theta}}}$ and
have the same finite-dimensional characteristic functions agree, by the
uniqueness of characteristic functions and the
$\pi$--$\lambda$ theorem.
The kernel property follows from the Borel measurability of the characteristic
functions
\eqref{eq:bec-common-gaussian-component-characteristic-functional} in
$\pairbk{r,\theta}$ and the same monotone-class
argument, and the linearity holds because
$\fun{\opfocksegal}{f+cg}-\fun{\opfocksegal}{f}-c\fun{\opfocksegal}{g}$ has the
characteristic function $1$.
\end{proof}

\begin{defn}[Gaussian component mixture]\label{def:bec-gaussian-component-mixture}
In the setting of Definition \ref{def:bec-common-gaussian-component-setting},
let $\msrprb_{\txtmix}$ be a Borel probability
measure on $\fldreal^{2}$.
The Gaussian component mixture with the mixing law $\msrprb_{\txtmix}$ is the
probability measure
\begin{equation}\label{eq:bec-gaussian-component-mixture}
\begin{aligned}
\fun{\msrprb_{\txtbec,\sminvtemperature}}{A}
&=
\sqfun{\prbexp_{\msrprb_{\txtmix}(r,\theta)}}{\fun{\msrprb_{\sminvtemperature,r,\theta}}{A}},
\quad
A\in\mblfmlfrak{S}_{\txttot}.
\end{aligned}
\end{equation}
\end{defn}

\begin{defn}[Order-coordinate realization]\label{def:bec-order-coordinate-realization}
In the setting of Definition \ref{def:bec-gaussian-component-mixture}, a
measurable map $\kappa\colon\prbqspace_{\txttot}\to\fldreal^{2}$ is an
order-coordinate realization if
$$\begin{aligned}
\pushout{\kappa}\msrprb_{\txtbec,\sminvtemperature}
&=
\msrprb_{\txtmix},
&
\fun{\msrprb_{\sminvtemperature,r,\theta}}{\inv{\kappa}\rbk{\setone{\pairbk{r,\theta}}}}
&=
1
\end{aligned}$$
for $\msrprb_{\txtmix}$-almost every $\pairbk{r,\theta}$.
\end{defn}

The map \(\kappa_0\) itself is an order-coordinate realization. By Lemma \ref{lem:bec-common-gaussian-component-uniqueness} and Fubini's theorem, for every \(f\in\sphilbfrak{D}_{0,\sminvtemperature}\) and every bounded measurable \(F\), \begin{equation}\label{eq:bec-characteristic-functional-factorization}
\begin{aligned}
\sqfun{\prbexp_{\msrprb_{\txtbec,\sminvtemperature}}}{\napiernum^{\imunit\fun{\opfocksegal}{f}}}
&=
\fnexp{-\frac{1}{4}\fun{\opform{q}_{\txtnonzero,\sminvtemperature}}{f}}
\sqfun{\prbexp_{\msrprb_{\txtmix}(r,\theta)}}{\napiernum^{\imunit\fun{\ell_{\sminvtemperature,r,\theta}}{f}}},
\\ 
\sqfun{\prbexp_{\msrprb_{\txtbec,\sminvtemperature}}}{F}
&=
\sqfun{\prbexp_{\msrprb_{\txtmix}(r,\theta)}}{\sqfun{\prbexp_{\msrprb_{\sminvtemperature,r,\theta}}}{F}}.
\end{aligned}
\end{equation}

\begin{lem}[Uniqueness of the mixing law of the free-Bose BEC measure]\label{lem:free-bose-mixing-law-uniqueness}
Assume $\fun{\smnumberdensity_0}{\sminvtemperature}>0$.
Then
\begin{equation}\label{eq:free-bose-mixing-law-uniqueness}
\begin{aligned}
\sqfun{\prbexp_{\msrprb_{\txtbec,\sminvtemperature}}}{\napiernum^{\imunit\fun{\opfocksegal}{f}}}
&=
\fnexp{-\frac{1}{4}\rbk{\fun{\opform{q}_{\txtnonzero,\sminvtemperature}}{f}+\fun{\opform{q}_0}{f}}}
\end{aligned}
\end{equation}
holds for every $f\in\sphilbfrak{D}_{0,\sminvtemperature}$ if and only if
$\msrprb_{\txtmix}=\msrprb_{\txtstd,\fldcmp}$.
\end{lem}

\begin{proof}
By \eqref{eq:bec-characteristic-functional-factorization},
\eqref{eq:free-bose-mixing-law-uniqueness} is equivalent to
$\sqfun{\prbexp_{\msrprb_{\txtmix}(r,\theta)}}{\napiernum^{\imunit\fun{\ell_{\sminvtemperature,r,\theta}}{f}}}=\fnexp{-\frac{1}{4}\fun{\opform{q}_0}{f}}$,
which holds for $\msrprb_{\txtstd,\fldcmp}$ by Lemma
\ref{lem:free-bose-zero-mode-polar-representation}.
Conversely, $s<d$ because the critical density is finite, so
$\mathsf{b}_{1,1}\in\sphilbfrak{D}_{0,\sminvtemperature}$, and
$f=\rbk{2\fun{\smnumberdensity_0}{\sminvtemperature}}^{-\frac{1}{2}}\rbk{t_1-\imunit t_2}\mathsf{b}_{1,1}$
gives
$\fun{\ell_{\sminvtemperature,r,\theta}}{f}=t\cdot w$ for the point $w$ of the
order coordinate and
$\fun{\opform{q}_0}{f}=\abs{t}^{2}$, $t=\pairbk{t_1,t_2}\in\fldreal^{2}$.
Hence $\msrprb_{\txtmix}$ and $\msrprb_{\txtstd,\fldcmp}$ have the same
characteristic function $\napiernum^{-\abs{t}^{2}/4}$.
\end{proof}

When \(\fun{\smnumberdensity_0}{\sminvtemperature}=0\), both \(\ell_{\sminvtemperature,r,\theta}\) and \(\opform{q}_0\) vanish, and the characteristic functional of the mixture is \(\fnexp{-\frac{1}{4}\opform{q}_{\txtnonzero,\sminvtemperature}}\) for every mixing law.

\begin{prop}[Component disintegration]\label{prop:free-bose-component-disintegration}
Let $\kappa$ be an order-coordinate realization
(Definition \ref{def:bec-order-coordinate-realization}).
For every $A\in\mblfmlfrak{S}_{\txttot}$ and every Borel set
$C\subset\fldreal^{2}$,
\begin{equation}\label{eq:free-bose-component-disintegration}
\begin{aligned}
\fun{\msrprb_{\txtbec,\sminvtemperature}}{A\cap\inv{\kappa}(C)}
&=
\sqfun{\prbexp_{\msrprb_{\txtmix}(r,\theta)}}{\fun{\fndef{C}}{r,\theta}\fun{\msrprb_{\sminvtemperature,r,\theta}}{A}},
\end{aligned}
\end{equation}
and, for every bounded measurable $F$,
$\sqfuncond{\prbexp_{\msrprb_{\txtbec,\sminvtemperature}}}{F}{\mblfmlgenerated{\kappa}}=\sqfun{\prbexp_{\msrprb_{\sminvtemperature,\kappa}}}{F}$
almost surely; that is,
$\pairbk{r,\theta}\mapsto\msrprb_{\sminvtemperature,r,\theta}$ is a regular
conditional probability given
$\kappa$.
\end{prop}

\begin{proof}
Apply \eqref{eq:bec-gaussian-component-mixture} to $A\cap\inv{\kappa}(C)$ and
use
$\fun{\msrprb_{\sminvtemperature,r,\theta}}{\inv{\kappa}(C)}=\fun{\fndef{C}}{r,\theta}$
for $\msrprb_{\txtmix}$-almost every
$\pairbk{r,\theta}$; the second statement follows from the first.
\end{proof}

On the total carrier, let \(\mathcal{T}_{\fldreal}\) be the test space of the basic model, and let \(\seq{u_j}{j\in\semigrposint}\subset\mathcal{T}_{\fldreal}\) have a dense linear span in \(\mathcal{T}_{\fldreal}\) for the norm \(\norm{j_{\sminvtemperature,0}\cdot}\). Let \(\fun{\opfocksegal_{\txtborel,0}}{u_j}\) be Borel versions of \(\fun{\opfocksegal_0}{u_j}\), let \(\mblfmlfrak{S}_{\txtborel,\sminvtemperature,0}\) be the \(\sigma\)-algebra on \(\prbqspace_{\sminvtemperature,-}\) generated by them, and let \(\mathcal{N}_{\txttot}\) be the family of the sets \(N\times\fldreal^{2}\) with \(\fun{\msrprb_{\txtnonzero,\sminvtemperature}}{N}=0\). For \(f\in\sphilbfrak{D}_{0,\sminvtemperature}\), let \(\fun{\opfocksegal_{\txtborel,0}}{\rbk{2\physham[h]}^{\frac{1}{2}}f}\) be an \(\mblfmlfrak{S}_{\txtborel,\sminvtemperature,0}\)-measurable version of \(\opfocksegal_0\rbk{\rbk{2\physham[h]}^{\frac{1}{2}}f}\), chosen as \(\fun{\opfocksegal_{\txtborel,0}}{u_j}\) for \(f=f_j=\rbk{2\physham[h]}^{-\frac{1}{2}}u_j\), and set \begin{equation}\label{eq:free-bose-carrier-time-zero-field}
\begin{aligned}
\fun{\fun{\opfocksegal}{f}}{q,r,\theta}
&=
\fun{\fun{\opfocksegal_{\txtborel,0}}{\rbk{2\physham[h]}^{\frac{1}{2}}f}}{q}+\fun{\ell_{\sminvtemperature,r,\theta}}{f},
&
\fun{\kappa_0}{q,r,\theta}
&=
\pairbk{r,\theta}.
\end{aligned}
\end{equation}

\begin{cor}[Free-Bose Gaussian components on the total carrier]\label{cor:free-bose-gaussian-components-carrier}
The restrictions of the measures \eqref{eq:free-bose-bec-component-carrier-law}
to $\mblfmlfrak{S}_{\txttot,\sminvtemperature,0}$,
with the fields and $\kappa_0$ of \eqref{eq:free-bose-carrier-time-zero-field},
the family $\seq{f_j}{j\in\semigrposint}$, and
$\mathcal{N}_{\txttot}$, satisfy Definition
\ref{def:bec-common-gaussian-component-setting} with
$\mblfmlfrak{S}_{\txttot}=\mblfmlfrak{S}_{\txttot,\sminvtemperature,0}$, and
their mixture with
$\msrprb_{\txtmix}=\msrprb_{\txtstd,\fldcmp}$ is
$\msrprb_{\txtbec,\sminvtemperature}$.
Consequently, \eqref{eq:bec-characteristic-functional-factorization} and
Proposition \ref{prop:free-bose-component-disintegration} hold on the total
carrier with $\kappa=\kappa_0$.
\end{cor}

\begin{proof}
The generating condition
\eqref{eq:bec-common-gaussian-generating-sigma-algebras} holds because
$\mblfmlfrak{S}_{\sminvtemperature,0}$
is generated by the sharp-time fields at time $0$ of the test space and the null
sets
\cite[Definition 4.4]{YoshitsuguSekine018}, and each of these fields is an
almost sure limit of linear combinations of the
$\fun{\opfocksegal_0}{u_j}$.
The characteristic functionals
\eqref{eq:bec-common-gaussian-component-characteristic-functional} follow from
\eqref{eq:free-bose-nonzero-segal-characteristic-function} and the linearity of
the sharp-time field.
\end{proof}

On the total carrier, \eqref{eq:free-bose-component-disintegration} is therefore the direct-integral decomposition of the free-Bose Euclidean field, and, by the identification after Proposition \ref{prop:free-bose-bec-component-carrier}, the probability kernel of Proposition \ref{prop:free-bose-component-disintegration} realizes the component decomposition of Section \ref{sec:resolvent-algebraic-bec}.

\subsection{Gauge transformation, symmetry breaking, and mixing}\label{sec:free-bose-gauge-mixing}

In this subsection, the gauge transformations permute the component measures and leave the BEC measure invariant, the component measures are mixing under the real-time evolution and under the spatial translations whereas the BEC measure is not, the real-time evolution is realized as a Borel action on a carrier of continuous paths, and the gauge-invariant events see only the amplitude of the order coordinate.

The real-time evolution \(\napiernum^{\imunit t\physham[h]}\) differs from the \(\sminvtemperature\)-periodic Euclidean time translations \(U_{\sminvtemperature,t}\). The space \(\sphilbfrak{D}_{0,\sminvtemperature}\) is not preserved by \(\napiernum^{\imunit t\physham[h]}\) in general: for \(d=3\), \(s=2\), \(t\neq0\), and \(\fun{f}{y}=\napiernum^{\imunit\abs{y}^{2}/\rbk{4t}}\fun{\fndef{\closedinterval{0}{1}^{3}}}{y}\), the function \(\napiernum^{\imunit t\physham[h]}f\) is, up to a phase and a constant, the Fourier transform of \(\fndef{\closedinterval{0}{1}^{3}}\) at a scaled point and is not integrable. The invariant subspace is therefore part of the following data.

\begin{defn}[Realization of test-function maps on the field space]\label{def:bec-real-time-evolution}
In the setting of Definition \ref{def:bec-common-gaussian-component-setting},
let $p\in\semigrposint$, let
$\sphilbfrak{D}_{V}\subset\sphilbfrak{D}_{0,\sminvtemperature}$ be a real
subspace containing every $f_j$, and let $V_b$,
$b\in\fldreal^{p}$, be real-linear maps of $\sphilbfrak{D}_{V}$ into itself.
A realization of $\fml{V_b}{b\in\fldreal^{p}}$ is a family of measurable maps
$V_{\txttot,b}$ of $\prbqspace_{\txttot}$ such that
$\kappa_0\circ V_{\txttot,b}=\kappa_0$, $\inv{V_{\txttot,b}}(N)$ is null for
every component measure if
$N\in\mathcal{N}_{\txttot}$, and
$\fun{\opfocksegal}{f}\circ V_{\txttot,b}=\fun{\opfocksegal}{V_bf}$ almost
surely for every component
measure and every $f\in\sphilbfrak{D}_{V}$.
\end{defn}

\begin{prop}[Gauge covariance and mixing of the BEC components]\label{prop:bec-gauge-mixing}
Assume Definitions \ref{def:bec-common-gaussian-component-setting} and
\ref{def:bec-gaussian-component-mixture} with
$\msrprb_{\txtmix}=\msrprb_{\txtstd,\fldcmp}$.
\begin{enumerate}
\item If measurable maps $\gamma_{\theta_0}$ of $\prbqspace_{\txttot}$ satisfy
$\pushoutrbk{\gamma_{\theta_0}}\msrprb_{\sminvtemperature,r,\theta}=\msrprb_{\sminvtemperature,r,\theta+\theta_0}$,
then
$\msrprb_{\txtbec,\sminvtemperature}$ is invariant under them.
\item Let $\fml{V_{\txttot,b}}{b\in\fldreal^{p}}$ be a realization of
$\fml{V_b}{b\in\fldreal^{p}}$, and assume, for
$f,g\in\sphilbfrak{D}_{V}$, $b\in\fldreal^{p}$, and every order coordinate,
\begin{equation}\label{eq:bec-mixing-hypotheses}
\begin{aligned}
\fun{\ell_{\sminvtemperature,r,\theta}}{V_bf}
&=
\fun{\ell_{\sminvtemperature,r,\theta}}{f},
&
\fun{\opform{q}_{\txtnonzero,\sminvtemperature}}{V_bf,V_bg}
&=
\fun{\opform{q}_{\txtnonzero,\sminvtemperature}}{f,g},
&
\lim_{\abs{b}\to\infty}\fun{\opform{q}_{\txtnonzero,\sminvtemperature}}{f,V_bg}
&=0.
\end{aligned}
\end{equation}
Then every $\msrprb_{\sminvtemperature,r,\theta}$ is invariant and mixing, that
is,
$$\begin{aligned}
\lim_{\abs{b}\to\infty}\sqfun{\prbexp_{\msrprb_{\sminvtemperature,r,\theta}}}{F\rbk{G\circ V_{\txttot,b}}}
&=
\sqfun{\prbexp_{\msrprb_{\sminvtemperature,r,\theta}}}{F}\sqfun{\prbexp_{\msrprb_{\sminvtemperature,r,\theta}}}{G}
\end{aligned}$$
for bounded measurable $F$ and $G$.
If an order-coordinate realization $\kappa$
(Definition \ref{def:bec-order-coordinate-realization}) satisfies
$\kappa\circ V_{\txttot,b}=\kappa$, then $\msrprb_{\txtbec,\sminvtemperature}$
is invariant and, if $\fun{\smnumberdensity_0}{\sminvtemperature}>0$, not
mixing.
\end{enumerate}
\end{prop}

\begin{proof}
Item (1) follows from \eqref{eq:bec-gaussian-component-mixture} and the rotation
invariance of $\msrprb_{\txtstd,\fldcmp}$.
For (2), the hypotheses and
\eqref{eq:bec-common-gaussian-component-characteristic-functional} show that the
image of
$\msrprb_{\sminvtemperature,r,\theta}$ has the same characteristic functional
and gives mass one to the same fiber of
$\kappa_0$, so Lemma \ref{lem:bec-common-gaussian-component-uniqueness} gives
the invariance.
For $F=\napiernum^{\imunit\fun{\opfocksegal}{f}}$ and
$G=\napiernum^{\imunit\fun{\opfocksegal}{g}}$, the expectation of
$F\rbk{G\circ V_{\txttot,b}}$ differs from the product of the expectations by
the factor
$\fnexp{-\frac{1}{2}\opreal\fun{\opform{q}_{\txtnonzero,\sminvtemperature}}{f,V_bg}}$,
which tends to $1$; the linear span of such
exponentials is dense in $\fun{\lp^2}{\msrprb_{\sminvtemperature,r,\theta}}$,
and the mixing follows.
If $\kappa\circ V_{\txttot,b}=\kappa$ and
$\fun{\smnumberdensity_0}{\sminvtemperature}>0$, then
$h=\fndef{\set{w\in\fldcmp}{\abs{w}^{2}\leq1}}$ and $F=G=h\circ\kappa$ give
$$\begin{aligned}
\sqfun{\prbexp_{\msrprb_{\txtbec,\sminvtemperature}}}{F\rbk{G\circ V_{\txttot,b}}}
=
\sqfun{\prbexp_{\msrprb_{\txtstd,\fldcmp}}}{h^{2}}
\neq
\sqfun{\prbexp_{\msrprb_{\txtstd,\fldcmp}}}{h}^{2}.
\end{aligned}$$
\end{proof}

\begin{lem}[Invariance and decay of the nonzero-mode form]\label{lem:free-bose-nonzero-form-decay}
For $f,g\in\sphilbfrak{H}_{\sminvtemperature}$, $t\in\fldreal$, and
$x\in\fldreal^d$, the operators $\napiernum^{\imunit t\physham[h]}$
and $\tau_x$ preserve $\opform{q}_{\txtnonzero,\sminvtemperature}$, and
\begin{equation}\label{eq:free-bose-nonzero-form-invariance-decay}
\begin{aligned}
\lim_{\abs{t}\to\infty}\fun{\opform{q}_{\txtnonzero,\sminvtemperature}}{f,\napiernum^{\imunit t\physham[h]}g}
&=
0
=
\lim_{\abs{x}\to\infty}\fun{\opform{q}_{\txtnonzero,\sminvtemperature}}{f,\tau_xg}.
\end{aligned}
\end{equation}
\end{lem}

\begin{proof}
The form is
$\int\fun{\kappa_{\sminvtemperature}}{\fun{\faftr{\physham[h]}}{k}}\cmpconj{\fun{\faftr{f}}{k}}\fun{\faftr{g}}{k}m\opdmsr{k}$
with the multiplier $m$ of \eqref{eq:free-bose-unimodular-multiplier-examples}.
The decay in $x$ is the Riemann--Lebesgue lemma.
For the decay in $t$, the image of the integrable function
$\fun{\kappa_{\sminvtemperature}}{\fun{\faftr{\physham[h]}}{k}}\cmpconj{\fun{\faftr{f}}{k}}\fun{\faftr{g}}{k}$
under $k\mapsto\abs{2\pi k}^{s}$ is absolutely continuous, and the
Riemann--Lebesgue lemma on $\fldreal$ applies.
\end{proof}

\begin{cor}[Gauge orbit, mixing, and spatial long-range order on the total carrier]\label{cor:free-bose-component-gauge-mixing}
In the setting of Corollary \ref{cor:free-bose-gaussian-components-carrier}, let
$\tilde{W}_U$ be the maps of
Lemma \ref{lem:free-bose-sharp-time-phase-invariance}.
\begin{enumerate}
\item For the phase $Uf=\napiernum^{\imunit\theta_0}f$, the maps
$$\begin{aligned}
\fun{\gamma_{\theta_0}}{q,r,\theta}
&=
\pairbk{\fun{\tilde{W}_U}{q},r,\theta+\theta_0}
\end{aligned}$$
satisfy
$$\begin{aligned}
\pushoutrbk{\gamma_{\theta_0}}\msrprb_{\sminvtemperature,r,\theta}
&=
\msrprb_{\sminvtemperature,r,\theta+\theta_0}
\end{aligned}$$
and transform the
fields \eqref{eq:free-bose-bec-enlarged-sharp-time-field} by
$f\mapsto\napiernum^{\imunit\theta_0}f$; the BEC measure is gauge
invariant, and the component measures with $r>0$ are not if
$\fun{\smnumberdensity_0}{\sminvtemperature}>0$.
\item The maps $\tilde{W}_{\napiernum^{\imunit t\physham[h]}}\times\id$ realize
$\fml{\napiernum^{\imunit t\physham[h]}}{t\in\fldreal}$ with
$\sphilbfrak{D}_{V}=\mathcal{T}_{\fldreal}$, and the maps
$\tilde{W}_{\tau_x}\times\id$ realize $\fml{\tau_x}{x\in\fldreal^{d}}$ with
$\sphilbfrak{D}_{V}=\sphilbfrak{D}_{0,\sminvtemperature}$; both satisfy
\eqref{eq:bec-mixing-hypotheses}.
\item Every component measure is invariant and mixing for both families, and the
BEC measure is invariant and, if
$\fun{\smnumberdensity_0}{\sminvtemperature}>0$, not mixing.
More precisely, if $\fun{\smnumberdensity_0}{\sminvtemperature}>0$ and
$\fun{\faftr{f}}{0}\neq0$, then
\begin{equation}\label{eq:free-bose-bec-long-range-order}
\begin{aligned}
\lim_{\abs{x}\to\infty}
\sqfun{\prbexp_{\msrprb_{\txtbec,\sminvtemperature}}}{
\napiernum^{\imunit\fun{\opfocksegal}{f}}
\rbk{\napiernum^{\imunit\fun{\opfocksegal}{f}}\circ\rbk{\tilde{W}_{\tau_x}\times\id}}}
&=
\napiernum^{-\frac{1}{2}\fun{\opform{q}_0}{f}}
\sqfun{\prbexp_{\msrprb_{\txtbec,\sminvtemperature}}}{\napiernum^{\imunit\fun{\opfocksegal}{f}}}^{2},
\end{aligned}
\end{equation}
with $\napiernum^{-\frac{1}{2}\fun{\opform{q}_0}{f}}<1$.
The same limit holds for
$\tilde{W}_{\napiernum^{\imunit t\physham[h]}}\times\id$ as
$\abs{t}\to\infty$ for every $f\in\sphilbfrak{D}_{0,\sminvtemperature}$ with
$\fun{\faftr{f}}{0}\neq0$.
The limit in $x$ is the spatial long-range order of the BEC measure.
\end{enumerate}
\end{cor}

\begin{proof}
For (1), Lemma \ref{lem:free-bose-sharp-time-phase-invariance} and
$\fun{\ell_{\sminvtemperature,r,\theta}}{\napiernum^{\imunit\theta_0}f}=\fun{\ell_{\sminvtemperature,r,\theta+\theta_0}}{f}$
give the
transformation of the fields, and Lemma
\ref{lem:bec-common-gaussian-component-uniqueness} gives the image measures.
For (2), the multipliers preserve $\fun{\faftr{f}}{0}$, the test space
$\mathcal{T}_{\fldreal}$ is preserved by
$\napiernum^{\imunit t\physham[h]}$, and Lemma
\ref{lem:free-bose-nonzero-form-decay} gives \eqref{eq:bec-mixing-hypotheses}.
Item (3) follows from Proposition \ref{prop:bec-gauge-mixing}; the limit
\eqref{eq:free-bose-bec-long-range-order} is
$$\begin{aligned}
\sqfun{\prbexp_{\msrprb_{\txtstd,\fldcmp}(r,\theta)}}{\napiernum^{2\imunit\fun{\ell_{\sminvtemperature,r,\theta}}{f}}}\napiernum^{-\frac{1}{2}\fun{\opform{q}_{\txtnonzero,\sminvtemperature}}{f}}
&=
\napiernum^{-\fun{\opform{q}_0}{f}-\frac{1}{2}\fun{\opform{q}_{\txtnonzero,\sminvtemperature}}{f}}
\end{aligned}$$
by \eqref{eq:free-bose-zero-mode-polar-representation}.
For the real-time composite observable, apply
Lemma \ref{lem:free-bose-sharp-time-phase-invariance} to the centered field and
keep
the constant mean unchanged; this uses only
$\napiernum^{\imunit t\physham[h]}f\in\sphilbfrak{H}_{\sminvtemperature}$, not
the
invariance of $\sphilbfrak{D}_{0,\sminvtemperature}$.
\end{proof}

The constant of the long-range order is the zero-mode form \(\fun{\opform{q}_0}{f}=2\fun{\smnumberdensity_0}{\sminvtemperature}\abs{\fun{\faftr{f}}{0}}^{2}\). For comparison, consider the canonical ensemble of \(N\) free bosons with the kinetic energy \(-\frac{1}{2}\laplacian\) in the box \(\Lambda_N=L_NU\), \(U=\closedinterval{-\frac{1}{2}}{\frac{1}{2}}^{d}\), with \(L_N^{d}=N/\bar{\smnumberdensity}\), \(d\geq3\), and periodic or diffusive boundary conditions, and let \(\varphi_1\) be the ground state of \(-\frac{1}{2}\laplacian\) on \(U\) with this boundary condition, normalized in \(\fun{\lp^2}{U}\). If \(\bar{\smnumberdensity}>\rho_{\ast}\), where \(\rho_{\ast}=\rbk{2\pi\sminvtemperature}^{-\frac{d}{2}}\fun{\riemzeta}{\frac{d}{2}}\) is the critical density for \(-\frac{1}{2}\laplacian\), then König, Vogel, and Zass proved that the kernel \(\gamma_N\) of the one-particle reduced density matrix satisfies, uniformly in \(x,y\in\Lambda_N\) as \(N\to\infty\), \[\begin{aligned}
\fun{\gamma_N}{x,y}
&=
\rbk{\bar{\smnumberdensity}-\rho_{\ast}+\fun{o}{1}}
\fun{\varphi_1}{\frac{x}{L_N}}\fun{\varphi_1}{\frac{y}{L_N}}
+\fun{\psi}{\abs{x-y}}+\fun{o}{1},
\end{aligned}\] with \(\fun{\psi}{r}\leq Cr^{2-d}\) \cite[Theorem 1.1]{KonigVogelZass001}. In the grand-canonical description here, \(\fun{\smnumberdensity_0}{\sminvtemperature}\) takes the place of \(\bar{\smnumberdensity}-\rho_{\ast}\), and the decaying cross term \(\fun{\opform{q}_{\txtnonzero,\sminvtemperature}}{f,\tau_xf}\) takes the place of the decaying term.

The maps \(\tilde{W}_{\napiernum^{\imunit t\physham[h]}}\times\id\) are defined for each \(t\) up to null sets and need not form a group. A carrier on which the real-time evolution is a Borel action is constructed from continuous paths. Let \(\prbqspace_{\conti}=\fun{\conti}{\fldreal;\fldreal}^{\semigrposint}\) with the topology of locally uniform convergence in each coordinate; it is Polish, its Borel \(\sigma\)-algebra is generated by the evaluations \(\omega\mapsto\fun{\omega_j}{a}\), \(j\in\semigrposint\), \(a\in\fldrat\), and the shifts \(\fun{\rbk{T_a\omega}_j}{b}=\fun{\omega_j}{a+b}\) form a continuous action of \(\fldreal\). Start with \(f_j=\rbk{2\physham[h]}^{-\frac{1}{2}}u_j\), \(u_j\) as before Corollary \ref{cor:free-bose-gaussian-components-carrier}. Delete the zero vectors, retain one representative of each orbit of \(\napiernum^{\imunit t\physham[h]}\), and renumber; the linear span of the retained orbits still contains the original family. Distinct pairs \(\pairbk{j,t}\) then give distinct vectors \(\napiernum^{\imunit t\physham[h]}f_j\): equality at two different times would force the nonzero continuous function \(\faftr{f_j}\) to be supported in a Lebesgue null union of level sets of \(\faftr{\physham[h]}\). Let \(\prbqspace_{\conti,\txttot}=\prbqspace_{\conti}\times\fldreal^{2}\) with \(\fun{\kappa_0}{\omega,w}=w\).

\begin{prop}[Real-time realization of the free-Bose component measures]\label{prop:free-bose-real-time-realization}
The following statements hold.
\begin{enumerate}
\item There is exactly one Borel probability measure
$\msrprb_{\conti,\sminvtemperature}$ on $\prbqspace_{\conti}$ under which the
coordinates $\fun{\omega_j}{t}$ form a centered Gaussian family with
\begin{equation}\label{eq:free-bose-real-time-covariance}
\begin{aligned}
\sqfun{\prbexp_{\msrprb_{\conti,\sminvtemperature}}}{\fun{\omega_{j}}{t}\fun{\omega_{k}}{t'}}
&=
\frac{1}{2}\opreal\fun{\opform{q}_{\txtnonzero,\sminvtemperature}}{\napiernum^{\imunit t\physham[h]}f_j,\napiernum^{\imunit t'\physham[h]}f_k},
\end{aligned}
\end{equation}
and it is invariant under every $T_t$.
\item With the component measures
$\msrprb_{\conti,\sminvtemperature,r,\theta}=\msrprb_{\conti,\sminvtemperature}\otimes\msrdirac_{\pairbk{r,\theta}}$,
there are Borel fields $\fun{\opfocksegal_{\conti}}{f}$,
$f\in\sphilbfrak{D}_{0,\sminvtemperature}$, with
$\fun{\fun{\opfocksegal_{\conti}}{\napiernum^{\imunit t\physham[h]}f_j}}{\omega,w}=\fun{\omega_j}{t}+\fun{\ell_{w}}{\napiernum^{\imunit t\physham[h]}f_j}$,
where $\ell_w=\ell_{\sminvtemperature,r,\theta}$ for
$w=\sqrt{r}\napiernum^{\imunit\theta}$, such that
Definition \ref{def:bec-common-gaussian-component-setting} holds with the family
$\set{\napiernum^{\imunit t\physham[h]}f_j}{j\in\semigrposint,\ t\in\fldrat}$
and $\mathcal{N}_{\txttot}=\emptyset$; the mixture with
$\msrprb_{\txtmix}=\msrprb_{\txtstd,\fldcmp}$ is
$\msrprb_{\conti,\sminvtemperature}\otimes\msrprb_{\txtstd,\fldcmp}$.
\item The maps $T_t\times\id$ form a Borel action of $\fldreal$ with
$\kappa_0\circ\rbk{T_t\times\id}=\kappa_0$, and they realize
$\fml{\napiernum^{\imunit t\physham[h]}}{t\in\fldreal}$ with
$\sphilbfrak{D}_{V}=\mathcal{T}_{\fldreal}$ and satisfy
\eqref{eq:bec-mixing-hypotheses}.
\end{enumerate}
\end{prop}

\begin{proof}
For (1), the right-hand side of \eqref{eq:free-bose-real-time-covariance} is a
positive semidefinite kernel, and Kolmogorov's
extension theorem gives a Gaussian family.
Since
$\fafouriertr f_j\in\fun{C_{\mathrm{c}}^{\infty}}{\fldreal^d\setminus\setone{0}}$,
the increments satisfy
$$\begin{aligned}
\sqfun{\prbexp_{\msrprb_{\conti,\sminvtemperature}}}{\abs{\fun{\omega_j}{t}-\fun{\omega_j}{t'}}^{2}}
&\leq
\frac{1}{2}\fun{\opform{q}_{\txtnonzero,\sminvtemperature}}{\physham[h]f_j}\abs{t-t'}^{2},
\end{aligned}$$
and Gaussianity gives a fourth-moment bound of order $\abs{t-t'}^{4}$; the
Kolmogorov--Chentsov theorem
\cite[Chapter I, Corollary to Theorem 4.3, p.~20]{IkedaWatanabe001} gives a
continuous
modification, whose law is
$\msrprb_{\conti,\sminvtemperature}$.
The covariance depends only on $t-t'$, which gives the invariance under $T_t$.
For (2), the real linear span of the $\napiernum^{\imunit t\physham[h]}f_j$ is
dense in $\sphilbfrak{D}_{0,\sminvtemperature}$ for
$\opform{q}_{\txtnonzero,\sminvtemperature}^{\frac{1}{2}}$ by the choice of the
$u_j$ and \cite[Lemma 5.2]{YoshitsuguSekine018}, and
the centered part of $\fun{\opfocksegal_{\conti}}{f}$ is a Borel version of the
$\fun{\lp^2}{\msrprb_{\conti,\sminvtemperature}}$-limit
of the corresponding linear combinations; adding $\fun{\ell_w}{f}$ gives the
field, and
\eqref{eq:bec-common-gaussian-component-characteristic-functional} follows from
the covariance.
The generating condition holds by the description of the Borel $\sigma$-algebra
and the continuity of the paths.
For (3), the action is continuous, and the shift identity holds pointwise on the
orbit coordinates $\fun{\omega_j}{t}$.
Almost sure linearity extends it to their real linear span, and, for
$f\in\mathcal{T}_{\fldreal}$, approximation by this span in
$\opform{q}_{\txtnonzero,\sminvtemperature}^{\frac{1}{2}}$, the invariance of
$\msrprb_{\conti,\sminvtemperature}$ under $T_t$, and the invariance of the form
under
$\napiernum^{\imunit t\physham[h]}$ give it in
$\fun{\lp^2}{\msrprb_{\conti,\sminvtemperature}}$; the mean functional vanishes
on
$\mathcal{T}_{\fldreal}$, so the identity holds under every component measure.
Lemma \ref{lem:free-bose-nonzero-form-decay} gives
\eqref{eq:bec-mixing-hypotheses}.
\end{proof}

\begin{cor}[Mixing on the real-time carrier]\label{cor:free-bose-real-time-mixing}
Every $\msrprb_{\conti,\sminvtemperature,r,\theta}$ is invariant and mixing for
$\fml{T_t\times\id}{t\in\fldreal}$, and
$\msrprb_{\conti,\sminvtemperature}\otimes\msrprb_{\txtstd,\fldcmp}$ is
invariant and, if
$\fun{\smnumberdensity_0}{\sminvtemperature}>0$, not mixing.
\end{cor}

\begin{proof}
Apply Proposition \ref{prop:bec-gauge-mixing}(2) with $\kappa=\kappa_0$.
\end{proof}

The gauge-invariant events see the order coordinate only through its amplitude. In the setting of Proposition \ref{prop:bec-gauge-mixing}(1), the gauge-invariant \(\sigma\)-algebra is \[\begin{aligned}
\mblfmlfrak{S}_{\txttot}^{\gamma}
&=
\set{A\in\mblfmlfrak{S}_{\txttot}}{\inv{\gamma_{\theta_0}}(A)=A\text{ for every }\theta_0\in\fldreal}.
\end{aligned}\]

\begin{prop}[Amplitude disintegration on gauge-invariant events]\label{prop:free-bose-amplitude-disintegration}
In the setting of Proposition \ref{prop:bec-gauge-mixing}(1), let $\kappa$ be an
order-coordinate realization (Definition
\ref{def:bec-order-coordinate-realization})
and $\kappa_{\mathrm{a}}=\abs{\kappa}^{2}$ its amplitude.
\begin{enumerate}
\item For $r\geq0$, $\theta\in\fldreal$, and bounded
$\mblfmlfrak{S}_{\txttot}^{\gamma}$-measurable $F$,
$\sqfun{\prbexp_{\msrprb_{\sminvtemperature,r,\theta}}}{F}=\sqfun{\prbexp_{\msrprb_{\sminvtemperature,r,0}}}{F}$.
\item For $A\in\mblfmlfrak{S}_{\txttot}^{\gamma}$ and Borel
$C\subset\nonneginterval$,
\begin{equation}\label{eq:free-bose-amplitude-disintegration}
\begin{aligned}
\fun{\msrprb_{\txtbec,\sminvtemperature}}{A\cap\inv{\kappa_{\mathrm{a}}}(C)}
&=
\int_{0}^{\infty}\fun{\fndef{C}}{r}\fun{\msrprb_{\sminvtemperature,r,0}}{A}\napiernum^{-r}\opdmsr{r}.
\end{aligned}
\end{equation}
\item For bounded $\mblfmlfrak{S}_{\txttot}^{\gamma}$-measurable $F$, the
conditional expectations of $F$ given
$\mblfmlgenerated{\kappa}$ and given $\mblfmlgenerated{\kappa_{\mathrm{a}}}$
coincide almost surely and equal
$\sqfun{\prbexp_{\msrprb_{\sminvtemperature,\kappa_{\mathrm{a}},0}}}{F}$.
\end{enumerate}
\end{prop}

\begin{proof}
Item (1) is
$\sqfun{\prbexp_{\msrprb_{\sminvtemperature,r,\theta}}}{F}=\sqfun{\prbexp_{\msrprb_{\sminvtemperature,r,0}}}{F\circ\gamma_{\theta}}$.
Items (2) and (3) follow from (1), Proposition
\ref{prop:free-bose-component-disintegration}, and the polar form
\eqref{eq:free-bose-radial-phase-law}.
\end{proof}

\subsection{Ergodic decomposition and the represented center}\label{sec:free-bose-ergodic-decomposition}

In this subsection, the invariant \(\sigma\)-algebra of the real-time evolution is identified with the \(\sigma\)-algebra of the order coordinate, the component disintegration is identified with the ergodic decomposition, and the invariant random variables are compared with the center of the GNS representation.

For a realization \(\fml{V_{\txttot,t}}{t\in\fldreal}\) with \(p=1\), the invariant \(\sigma\)-algebra is \[\begin{aligned}
\mblfmlfrak{S}_{\txttot}^{V_{\txttot}}
&=
\set{A\in\mblfmlfrak{S}_{\txttot}}{\inv{V_{\txttot,t}}(A)=A\text{ for all }t\in\fldreal}.
\end{aligned}\]

\begin{prop}[Invariant $\sigma$-algebra, represented center, and ergodic decomposition]\label{prop:bec-invariant-sigma-center}
Assume the hypotheses of Proposition \ref{prop:bec-gauge-mixing}(2) with $p=1$,
and let $\kappa$ be an order-coordinate
realization (Definition \ref{def:bec-order-coordinate-realization}) with
$\kappa\circ V_{\txttot,t}=\kappa$.
\begin{enumerate}
\item Every $\msrprb_{\sminvtemperature,r,\theta}$ is ergodic,
$\mblfmlgenerated{\kappa}\subset\mblfmlfrak{S}_{\txttot}^{V_{\txttot}}$, and
every $A\in\mblfmlfrak{S}_{\txttot}^{V_{\txttot}}$ agrees with some
$\inv{\kappa}(C)$ up to a
$\msrprb_{\txtbec,\sminvtemperature}$-null set.
Consequently, $q\mapsto\msrprb_{\sminvtemperature,\fun{\kappa}{q}}$ is a regular
conditional probability given
$\mblfmlfrak{S}_{\txttot}^{V_{\txttot}}$, and the bounded invariant random
variables are, up to null sets, the bounded Borel
functions of $\kappa$.
\item Let $\oarepn$ be a representation of a $\oacstar$-algebra $\oa{A}$ that is
unitarily equivalent to a direct integral over
$\pairbk{\fldreal^{2},\msrprb_{\txtstd,\fldcmp}}$ of factor representations
$\oarepn_{r,\theta}$ whose diagonal algebra is
contained in $\oadoublecommutant{\fun{\oarepn}{\oa{A}}}$, and assume that a
countable subset $\oa{A}_0\subset\oa{A}$ satisfies
$\oadoublecommutant{\fun{\oarepn_{r,\theta}}{\oa{A}_0}}=\oadoublecommutant{\fun{\oarepn_{r,\theta}}{\oa{A}}}$
almost everywhere.
Then the center of $\oadoublecommutant{\fun{\oarepn}{\oa{A}}}$ is the diagonal
algebra $\fun{\lp^\infty}{\fldreal^{2},\msrprb_{\txtstd,\fldcmp}}$.
\item If, moreover, $\pairbk{\prbqspace_{\txttot},\mblfmlfrak{S}_{\txttot}}$ is
a standard Borel space and
$\pairbk{t,q}\mapsto\fun{V_{\txttot,t}}{q}$ is a Borel action of $\fldreal$,
then the ergodic decomposition of
$\msrprb_{\txtbec,\sminvtemperature}$
\cite[Theorem 4.2, p.~208]{VSVaradarajan002} is
given by the components: every invariant
measurable map $q\mapsto\msrprb_q$ into the ergodic measures with the properties
of \cite[Theorem 4.2, p.~208]{VSVaradarajan002} satisfies
$\msrprb_q=\msrprb_{\sminvtemperature,\fun{\kappa}{q}}$ for
$\msrprb_{\txtbec,\sminvtemperature}$-almost every $q$, and
\begin{equation}\label{eq:bec-component-ergodic-decomposition}
\begin{aligned}
\fun{\msrprb_{\txtbec,\sminvtemperature}}{A}
&=
\sqfun{\prbexp_{\msrprb_{\txtstd,\fldcmp}(r,\theta)}}{\fun{\msrprb_{\sminvtemperature,r,\theta}}{A}}
=
\sqfun{\prbexp_{\msrprb_{\txtbec,\sminvtemperature}(q)}}{\fun{\msrprb_q}{A}},
\quad
A\in\mblfmlfrak{S}_{\txttot}.
\end{aligned}
\end{equation}
\end{enumerate}
\end{prop}

\begin{proof}
For (1), mixing with $F=G=\fndef{A}$ for invariant $A$ gives
$\fun{\msrprb_{\sminvtemperature,r,\theta}}{A}\in\setone{0,1}$; with
$C=\set{\pairbk{r,\theta}}{\fun{\msrprb_{\sminvtemperature,r,\theta}}{A}=1}$,
every component gives
$A\symdiff\inv{\kappa}(C)$ measure zero, and
\eqref{eq:free-bose-component-disintegration} gives the same for the mixture.
Proposition \ref{prop:free-bose-component-disintegration} and the Doob--Dynkin
lemma give the rest.

For (2), let $\oa{D}$ be the diagonal algebra and $\oa{A}_1$ the separable
$\oacstar$-algebra generated by $\oa{A}_0$.
The direct integral
$\int^{\oplus}\oadoublecommutant{\fun{\oarepn_{r,\theta}}{\oa{A}_1}}\opdmsr{\msrprb_{\txtstd,\fldcmp}}$
is the von Neumann
algebra generated by $\fun{\oarepn}{\oa{A}_1}$ and $\oa{D}$
\cite[Theorem IV.8.18, pp.~273--274, and the proof of Theorem IV.8.32, p.~284]{MasamichiTakesaki003},
which is
$\oadoublecommutant{\fun{\oarepn}{\oa{A}}}$ because
$\oa{D}\subset\oadoublecommutant{\fun{\oarepn}{\oa{A}}}$ and the fibers coincide
almost
everywhere, and its center is $\oa{D}$ because the fibers are factors
\cite[Corollary IV.8.20, pp.~274--275]{MasamichiTakesaki003}.

For (3), the components are invariant and ergodic by (1).
For each measurable set $A$, both $q\mapsto\fun{\msrprb_q}{A}$ and
$q\mapsto\fun{\msrprb_{\sminvtemperature,\fun{\kappa}{q}}}{A}$ are versions of
the
conditional expectation of $\fndef{A}$ given the invariant $\sigma$-algebra, by
\cite[Lemma 4.3, p.~204, and Theorem 4.2, p.~208]{VSVaradarajan002} and by (1),
and hence they
agree almost surely.
The $\sigma$-algebra of a standard Borel space is generated by a countable
$\pi$-system; one common exceptional null set for this system and the
$\pi$--$\lambda$ theorem give the equality of the probability measures.
\end{proof}

\begin{cor}[Invariant $\sigma$-algebra, center, and ergodic decomposition for the free Bose gas]\label{cor:free-bose-invariant-decomposition-center}
Let
$\msrprb_{\conti,\txtbec,\sminvtemperature}=\msrprb_{\conti,\sminvtemperature}\otimes\msrprb_{\txtstd,\fldcmp}$
on the real-time carrier
$\prbqspace_{\conti,\txttot}$ of Proposition
\ref{prop:free-bose-real-time-realization}.
\begin{enumerate}
\item Up to null sets of $\msrprb_{\conti,\txtbec,\sminvtemperature}$, the order
coordinate $\kappa_0$ generates the invariant $\sigma$-algebra of
$\fml{T_t\times\id}{t\in\fldreal}$, and
$$\begin{aligned}
\fun{\msrprb_{\conti,\txtbec,\sminvtemperature}}{A}
&=
\sqfun{\prbexp_{\msrprb_{\txtstd,\fldcmp}(r,\theta)}}{\fun{\msrprb_{\conti,\sminvtemperature,r,\theta}}{A}}
\end{aligned}$$
is the ergodic decomposition for this Borel action.
\item If $0<s<d$ and $\fun{\smnumberdensity_0}{\sminvtemperature}>0$, then the
hypotheses of
Proposition \ref{prop:bec-invariant-sigma-center}(2) hold for
$\oa{A}=\fun{\oaresolventalgebra}{\sphilbfrak{D}_{0,\sminvtemperature}}$,
the GNS representation of $\oastate[\psi_{\txtbec,\sminvtemperature}]$, the
component representations
$\oarepn_{\txtnonzero,\sminvtemperature,r,\theta}$, and
$\oa{A}_0=\set{\fun{\oaresolvent}{1,f}}{f\in\sphilbfrak{D}_{\ast}}$ for a
countable $\opform{q}_{\txtnonzero,\sminvtemperature}^{\frac{1}{2}}$-dense
subset $\sphilbfrak{D}_{\ast}\subset\sphilbfrak{D}_{0,\sminvtemperature}$;
the represented center is
$\fun{\lp^\infty}{\fldreal^{2},\msrprb_{\txtstd,\fldcmp}}$.
\end{enumerate}
Both the invariant $\sigma$-algebra and the center are therefore parametrized by
the order coordinate.
\end{cor}

\begin{proof}
Item (1) follows from Proposition \ref{prop:free-bose-real-time-realization},
Corollary \ref{cor:free-bose-real-time-mixing}, and
Proposition \ref{prop:bec-invariant-sigma-center}(1) and (3), since
$\prbqspace_{\conti,\txttot}$ is Polish.
For (2), Theorem \ref{thm:free-bose-resolvent-central-decomposition} gives the
direct integral, the factor property, and the
containment of the diagonal algebra.
The space $\sphilbfrak{D}_{0,\sminvtemperature}$ is separable for
$\opform{q}_{\txtnonzero,\sminvtemperature}^{\frac{1}{2}}$, since
$K_{\sminvtemperature}^{\frac{1}{2}}$ maps it isometrically into
$\fun{\lp^2}{\fldreal^d}$.
After conjugation by $U$ of Lemma \ref{lem:free-bose-araki-woods-components}(2),
the von Neumann algebra generated by
$\fun{\oarepn_{\sminvtemperature,r,\theta}}{\fun{\oaresolvent}{1,f}}$ contains
$\napiernum^{\imunit t\rbk{\fun{\phi_{\txtleft,\sminvtemperature}}{f}+\fun{\ell_{\sminvtemperature,r,\theta}}{f}}}$,
and hence the left
Weyl operators of $tf$, $t\in\fldreal$; by the density of
$\sphilbfrak{D}_{\ast}$ in $\sphilbfrak{H}_{\sminvtemperature}$ and the
strong continuity of the Weyl operators, $\oa{A}_0$ generates the factor of
Lemma \ref{lem:free-bose-araki-woods-components}(2) for every order coordinate.
\end{proof}

\subsection{Exterior conditional laws and the tail of the time-averaged field}\label{sec:free-bose-time-averaged-tail}

In this subsection, we compare the component decomposition with mixtures of Gibbs measures admitted by a common specification on a countable parameter set such as a lattice. These measures share the prescribed conditional laws given the exterior of finite regions \cite[Definition 1.23, p.~16]{HansOttoGeorgii004}, and the extreme Gibbs measures are exactly the tail-trivial Gibbs measures \cite[Theorem 7.7(a), p.~118]{HansOttoGeorgii004}. This framework assumes a countably infinite parameter set \cite[Section 1.1, p.~11]{HansOttoGeorgii004} and does not apply directly to the time-averaged field below, which is indexed by test functions on \(\fldreal^d\) and conditioned on the exterior of bounded open sets; the counterparts of these two properties are proved directly in Theorem \ref{thm:free-bose-time-averaged-exterior-tail} and Proposition \ref{prop:free-bose-time-averaged-extreme-points}. We use the field at the time-constant test functions \(F_{\rbk{2\physham[h]}^{\frac{1}{2}}f}\), which we call the time-averaged field.

Assume \(0<s<d\). For \(f\in\fun{C_{\mathrm{c}}^{\infty}}{\fldreal^d;\fldreal}\), \eqref{eq:free-bose-time-constant-covariances} gives \(F_{\rbk{2\physham[h]}^{\frac{1}{2}}f}\in\dom C_{\sminvtemperature}^{\frac{1}{2}}\), and the time-averaged field and its form are \begin{equation}\label{eq:free-bose-time-averaged-field}
\begin{aligned}
\fun{\overline{\opfocksegal}}{f}
&=
\fun{\opfocksegal}{\iota_{C_{\sminvtemperature}}F_{\rbk{2\physham[h]}^{\frac{1}{2}}f}},
&
\fun{\opform{p}}{f,g}
&=
\frac{2}{\sminvtemperature}\int_{\fldreal^d}\fun{\faftr{\physham[h]}}{k}^{-1}\opreal\rbk{\cmpconj{\fun{\faftr{f}}{k}}\fun{\faftr{g}}{k}}\opdmsr{k}.
\end{aligned}
\end{equation} The form \(\opform{p}\) is an inner product, \(\overline{\opfocksegal}\) is a centered Gaussian field with covariance \(\opform{p}/2\), and it extends to the completion \(\sphilbfrak{E}\) of \(\fun{C_{\mathrm{c}}^{\infty}}{\fldreal^d;\fldreal}\) in \(\opform{p}\). Fix Borel versions \(\fun{\overline{\opfocksegal}_{\txtborel}}{f}\) and set, on the total carrier, \begin{equation}\label{eq:free-bose-time-averaged-total-field}
\begin{aligned}
\fun{\fun{\overline{\opfocksegal}_{\txttot}}{f}}{q,r,\theta}
&=
\fun{\fun{\overline{\opfocksegal}_{\txtborel}}{f}}{q}+m_{r,\theta}\int_{\fldreal^d}\fun{f}{x}\opdmsr{x},
&
m_{r,\theta}
&=
\sqrt{2\fun{\smnumberdensity_0}{\sminvtemperature}r}\fun{\cos}{\theta},
\end{aligned}
\end{equation} so that the mean is \(\fun{\ell_{\sminvtemperature,r,\theta}}{f}\). For a bounded open set \(U\subset\fldreal^d\), let \(\mathcal{C}_U\) consist of the functions in \(\fun{C_{\mathrm{c}}^{\infty}}{\fldreal^d;\fldreal}\) whose supports are disjoint from \(U\), let \(\sphilbfrak{E}_U\) be its closure in \(\sphilbfrak{E}\) with the orthogonal projection \(e_U\), and let \begin{equation}\label{eq:free-bose-time-averaged-exterior-tail}
\begin{aligned}
\mblfmlfrak{S}_{\txtexternal,U}
&=
\mblfmlgenerated{\set{\fun{\overline{\opfocksegal}_{\txttot}}{f}}{f\in\mathcal{C}_U}},
\\ 
\mblfmlfrak{S}_{\infty}
&=
\bigcap_{n\in\semigrposint}\mblfmlfrak{S}_{\txtexternal,B_n},
\\ 
\overline{\mblfmlfrak{S}}
&=
\mblfmlgenerated{\set{\fun{\overline{\opfocksegal}_{\txttot}}{f}}{f\in\fun{C_{\mathrm{c}}^{\infty}}{\fldreal^d;\fldreal}}},
\end{aligned}
\end{equation} where \(B_n\) is the open ball of radius \(n\) centered at \(0\). Fix \(f_1\) with support in \(\set{x}{1<\abs{x}<2}\) and integral \(1\), set \(\fun{f_R}{x}=R^{-d}\fun{f_1}{R^{-1}x}\), and let \(\tilde{M}\) be \(\limsup_{j\to\infty}\fun{\overline{\opfocksegal}_{\txttot}}{f_{2^j}}\) where it is finite and \(0\) elsewhere; it is \(\mblfmlfrak{S}_{\infty}\)-measurable.

\begin{thm}[Exterior conditional laws and tail of the time-averaged field]\label{thm:free-bose-time-averaged-exterior-tail}
Assume $0<s<d$.
\begin{enumerate}
\item For every bounded open set $U$ and real-valued test function $g$, there is
an $\mblfmlfrak{S}_{\txtexternal,U}$-measurable function
$\Psi_{U,g}$ with
$\fun{\Psi_{U,g}}{q,r,\theta}=\fun{\fun{\overline{\opfocksegal}}{e_Ug}}{q}+m_{r,\theta}\int g\opdmsr{x}$
almost surely,
such that, for every $\msrprb$ among the $\msrprb_{\sminvtemperature,r,\theta}$
and $\msrprb_{\txtbec,\sminvtemperature}$,
\begin{equation}\label{eq:free-bose-time-averaged-exterior-conditional-law}
\begin{aligned}
\sqfuncond{\prbexp_{\msrprb}}{\napiernum^{\imunit\fun{\overline{\opfocksegal}_{\txttot}}{g}}}{\mblfmlfrak{S}_{\txtexternal,U}}
&=
\fnexp{\imunit\Psi_{U,g}-\frac{1}{4}\fun{\opform{p}}{\rbk{1-e_U}g,\rbk{1-e_U}g}}
\end{aligned}
\end{equation}
almost surely.
\item Every $\msrprb_{\sminvtemperature,r,\theta}$ is trivial on
$\mblfmlfrak{S}_{\infty}$.
\item The restriction of $\msrprb_{\sminvtemperature,r,\theta}$ to
$\overline{\mblfmlfrak{S}}$ depends on $\pairbk{r,\theta}$ only through
$m_{r,\theta}$.
\item $\fun{\tilde{M}}{q,r,\theta}=m_{r,\theta}$ for almost every $q$ and every
$\pairbk{r,\theta}$.
If
$\bar{\smnumberdensity}>\fun{\smnumberdensity_{\txtcritical}}{\sminvtemperature}$,
every $A\in\mblfmlfrak{S}_{\infty}$ agrees with some
$\inv{\tilde{M}}(D)$ up to a $\msrprb_{\txtbec,\sminvtemperature}$-null set, and
$\msrprb_{\txtbec,\sminvtemperature}$ is not trivial on
$\mblfmlfrak{S}_{\infty}$.
\end{enumerate}
\end{thm}

\begin{proof}
We first show $e_{B_n}v\to0$ for $v\in\sphilbfrak{E}$.
The map $f\mapsto\faftr{f}$ extends to $\sphilbfrak{E}$ with values $w_v$
satisfying
$\int\abs{w_v}^{2}\fun{\faftr{\physham[h]}}{k}^{-1}\opdmsr{k}<\infty$, and
$\int w_v\cmpconj{\faftr{\varphi}}\opdmsr{k}=0$ for
$v\in\sphilbfrak{E}_{B_n}$ and $\varphi\in\fun{C_{\mathrm{c}}^{\infty}}{B_n}$.
If $v$ lies in every $\sphilbfrak{E}_{B_n}$, then the functional
$\varphi\mapsto\int w_v\cmpconj{\faftr{\varphi}}\opdmsr{k}$, which is continuous
on the
Schwartz space by the weighted $\lp^2$ bound and the Cauchy--Schwarz inequality,
vanishes on $\fun{C_{\mathrm{c}}^{\infty}}{\fldreal^d}$ and hence on the
Schwartz
space; that is, the inverse Fourier transform of $w_v$, as a tempered
distribution, vanishes.
Hence $w_v=0$ and $v=0$, and the decreasing projections $e_{B_n}$ converge
strongly to $0$.
Next, $\int f_R\opdmsr{x}=1$, $\supp f_R\subset\set{x}{R<\abs{x}<2R}$, and
$\fun{\opform{p}}{f_R,f_R}=R^{s-d}\fun{\opform{p}}{f_1,f_1}\to0$.

For (1), choose $a_j\in\mathcal{C}_U$ with $\opform{p}$-distance at most
$2^{-j-1}$ from $e_Ug$, and correct the integral by
$b_j=a_j+\rbk{\int g\opdmsr{x}-\int a_j\opdmsr{x}}f_{R_j}$ with $R_j$ large, so
that $b_j\in\mathcal{C}_U$, $\int b_j\opdmsr{x}=\int g\opdmsr{x}$,
and $\fun{\opform{p}}{b_j-e_Ug,b_j-e_Ug}^{\frac{1}{2}}\leq2^{-j}$.
Let $\Psi_{U,g}$ equal
$\limsup_{j\to\infty}\fun{\overline{\opfocksegal}_{\txttot}}{b_j}$
where this limit superior is finite and $0$ elsewhere.
It is $\mblfmlfrak{S}_{\txtexternal,U}$-measurable and has the stated form
almost
surely by the Borel--Cantelli lemma, and the conditional
expectation of a Gaussian exponential given the $\sigma$-algebra of a closed
subspace gives
\eqref{eq:free-bose-time-averaged-exterior-conditional-law} for the centered
field and hence for every component; integration over
$\msrprb_{\txtstd,\fldcmp}$ gives it for the BEC measure.

For (2), (1) with $U=B_n$ and $e_{B_n}g\to0$ show that, for
$A\in\mblfmlfrak{S}_{\infty}$, the conditional characteristic functions
of the field given $A$ under a component are those of the component, and hence
$A$ is independent of $\overline{\mblfmlfrak{S}}$ and of
itself.
Item (3) follows from the characteristic functions.
For (4), $\fun{\overline{\opfocksegal}}{f_{2^j}}$ has variance
$2^{j\rbk{s-d}}\fun{\opform{p}}{f_1,f_1}/2$, which is summable, so
$\fun{\overline{\opfocksegal}_{\txttot}}{f_{2^j}}\to m_{r,\theta}$ almost
surely.
For a tail event $A$, items (2) and (3) show that its component probability
takes values in $\setone{0,1}$ and depends only on $m_{r,\theta}$; with
$c=2\fun{\smnumberdensity_0}{\sminvtemperature}>0$, the Borel section
$m\mapsto m/\sqrt{c}\in\fldcmp$ gives a Borel set $D\subset\fldreal$ for which
this
probability is $\fun{\fndef{D}}{m_{r,\theta}}$.
Since $\tilde{M}=m_{r,\theta}$ almost surely under each component, integration
with \eqref{eq:free-bose-bec-measure-mixture} gives
$\fun{\msrprb_{\txtbec,\sminvtemperature}}{A\symdiff\inv{\tilde{M}}(D)}=0$.
Finally, $\set{\pairbk{q,r,\theta}}{\fun{\tilde{M}}{q,r,\theta}>0}$ has BEC
probability $\frac{1}{2}$, so the BEC measure is not trivial on
$\mblfmlfrak{S}_{\infty}$.
\end{proof}

Item (1) has the form of the DLR equation with the Gaussian law of mean \(\Psi_{U,g}\) and variance \(\frac{1}{2}\fun{\opform{p}}{\rbk{1-e_U}g,\rbk{1-e_U}g}\) in place of a specification kernel. For a specification in the sense of \cite[Definition 1.23, p.~16]{HansOttoGeorgii004}, on a countably infinite parameter set, the extreme points of the convex set of its Gibbs measures are exactly the Gibbs measures trivial on the tail \(\sigma\)-algebra \cite[Theorem 7.7(a), p.~118]{HansOttoGeorgii004}. Neither this definition nor this theorem applies directly here, since the time-averaged field is indexed by test functions on \(\fldreal^d\) and the conditioning is on the exterior of bounded open sets. An analogous specification admitting the component measures and the BEC measure would require proper, consistent probability kernels from \(\mblfmlfrak{S}_{\txtexternal,U}\) to \(\overline{\mblfmlfrak{S}}\). To realize item (1), their conditional characteristic functions would have to agree with the right-hand side of \eqref{eq:free-bose-time-averaged-exterior-conditional-law} almost surely under each component measure and the BEC measure, for each \(U\) and \(g\). The versions \(\Psi_{U,g}\), chosen separately for each \(g\), have not been assembled into kernels satisfying properness and consistency at every configuration, and no such specification is constructed here. Items (2) and (4) therefore do not by themselves make the component measures extreme Gibbs measures, and we use instead the convex set defined by the exterior conditional laws of item (1).

\begin{defn}[Measures with the exterior conditional laws of the time-averaged field]\label{def:free-bose-time-averaged-exterior-law-set}
Assume $0<s<d$.
Let $\mathcal{G}_{\txtexternal}$ be the set of probability measures $\msrprb$ on
$\pairbk{\prbqspace_{\txttot,\sminvtemperature},\overline{\mblfmlfrak{S}}}$
such that
\begin{equation}\label{eq:free-bose-time-averaged-field-linearity}
\begin{aligned}
\fun{\overline{\opfocksegal}_{\txttot}}{g+ah}
&=
\fun{\overline{\opfocksegal}_{\txttot}}{g}+a\fun{\overline{\opfocksegal}_{\txttot}}{h}
\quad
\msrprb\text{-almost surely},
\quad
g,h\in\fun{C_{\mathrm{c}}^{\infty}}{\fldreal^d;\fldreal},\ a\in\fldreal,
\end{aligned}
\end{equation}
and \eqref{eq:free-bose-time-averaged-exterior-conditional-law} holds
$\msrprb$-almost surely for every bounded open set $U\subset\fldreal^d$ and
every $g\in\fun{C_{\mathrm{c}}^{\infty}}{\fldreal^d;\fldreal}$, with the
functions $\Psi_{U,g}$ of
Theorem \ref{thm:free-bose-time-averaged-exterior-tail}(1).
\end{defn}

\begin{prop}[Extreme points and tail triviality]\label{prop:free-bose-time-averaged-extreme-points}
Assume $0<s<d$.
\begin{enumerate}
\item The set $\mathcal{G}_{\txtexternal}$ is a convex subset of the real vector
space of finite signed measures on
$\pairbk{\prbqspace_{\txttot,\sminvtemperature},\overline{\mblfmlfrak{S}}}$,
and it contains the restrictions of every $\msrprb_{\sminvtemperature,r,\theta}$
and of $\msrprb_{\txtbec,\sminvtemperature}$ to $\overline{\mblfmlfrak{S}}$.
\item A measure $\msrprb\in\mathcal{G}_{\txtexternal}$ is an extreme point of
$\mathcal{G}_{\txtexternal}$ if and only if
$\fun{\msrprb}{A}\in\setone{0,1}$ for every $A\in\mblfmlfrak{S}_{\infty}$.
\item The restriction of every $\msrprb_{\sminvtemperature,r,\theta}$ to
$\overline{\mblfmlfrak{S}}$ is an extreme point of
$\mathcal{G}_{\txtexternal}$.
If
$\bar{\smnumberdensity}>\fun{\smnumberdensity_{\txtcritical}}{\sminvtemperature}$,
the restriction of $\msrprb_{\txtbec,\sminvtemperature}$ to
$\overline{\mblfmlfrak{S}}$ is not.
\end{enumerate}
\end{prop}

\begin{proof}
(1): Convexity follows from the defining identity of the conditional
expectation, since the right-hand side
of \eqref{eq:free-bose-time-averaged-exterior-conditional-law} does not depend
on the measure.
The Borel functions $\fun{\overline{\opfocksegal}_{\txtborel}}{f}$ are versions
of the values of the linear field $\overline{\opfocksegal}$, and the term
$m_{r,\theta}\int f\opdmsr{x}$ is linear in $f$; this gives the linearity under
every component, and \eqref{eq:free-bose-bec-measure-mixture} gives it under
the BEC measure.
The conditional laws are Theorem
\ref{thm:free-bose-time-averaged-exterior-tail}(1), whose integrand and
conditioning $\sigma$-algebras lie in $\overline{\mblfmlfrak{S}}$.

(2): A bounded open set $U$ lies in some $B_n$, and then
$\mathcal{C}_{B_n}\subset\mathcal{C}_U$ gives
$\mblfmlfrak{S}_{\infty}\subset\mblfmlfrak{S}_{\txtexternal,U}$.
If $A\in\mblfmlfrak{S}_{\infty}$ and $0<\fun{\msrprb}{A}<1$, the conditioned
measures
$\funcond{\msrprb}{E}{A}=\fun{\msrprb}{E\cap A}/\fun{\msrprb}{A}$ and
$\funcond{\msrprb}{\cdot}{\prbqspace_{\txttot,\sminvtemperature}\setminus A}$
therefore
satisfy \eqref{eq:free-bose-time-averaged-exterior-conditional-law}, belong to
$\mathcal{G}_{\txtexternal}$, and are distinct, and their convex combination
with weights $\fun{\msrprb}{A}$ and $1-\fun{\msrprb}{A}$ is $\msrprb$.

Conversely, let $\msrprb$ be trivial on $\mblfmlfrak{S}_{\infty}$, and let
$\msrprb=\rbk{1-t}\msrprb_1+t\msrprb_2$ with
$\msrprb_1,\msrprb_2\in\mathcal{G}_{\txtexternal}$ and $0<t<1$.
By the Radon--Nikodym theorem
\cite[Theorem 32.2, p.~422]{PatrickBillingsley002}, $\msrprb_1$ has a density
$f$ with respect to $\msrprb$ with $0\leq f\leq\rbk{1-t}^{-1}$.
Let $f_n$ be a version of
$\sqfuncond{\prbexp_{\msrprb}}{f}{\mblfmlfrak{S}_{\txtexternal,B_n}}$ with values
in $\closedinterval{0}{\rbk{1-t}^{-1}}$.
For $B\in\mblfmlfrak{S}_{\txtexternal,B_n}$ and
$g\in\fun{C_{\mathrm{c}}^{\infty}}{\fldreal^d;\fldreal}$,
\eqref{eq:free-bose-time-averaged-exterior-conditional-law} with $U=B_n$ under
$\msrprb_1$ and under $\msrprb$, whose right-hand side is bounded and
$\mblfmlfrak{S}_{\txtexternal,B_n}$-measurable, and the
$\mblfmlfrak{S}_{\txtexternal,B_n}$-measurability of $f_n$ give
$$\begin{aligned}
\sqfun{\prbexp_{\msrprb_1}}{\fndef{B}\napiernum^{\imunit\fun{\overline{\opfocksegal}_{\txttot}}{g}}}
&=
\sqfun{\prbexp_{\msrprb}}{f_n\fndef{B}\napiernum^{\imunit\fun{\overline{\opfocksegal}_{\txttot}}{g}}}.
\end{aligned}$$
By the linearity of the field, the same identity holds with finite linear
combinations of the field in the exponent.
The uniqueness of characteristic functions on $\fldreal^{k}$
\cite[Section 29, p.~382]{PatrickBillingsley002} and the $\pi$--$\lambda$
theorem then show that $f_n$ is also a density of $\msrprb_1$, so that $f=f_n$
$\msrprb$-almost surely.

The function $\limsup_{n\to\infty}f_n$ is $\mblfmlfrak{S}_{\infty}$-measurable,
because the $\sigma$-algebras $\mblfmlfrak{S}_{\txtexternal,B_n}$ decrease, and
equals $f$ almost surely.
By the triviality of $\msrprb$ on $\mblfmlfrak{S}_{\infty}$, it is almost
surely constant, and the constant is
$\sqfun{\prbexp_{\msrprb}}{f}=1$.
Hence $\msrprb_1=\msrprb$, and then $\msrprb_2=\msrprb$.

(3): This follows from (1), (2), and Theorem
\ref{thm:free-bose-time-averaged-exterior-tail}(2) and (4).
\end{proof}

Proposition \ref{prop:free-bose-time-averaged-extreme-points}(2) has the same form as the lattice result \cite[Theorem 7.7(a), p.~118]{HansOttoGeorgii004}, with the set of Gibbs measures replaced by \(\mathcal{G}_{\txtexternal}\); it is proved directly above and does not rely on that theorem. By \eqref{eq:free-bose-bec-measure-mixture}, the restriction of the BEC measure to \(\overline{\mblfmlfrak{S}}\) is the \(\msrprb_{\txtstd,\fldcmp}\)-mixture of the extreme points of Proposition \ref{prop:free-bose-time-averaged-extreme-points}(3). Whether \(\mathcal{G}_{\txtexternal}\) has other elements, and whether it is the set of Gibbs measures of a specification, is not determined here. By Theorem \ref{thm:free-bose-time-averaged-exterior-tail}(3), component measures with the same \(m_{r,\theta}\) have the same restriction to \(\overline{\mblfmlfrak{S}}\), and the time-averaged field does not distinguish all the values of the order coordinate of Corollary \ref{cor:free-bose-invariant-decomposition-center}. For \(s=2\) and \(d\geq3\), the covariance of \(\sminvtemperature^{\frac{1}{2}}\overline{\opfocksegal}\) is that of the whole-space massless Gaussian free field.

\section{Conclusion and outlook}\label{sec:conclusion}

In this section, we summarize the correspondence between the two decompositions, compare it with the Bose-loop representation, and state problems left open.

For the free Bose gas above the critical density, the BEC state on the resolvent algebra and the BEC measure on the Euclidean path space are mixtures over the same standard complex Gaussian law of the order coordinate. On the operator-algebraic side, the center of the BEC representation is the algebra of bounded measurable functions of the order coordinate; it is generated by strong limits of asymptotically central resolvents, while the represented \(\oacstar\)-algebra has trivial center. On the probabilistic side, the order coordinate generates the invariant \(\sigma\)-algebra of the real-time evolution, and the component measures are the ergodic components. The gauge group rotates the phase, and on gauge-invariant events the decomposition depends only on the amplitude (Proposition \ref{prop:free-bose-amplitude-disintegration}).

The Bose-loop representation of the free Bose gas computes partition functions and reduced density matrices, that is, expectations of gauge-invariant products with equal numbers of creation and annihilation operators; for the canonical ensemble, the long loops carry the off-diagonal long-range order \cite[Theorem 1.1 and Lemmas 1.2 and 1.3]{KonigVogelZass001}. By the gauge covariance of Proposition \ref{prop:free-bose-resolvent-gauge-clustering}(1), components with the same amplitude are indistinguishable on gauge-invariant elements of the resolvent algebra, and Proposition \ref{prop:free-bose-amplitude-disintegration} proves the corresponding statement for gauge-invariant Euclidean events. A decomposition expressed through loop configurations can therefore be expected to resolve at most the amplitude. The scaled amplitude \(\fun{\smnumberdensity_0}{\sminvtemperature}r\) has the exponential law of the limiting scaled zero-mode occupation (Proposition \ref{prop:free-zero-mode-geometric-limit}), and a possible loop counterpart of it is the density of the particles in long loops. The identification of an amplitude disintegration of the reduced density kernels by this density with the correlations of the component states constructed here is not established.

The following problems remain open.

\begin{enumerate}
\item Reconstruct the conjugate momentum observables from the Euclidean field,
so that the identification of the component
measures with the component states extends from real-valued to complex test
functions.
\item Construct a specification for the time-averaged field on a configuration
space of distributions, and determine whether the field laws of its Gibbs
measures coincide with the field laws of the measures in the convex set of
Definition \ref{def:free-bose-time-averaged-exterior-law-set}.
\item Extend the density-selected limit of Theorem
\ref{thm:free-bose-density-selected-functional-integral} from one-dimensional
characteristic functions to the convergence of the finite-dimensional
distributions of the fields at several Euclidean times.
\item Carry out the corresponding analysis for models with infrared
singularities, such as the van Hove model and the
Hubbard--phonon system \cite{YoshitsuguSekine001,YoshitsuguSekine002}, in which
a competition between zero modes and the
condensate is expected, and for which the singular Gaussian path spaces of
\cite{YoshitsuguSekine018} supply the Euclidean
framework.
\end{enumerate}

\bibliography{myref.bib}

@article{
    ArakiWoods001,
    author = {H. Araki and E. J. Woods},
    journal = "J. Math. Phys.",
    pages = "637-662",
    tag = {mathematical physics, mathematics, physics, quantum statistical physics},
    title = {{Representations of the canonical commutation relations describing a nonrelativistic infinite free Bose gas}},
    volume = "4",
    year = "1963",
}

@book{
    AsaoArai028,
    author = {Asao Arai},
    month = "7",
    pages = "1-384",
    publisher = "Kyoritsu publishing",
    tag = {mathematical physics, quantum statistical mechanics, operator algebra},
    title = {{Mathematical Principles of Quantum Statistical Mechanics (in Japanese)}},
    year = "2008",
}

@article{
    BuchholzGrundling002,
    author = {Detlev Buchholz and Hendrik Grundling},
    journal = "J. Funct. Anal.",
    month = "June",
    pages = "2725-2779",
    tag = {mathmatical physics, operator algebra},
    title = {{The Resolvent Algebra: A New Approach to Canonical Quantum Systems}},
    volume = "254",
    year = "2008",
}

@book{
    DerezinskiGerard001,
    author = {Jan Dereziński and Christian Gérard},
    pages = "688",
    publisher = "Cambridge University Press",
    tag = {mathematical physics},
    title = {{Mathematics of Quantization and Quantum Fields}},
    year = "2022",
}

@article{
    DetlevBuchholz001,
    author = {Detlev Buchholz},
    issue = "5",
    journal = "J. Funct. Anal.",
    month = "March",
    pages = "3286-3302",
    tag = {mathematical physics, operator algebra},
    title = {{The resolvent algebra: Ideals and dimension}},
    volume = "266",
    year = "2014",
}

@article{
    DetlevBuchholz002,
    author = {Detlev Buchholz},
    journal = "Commun. Math. Phys.",
    month = "May",
    pages = "949-981",
    tag = {mathematical physics, operator algebra},
    title = {{The Resolvent Algebra of Non-relativistic Bose Fields: Observables, Dynamics and States}},
    volume = "362",
    year = "2018",
}

@article{
    BahnsBuchholz001,
    author = {Dorothea Bahns and Detlev Buchholz},
    journal = "J. Math. Phys.",
    month = "April",
    pages = "041903",
    tag = {mathematical physics, operator algebra, Bose-Einstein condensation, resolvent algebra},
    title = {{Trapped bosons, thermodynamic limit and condensation: a study in the framework of resolvent algebras}},
    url = "https://doi.org/10.1063/5.0042830",
    volume = "62",
    year = "2021",
}

@article{
    DetlevBuchholz004,
    author = {Detlev Buchholz},
    journal = "J. Math. Phys.",
    month = "January",
    number = "1",
    pages = "011903",
    tag = {mathematical physics, Bose-Einstein condensation, proper condensates, resolvent algebra},
    title = {{Proper condensates}},
    url = "https://doi.org/10.1063/5.0070866",
    volume = "63",
    year = "2022",
}

@article{
    DetlevBuchholz005,
    author = {Detlev Buchholz},
    journal = "J. Math. Phys.",
    month = "August",
    number = "8",
    pages = "081905",
    tag = {mathematical physics, Bose-Einstein condensation, off-diagonal long range order, proper condensates},
    title = {{Proper condensates and long range order}},
    url = "https://doi.org/10.1063/5.0087983",
    volume = "63",
    year = "2022",
}

@book{
    GeraldFolland001,
    author = {Gerald B. Folland},
    month = "4",
    pages = "416",
    publisher = "Wiley",
    tag = {mathematics, analysis},
    title = {{Real Analysis Modern Techniques and Their Applications}},
    year = "1999",
}

@article{
    KleinLandau001,
    author = {Abel Klein and Lawrence J. Landau},
    journal = "J. Funct. Anal.",
    month = "6",
    pages = "368-428",
    tag = {mathematics, mathematical physics, operator algebra, probability},
    title = {{Stochastic Processes Associated with KMS States}},
    volume = "42",
    year = "1981",
}

@book{
    MasamichiTakesaki003,
    author = {Masamichi Takesaki},
    pages = "415",
    publisher = "Springer",
    tag = {mathematics, operator algebra},
    title = {{Theory of Operator Algebras I}},
    year = "2002",
}

@article{
    VIYukalov001,
    author = {V. I. Yukalov},
    journal = "Laser Phys.",
    month = "June",
    pages = "1145-1168",
    tag = {laser physics},
    title = {{Difference in Bose-Einstein condensation of conserved and unconserved particles}},
    url = "https://link.springer.com/article/10.1134/S1054660X12070171",
    volume = "22",
    year = "2012",
}

@article{
    VSVaradarajan002,
    author = {V. S. Varadarajan},
    doi = "10.1090/S0002-9947-1963-0159923-5",
    journal = "Trans. Amer. Math. Soc.",
    number = "2",
    pages = "191--220",
    tag = {mathematics, ergodic theory, Borel spaces},
    title = {{Groups of automorphisms of Borel spaces}},
    url = "https://www.ams.org/tran/1963-109-02/S0002-9947-1963-0159923-5/",
    volume = "109",
    year = "1963",
}

@misc{
    YoshitsuguSekine001,
    author = {Y. Sekine},
    journal = "arxiv:10082056",
    month = "8",
    pages = "1-9",
    tag = {mathematical physics, mathematics, operator theory, operator algebra, quantum statistical mechanics, ferromagnetism},
    title = {{Magnetism and infrared divergence in a Hubbard-phonon interacting system}},
    year = "2010",
}

@misc{
    YoshitsuguSekine002,
    author = {Y. Sekine},
    journal = "arxiv:13085589",
    month = "8",
    note = "arXiv:1308.5589",
    pages = "1-14",
    tag = {mathematical physics, mathematics, operator theory, operator algebra, quantum statistical mechanics, ferromagnetism},
    title = {{Phonon Bose-Einstein condensation in a Hubbard-phonon interacting system with infrared divergence}},
    year = "2013",
}

@misc{
    YoshitsuguSekine018,
    author = {Yoshitsugu Sekine},
    arxiv = "2610.34783",
    doi = "10.48550/arXiv.2610.34783",
    journal = "arXiv:2610.34783",
    month = "10",
    note = "arXiv:2610.34783",
    pages = "17",
    tag = {mathematical physics, operator algebra, Euclidean field theory, beta-Markov path spaces, Araki--Woods representation, Bose--Einstein condensation},
    title = {{Gaussian $\beta$-Markov Path Spaces without a Spectral Gap and the Araki--Woods Representations}},
    url = "https://arxiv.org/abs/2610.34783",
    year = "2026",
}

@book{
    PatrickBillingsley002,
    author = {Patrick Billingsley},
    isbn = "9780471007104",
    note = "3rd ed.",
    publisher = "Wiley",
    series = "Wiley Series in Probability and Statistics",
    tag = {mathematics, probability},
    title = {{Probability and Measure}},
    url = "https://books.google.com/books/about/Probability_and_Measure.html?id=z39jQgAACAAJ",
    volume = "245",
    year = "1995",
}

@book{
    HansOttoGeorgii004,
    author = {Hans-Otto Georgii},
    pages = "545",
    publisher = "Walter de Gruyter",
    series = "De Gruyter Studies in Mathematics",
    tag = {mathematics, probability},
    title = {{Gibbs Measures and Phase Transitions}},
    edition = "2",
    url = "https://doi.org/10.1515/9783110250329",
    volume = "9",
    year = "2011",
}

@book{
    IkedaWatanabe001,
    author = {N. Ikeda and S. Watanabe},
    isbn = "9780444861726",
    publisher = "North-Holland",
    series = "North-Holland Mathematical Library",
    tag = {mathematics, probability},
    title = {{Stochastic Differential Equations and Diffusion Processes}},
    url = "https://www.sciencedirect.com/bookseries/north-holland-mathematical-library/vol/24/suppl/C",
    volume = "24",
    year = "1981",
}

@article{
    KonigVogelZass001,
    author = {W. König and Q. Vogel and A. Zass},
    doi = "10.1214/25-AAP2213",
    journal = "Ann. Appl. Probab.",
    note = "arXiv: 2312.07481",
    number = "6",
    pages = "4030-4066",
    tag = {mathematical physics, probability, Bose gas},
    title = {{Off-Diagonal Long-Range Order for the Free Bose Gas via the Feynman--Kac Formula}},
    url = "https://doi.org/10.1214/25-AAP2213",
    volume = "35",
    year = "2025",
}

@book{
    BratteliRobinson003,
    author = {O. Bratteli and D. W Robinson},
    doi = "10.1007/978-3-662-02520-8",
    isbn = "9783540170938",
    note = "2nd edition.",
    publisher = "Springer",
    series = "Texts and Monographs in Physics",
    tag = {mathematical physics, algebraic quantum field theory},
    title = {{Operator Algebras and Quantum Statistical Mechanics 1: $C^*$- and $W^*$-Algebras, Symmetry Groups, Decomposition of States}},
    url = "https://doi.org/10.1007/978-3-662-02520-8",
    year = "1987",
}

@book{
    BratteliRobinson004,
    author = {O. Bratteli and D. W Robinson},
    doi = "10.1007/978-3-662-03444-6",
    isbn = "9783540614432",
    note = "2nd edition.",
    publisher = "Springer",
    series = "Texts and Monographs in Physics",
    tag = {mathematical physics, algebraic quantum field theory},
    title = {{Operator Algebras and Quantum Statistical Mechanics 2: Equilibrium States, Models in Quantum Statistical Mechanics}},
    url = "https://doi.org/10.1007/978-3-662-03444-6",
    year = "1997",
}

@book{
    VladimirBogachev002,
    author = {V. I. Bogachev},
    address = "Providence, RI",
    doi = "10.1090/surv/062",
    pages = "xii+433",
    publisher = "American Mathematical Society",
    series = "Mathematical Surveys and Monographs",
    tag = {probability, Gaussian measures},
    title = {{Gaussian measures}},
    url = "https://doi.org/10.1090/surv/062",
    volume = "62",
    year = "1998",
}

\end{document}